\documentclass[twocolumn,superscriptaddress,rmp]{revtex4-2}
\usepackage[T1]{fontenc}
\usepackage{times}
\usepackage[utf8]{inputenc}

\usepackage[resetlabels,labeled]{multibib}
\usepackage{enumitem}
\usepackage{amsmath}
\usepackage{amsthm}
\usepackage{amsfonts}
\usepackage{amssymb}
\usepackage{bm}
\usepackage{bbm}
\usepackage{bbold}
\usepackage{adjustbox}
\usepackage{graphicx,epsfig,color}
\usepackage[many]{tcolorbox}  
\usepackage{physics}
\usepackage{floatrow}
\usepackage{float}
\usepackage{mathtools}
\usepackage{mathrsfs}
\usepackage{latexsym}
\usepackage[colorlinks=true,linkcolor=teal,citecolor=purple]{hyperref}
\usepackage{cleveref}
\usepackage{lipsum} 
\usepackage[ruled,vlined]{algorithm2e}
\usepackage[mathscr]{eucal}
\usepackage{eucal}
\usepackage{extarrows}
\usepackage[textsize=small,colorinlistoftodos]{todonotes}
\usepackage{tikz}
\usetikzlibrary{quantikz}
\usepackage{comment}
\usepackage{array}

\definecolor{nggreen}{rgb}{0,0.5,0}
\definecolor{plbblue}{rgb}{0,0,0.7}
\definecolor{henrik}{rgb}{0.6,0,0}
\definecolor{darkgreen}{rgb}{0,.6,0}
\definecolor{darkyellow}{rgb}{.8,.6,.0}

\newcommand{\PreserveBackslash}[1]{\let\temp=\\#1\let\\=\temp}
\newcolumntype{C}[1]{>{\PreserveBackslash\centering}p{#1}}
\newcolumntype{R}[1]{>{\PreserveBackslash\raggedleft}p{#1}}
\newcolumntype{L}[1]{>{\PreserveBackslash\raggedright}p{#1}}

\newtheorem{theorem}{Theorem}[section]

\newtheorem{lemma}[theorem]{Lemma}

\newtheorem{definition}[theorem]{Definition}

\newtheorem{example}[theorem]{Example}

\usepackage{pifont}
\newcommand{\mc}[1]{\mathcal{#1}}
\newcommand{\ms}[1]{\mathsf{#1}}
\newcommand{\id}{\mathbb{1}}
\newcommand{\ot}{\otimes}
\newcommand{\AC}[2]{\xhookrightarrow[{#1}]{#2}}
\newcommand{\cmark}{{\color{darkgreen} \ding{51}}}%
\newcommand{\nmark}{{\color{darkyellow} (\ding{51})}}%
\newcommand{\xmark}{{\color{red} \ding{55}}}%

\newcommand{\cDS}{\mc D (\ms S)	}
\newcommand{\mm}{\vec \pi}
\newcommand{\schmidt}{\bm{\lambda}}
\renewcommand{\vec}[1]{\bm{#1}}
\newcommand{\eig}{\vec\lambda}
\newcommand{\dist}{d}
\newcommand{\e}{\mathrm{e}}

\newcommand{\NN}{\mathbb{N}}

\newcommand{\locc}{\mathsf{LOCC}}
\newcommand{\losr}{\mathsf{LOSR}}

\renewcommand{\proj}[1]{|#1\rangle\langle#1|}

\definecolor{main}{HTML}{5989cf}    
\definecolor{sub}{HTML}{cde4ff}     
\tcbset{
	sharp corners,
	colback = white,
	before skip = 0.2cm,    
	after skip = 0.5cm      
}   

\newtcolorbox{boxH}{
	sharpish corners, 
	colback = sub, 
	colframe = main, 
	boxrule = 0pt, 
	leftrule = 2pt, 
	enhanced,
	fuzzy shadow = {0pt}{-2pt}{-0.5pt}{0.5pt}{black!35} 
}

\setenumerate{topsep=5pt,itemsep=0ex,partopsep=1ex,parsep=1ex,leftmargin=20pt}
\setitemize{topsep=5pt,itemsep=0ex,partopsep=1ex,parsep=1ex,leftmargin=20pt}

\usepackage[normalem]{ulem}

\begin{document}
	
	\title{Catalysis in Quantum Information Theory}
	\author{Patryk Lipka-Bartosik}
	\affiliation{Department of Applied Physics, University of Geneva, 1211 Geneva, Switzerland}
	\author{Henrik Wilming}
	\affiliation{Leibniz Universit\"at Hannover, Appelstraße 2, 30167 Hannover, Germany}
	\author{Nelly H. Y. Ng}
	\affiliation{School of Physical and Mathematical Sciences, Nanyang Technological University, 639673, Singapore}
	\date{\today}
	\keywords{}
	
	\begin{abstract}
		Catalysts open up new reaction pathways that can speed up chemical reactions while not consuming the catalyst. A similar phenomenon has been discovered in quantum information science, where physical transformations become possible by utilizing a quantum degree of freedom that returns to its initial state at the end of the process. In this review, a comprehensive overview of the concept of catalysis in quantum information science is presented and its applications in various physical contexts are discussed.
	\end{abstract}
	
	\maketitle
	\tableofcontents
	\section{Introduction}
	
	Puzzles have been around since the dawn of human history, and have guided towards countless discoveries.  A good puzzle is easy to formulate but challenging to solve, while its solution can lead to deep insights. This review article is devoted to quantum catalysis, a puzzle that has propagated to various areas of quantum information theory. As a warm up, we will first consider another puzzle, known as the \emph{Towers of Hanoi}.  This brain-teaser will help us elucidate some of the fundamental concepts of catalysis, without having to deal with the complexities of quantum physics.
	
	The Towers of Hanoi is a mathematical puzzle that involves $k$ rods and $n$ disks of different sizes. At the beginning, the disks are arranged on a single rod in a decreasing order of size, with the smallest disk on the top (see Fig.~\ref{fig:hanoi}). The goal of the puzzle is to find the minimal number of moves needed to transfer the entire stack of disks from one rod to another, given that disks can only be placed on top of larger disks. 
	
	When the number of disks $n$ is greater than $1$, the puzzle requires at minimum $k = 3$ rods to be feasible, in which case the minimal number of moves needed to transfer the disks from one rod to another is $2^n-1$ \cite{petkovi_2009famous}. Removing one of the three rods makes the puzzle impossible to solve, since the set of allowed moves becomes too restrictive. On the other hand, adding a fourth rod allows to solve the puzzle in only six moves, which is less than the minimal number of moves (seven) required in the case of $n=3$ \cite{bousch2014quatrieme}. Interestingly, the minimal number of moves for more than four rods remains an open problem.
	
	We can refer to the arrangement of disks on each rod as the \emph{state} of the rod. When a new rod is introduced, the problem becomes more complex due to the increased number of potential states. Crucially, to solve the puzzle, the newly added rod must ultimately return to its initial state (i.e., remain empty). In more technical terms, adding extra degrees of freedom expands the configuration space of the problem. The key point to remember is that even though the additional degrees of freedom return to their initial states in the end, expanding the configuration space allows for solving the puzzle in a manner that was previously unattainable. This phenomenon can be seen as an instance of \emph{catalysis}. Its analogue in quantum information science is the main focus of this review.
	
	The Towers of Hanoi puzzle may serve as a simplified representation of a physical scenario. Yet, it contains several essential ingredients that we will encounter in this review. Firstly, the possible states of each rod in the puzzle are analogous to the potential states of a physical system (e.g. an atom). Secondly, the puzzle's rules correspond to limitations on how the physical system can be manipulated (e.g. due to energy conservation). Thirdly, the puzzle's objective can be viewed as a state transition problem, which involves determining whether a physical system in one state can be transformed into another (e.g. transforming a hot atom into a colder one). The rods in the puzzle represent various physical systems utilized to accomplish the task. Lastly, despite its seemingly straightforward description, finding the minimum number of moves required for $k>4$ remains an unsolved problem. This demonstrates the complexity and richness of the puzzle.
	
	The laws of physics place fundamental constraints on what is possible, much like the rules of the Towers of Hanoi puzzle. For example, conservation laws tell us that certain measurable properties of isolated physical systems remain constant as the system evolves over time. These laws are essential to our understanding of the physical world, in that they describe which processes can or cannot occur in nature. Since we cannot globally create or destroy conserved quantities, such as mass-energy, we can only change them locally by introducing additional degrees of freedom that store or inject these physical quantities. In the same way, the time-translation symmetry of physical processes means that we cannot create a state that violates time-translation invariance without accessing a system that already breaks it. Clocks are an example of such systems; They use an internal asymmetric flow of information to distinguish a preferred direction of time. 

	\begin{figure}[t]
		\includegraphics{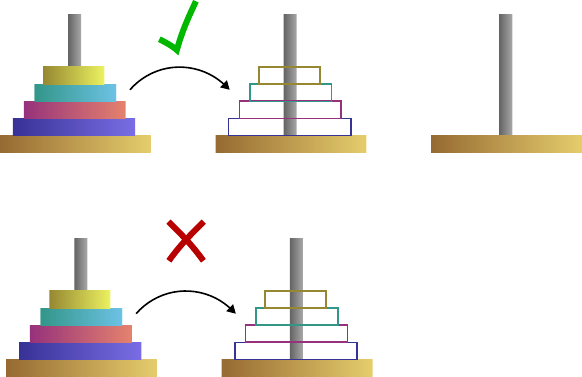}
		\caption{\textbf{Towers of Hanoi.} 
			The puzzle involves $k$ rods that can hold disks of different sizes. The objective is to move a stack of $n>1$ disks from the first rod to the second rod, one disk at a time, without placing a larger disk on top of a smaller disk. The puzzle is solvable when $k=3$, but it cannot be solved when $k=2$.}
		\label{fig:hanoi}
	\end{figure}
	
	The situation becomes more intricate when quantum theory is considered as the fundamental description of nature. When all the involved degrees of freedom are taken into account, the overall dynamics is ultimately unitary. This, combined with conversation laws, imposes non-trivial constraints that are not immediately evident from conservation laws alone. For instance, consider a probabilistic mixture of spin-up and spin-down of a spin-$1/2$ particle in a magnetic field. Although this mixture has the same average energy as a superposition of the two states, we cannot transform the mixture into the superposition using unitary transformations. This is true regardless of the fact that both states have the same expected energy, and whether we have access to a clock or not (see also Ref. \cite{Bartlett2006a}).
	
	To give another example, the relativistic principle of locality states that we cannot instantaneously affect space-like separated systems. Entanglement can be described as the property between such systems that cannot be created using only local actions and classical communication. To overcome the limitations imposed by the relativistic locality, additional resources (e.g. pre-shared entanglement) are required. Yet, even if a quantum state $\psi$ contains more entanglement than $\phi$, it does not necessarily mean that $\psi$ can be converted into $\phi$ using local operations and classical communication. In this case, the interplay between operational restrictions gives rise to a rich mathematical structure which requires more refined ways of quantifying entanglement, see Ref. \cite{Horodecki2009}.

	Using auxiliary degrees of freedom to lift restrictions imposed by physical laws often comes at a cost: Such degrees of freedom typically lose their ability to lift restrictions, as  seen in the previous examples. For instance, a clock used to break time-reversal symmetry ideally produces a series of evenly distributed ticks. However, this process has an intrinsic back-action on the clock, driving it closer to equilibrium with each tick and gradually decreasing its time-keeping potential. To counteract this mechanism, clocks must continuously consume non-equilibrium resources, see e.g. \cite{cao_free-energy_2015,Barato2016,Erker2017,Milburn2020}.
Similarly, the use of pre-shared entanglement in the form of an ancillary quantum state to entangle two distant physical systems results in the loss of entanglement of the ancillary state itself \cite{lo_concentrating_2001}.
	
	It is surprising that some restrictions imposed by unitarity of quantum mechanics can often be lifted using auxiliary degrees of freedom that \emph{do not} degrade, i.e. their quantum state remains identical before and after the process. Nevertheless, their very presence allows for physical transformations that would otherwise be impossible to accomplish. 
	
	This phenomenon has been dubbed \emph{quantum catalysis}, and demonstrates that the very presence of quantum states opens up new and interesting possibilities. The name ``quantum catalysis'' is derived from the analogy to chemistry, where catalysts allow for chemical reactions at higher rates by opening up new dynamical pathways with lower activation energies. In a similar manner, quantum catalysts open up new dynamical pathways in the Hilbert space, connecting quantum states that would otherwise be dynamically disconnected due to physical constraints. Understanding when such alternative connections can be formed not only helps us better understand the fundamental limitations of quantum physics, but also provides protocols exploiting quantum resources more efficiently.
	
	To this day, several reviews that discuss quantum catalysis in brevity have been published.  
	Ref. \cite{Horodecki2009} mentions the usage of catalysis in the context of entanglement transformations. Refs. \cite{goold_role_2016,Vinjanampathy_2016} and \cite{lostaglio_introductory_2019} discuss some aspects of catalysis from a thermodynamic point of view. 
	Refs. \cite{gour_resource_2015} and \cite{chitambar_quantum_2019} briefly discuss the subject for generic resource theories (see Section ~\ref{sec:rt}). 
	Importantly, what all of these works have in common is that they discuss quantum catalysis only as a supporting concept, i.e. an extension of existing protocols like resource conversion or work extraction. A similar treatment of quantum catalysis is manifested in numerous prominent works that use catalytic effects as a central proof technique, see e..g \cite{groisman_quantum_2005, bose_purification_1999,berta_quantum_2011,bartlett_reference_2007} and \cite{bennett_quantum_2014}. This approach, unfortunately, does not emphasize the role of catalysis as a stand-alone concept.
	
	This review article gathers our current understanding of quantum catalysis and its applications across various fronts of quantum physics. During the preparation of this review, another review on catalysis in the context of quantum resource theories appeared as a preprint \cite{Datta2022}. We invite the reader to consider it as a complementary reading.

	{\bfseries Structure of the review.} In the remaining part of this section we establish basic notation used throughout the review. Section~\ref{sec:concept} introduces the basic concept of catalysis as understood in this review in a general manner. This is followed by a discussion of the concept of resource theories and a summary of basic mathematical tools used throughout the review.
	In Section~\ref{subsec:types}, we provide a detailed introduction to the various types of catalysis that can arise by relaxing constraints discussed in Section~\ref{sec:concept} in different ways. We also offer an example to illustrate the different types in detail.
	Subsequently, Section~\ref{sec:constructions} explores the known ways to mathematically construct catalyst states. The core of the review is Section~\ref{sec:applications}, which collects and discusses the applications of catalysis in various physical settings. Section~\ref{sec:further-topics} collects additional applications and discussions that do not naturally fit into the previous sections. However, this does not imply that we consider them less important or less interesting.
	Open problems and avenues for future work are mentioned throughout the review. Due to their diverse nature it would be difficult to collect them in a conclusions section and we have intentionally refrained from doing so.

	\subsection{Quantum states and channels}\label{sec:preliminaries}
	This short section sets up the notation and prerequisites used throughout the review.
	We start by introducing some basic terminology used in quantum mechanics. Physical systems (or systems for short) will be denoted by Latin letters, e.g. $\ms{A, B, S}$. Sometimes describing the problem may require using multiple such systems, in which case the systems will be addressed by additional subscripts, such as $\ms S_1$ or $\ms S_2$. 
	To every physical system we associate a complex Hilbert space denoted by $\mc H_{\ms S}$ for a system $\ms S$, with $|\ms S| := \mathrm{dim}(\mathcal{H}_{\ms S})$ being its dimension. 
	For most part of the review it will be sufficient to consider finite-dimensional Hilbert spaces. 
	The set of (bounded) linear operators acting on $\mc H$ will be denoted with $\mc L(\mc H)$. 
	An operator $A \in \mc L(\mc H)$ is called positive semi-definite if it is self-adjoint and satisfies $ \bra{x} A \ket{x} \geq 0$ for all $\ket{x} \in \mc H$. 
	We will denote by ``$\geq$'' the L\"{o}wner partial order, i.e., for two linear operators $X$ and $Y$ the relation $X \geq Y$ means that $X - Y$ is positive semi-definite. 
	
	In quantum mechanics the possible states of a system $\ms S$ are described using density operators, i.e., positive semi-definite operators acting on $\mc H_{\ms S}$ with unit trace. We collect such operators in a set
	\begin{align}
		\mc D(\ms S) := \{\rho \in \mc L(\mc H_{\ms S})\ | \ \rho\geq 0,\  \tr[\rho]=1\},
	\end{align}
	To highlight that a density operator $\rho$ corresponds to a specific system $\ms S$ we will often write $\rho_{\ms S}$. The evolution (or dynamics) of quantum systems is fundamentally unitary if all involved degrees of freedom are incorporated. This means that a quantum system $\ms S$ prepared in a state $\rho_{\ms S}$ evolves as $\rho_{\ms S} \rightarrow U \rho_{\ms S} U^{\dagger}$,
	where $U$ satisfies $UU^\dagger = U^\dagger U = \id$. The effective unitary $U$ represents discrete-time dynamics obtained by integrating the Schrödinger equation up to a fixed time ($\hbar$ is set to 1).
	
	Often the system of interest $\ms S$ interacts with other systems, e.g. a thermal environment or a measurement apparatus. In such cases, 
	the effective dynamics on the system of interest may no longer be unitary.
	To model this, one introduces an environment which encompasses all of the auxiliary degrees of freedom. 
	Crucially, due to their role in the global evolution, in this review we differentiate two types of environments: an ordinary environment $\ms{E}$ and a catalytic environment (or catalyst for short) $\ms{C}$. 
	The environment $\ms{E}$ corresponds to all of the degrees of freedom that cannot be accessed or are seen as irrelevant. 
	This is usually the case, for example, for thermal environments. 
	On the other hand, the environment $\ms{C}$ captures all of the degrees of freedom that cannot, or should not, experience any irreversible backaction from the system. 
	In the next sections, we demonstrate that this distinction is well-motivated both from a physical and mathematical perspective.
	It is also often more convenient to consider the effective dynamics on $\ms S\ms C$ alone, instead of unitary dynamics on $\ms S\ms C\ms E$. 
	If the initial state of $\ms E$ is uncorrelated to $\ms{SC}$, then the resulting dynamics take the form of a \emph{quantum channel} $\mathcal{E}:D(\ms{SC})\rightarrow D(\ms{SC})$, which is a completely positive trace-preserving (CPTP) linear map.
	In most cases, when the space on which the channel acts is clear from the context, we avoid writing explicitly the domain and image of the channel. 
	
	Conversely, any quantum channel acting on $\ms{SC}$ may be realized by a suitable choice of environment state $\rho_{\ms E}$ and unitary $U$ \cite{Stinespring1955}.
	Given a quantum channel $ \mc E$, we refer to $(\rho_{\ms E},U)$ as a \emph{dilation} of $\mc E$. The channel is obtained by averaging over all the degrees of freedom of the environment $\ms{E}$. Mathematically, this is obtained by applying the partial trace $\tr_{\ms E}(\cdot)$ with respect to $\ms E$, which is a quantum channel $\tr_{\ms E}: \mc{D}(\ms{SCE}) \rightarrow \mc{D}(\ms{SC})$. The dilation $(\rho_{\ms E},U)$ is said to implement a quantum channel $\mc E$ when
	\begin{align}
		\mc E[\rho_{\ms{SC}}] = \tr_{\ms E}[U (\rho_{\ms{SC}}\otimes\rho_{\ms E}) U^\dagger],
	\end{align}
	holds for every $\rho_{\ms{SC}} \in \mc{D}(\ms{SC})$.  
	We also call $\mc E$ the quantum channel \emph{induced by} $(\rho_{\ms E},U)$.	
	Importantly, a quantum channel may be physically realized in different ways, having multiple dilations inducing the same channel locally. 
	However, if we restrict to pure states $\rho_{\ms E}$ there always exists a minimal Stinespring dilation in which $\ms E$ has minimal dimension. 
	Any other Stinespring dilation with pure state on a system $\ms E'$ is related to this one via an isometry from $\ms E$ to $\ms E'$.  
	
	\section{The concept of catalysis}
	\label{sec:concept}
	
	\noindent 
	In this section we formalize the concept of catalysis in a general manner. 
	We then describe the paradigm of resource theories, and argue that it provides a convenient set of tools allowing for a systematic study of catalytic effects. 
	Finally, we summarize the most relevant mathematical tools and techniques that will be used throughout this review. 
	We emphasize that this section focuses on generic features of catalysis which are valid irrespectively of the particular physical setting.  
	
	We begin by describing a simple motivating example from the early times of quantum information science, to illuminate some essential features of catalysis. 
	This example was first proposed theoretically by \cite{Phoenix_1993} and \cite{Cirac_1994} and the first experimental demonstration was reported in \cite{Hagley_1997}.  
	Consider an optical cavity with a field mode $\ms F$ in resonance with the transition frequency of \emph{two} two-level systems, such as two energy eigenstates $\ket{\downarrow},\ket{\uparrow}$ of two identical atoms $\ms A$ and $\ms B$.  
	We assume that we are restricted to performing the following actions: initializing the cavity in the vacuum state $\ket{0}_{\ms F}$, initializing the atoms in one of their two eigenstates, and turning on an energy-preserving interaction Hamiltonian between atom ${\ms A} ({\ms B})$ and cavity ${\ms F}$ for a chosen amount of time $t$.

	We will now observe how one can make use of the cavity to prepare a maximally entangled Bell state between atoms $\ms A$ and $\ms B$, while returning the cavity to its initial vacuum state. 
	First, prepare atom $\ms A$ in the excited state $\ket{\uparrow}_{\ms A}$ and let it interact with the cavity for time $t_{\ms A}$. 
	If the interaction time is short enough, we can assume that no dissipation occurs and the interaction is described by the Jaynes-Cummings Hamiltonian, which has the property that in the interaction picture $\ket{0}_{\ms F}\ket{\ms \downarrow}_{\ms A}$ is invariant while the states $\ket{1}_{\ms F}\ket{\downarrow}_{\ms A}$ and $\ket{0}_{\ms F}\ket{\uparrow}_{\ms A}$ experience Rabi-oscillations. 
	Here, $\ket{1}_{\ms F}$ is the state of the cavity with one photon.  
	If we choose $t_{\ms A}$ to correspond to quarter of a Rabi-oscillation, the state of $\ms F\ms A$ after the interaction is (in the interaction picture)
	\begin{align}
		U_{\ms {FA}}(t_{\ms A})\ket{0}_{\ms F}\ket{\ms \uparrow}_{\ms A} = \frac{1}{\sqrt 2}\left(\ket{0}_{\ms F}\ket{\uparrow}_{\ms A} - \ket{1}_{\ms F}\ket{\downarrow}_{\ms A}\right).
	\end{align}
	Now prepare the second atom $\ms B$ in the ground state $\ket{\downarrow}_{\ms B}$ and let it interact with the cavity for a time $t_{\ms B}$ such that $\ket{1}_{\ms F}\ket{\downarrow}_{\ms B}$ experiences half a Rabi-oscillation so that the states $\ket{1}_{\ms F}\ket{\downarrow}_{\ms B} \leftrightarrow \ket{0}_{\ms F}\ket{\uparrow}_{\ms B}$ are interchanged (up to a phase). 
	Since $\ket{0}_{\ms F}\ket{\downarrow}_{\ms B}$ is invariant under the dynamics, the final state of $\ms F \ms A \ms B$ is (up to a global phase)	
	\begin{align}
		\ket{0}_{\ms F}\otimes \frac{1}{\sqrt 2}\left(\ket{\uparrow}_{\ms A}\ket{\downarrow}_{\ms B} - \ket{\downarrow}_{\ms A}\ket{\uparrow}_{\ms B}\right).
	\end{align}
	The systems	$\ms A\ms B$ end up in a Bell state while $\ms F$ returns to its initial state, and can be used to entangle further pairs of atoms. 
	In the words of \cite{Hagley_1997}: "The field, which starts and ends up in vacuum and remains at the end of the process decorrelated from the atoms, acts as a `catalyst' for the atomic entanglement."
	Similar techniques can be used to implement a two-qubit \texttt{CNOT} gate between trapped ions, where the role of the cavity is now played by the center-of-mass mode of the trapped ions. This is the basis for trapped-ion quantum computers as envisioned by \cite{Cirac_1995}. At the same time it is clear that without the cavity, we cannot achieve this.  
	
	In comparison with the Tower of Hanoi puzzle from the introduction, the limited actions correspond to the allowed moves in the game and the empty cavity corresponds to an empty rod.
	
	The above simple example demonstrates that catalytic effects can be observed in generic quantum systems. While it is only a specific instance of catalysis, it contains some of its essential features. 
	For one, catalysis always requires at least two systems: a system of interest $\ms S$ and a catalyst $\ms C$. 
	In our example above, the field corresponds to catalyst $\ms C$ and the atoms to the system $\ms S$.
	Given a quantum system $\ms S$, we introduce a family of quantum channels $\mc O$ that describes all physical operations that can be implemented under the constraints of a particular physical setting.
	In the above example, the set $\mc{O}$ corresponds to all unitary evolutions that can be generated by turning-on the interaction between the cavity and one of the atoms for some time.
	More generally, any physical operation in $\mc O$ will be described by some quantum channel $\mc F \in \mc{O}$. 
	Limitations on $\mc O$ can originate in many ways, e.g.
	\begin{itemize}
		\item As a consequence of conservation laws,
		\item Due to limited resources, (e.g. finite power supply),
		\item As a result of finite control accuracy or complexity, (e.g. limited time of evolution),
		\item Due to the causality constraints in special relativity.
	\end{itemize}
	The limitations may be of fundamental nature or arise from practical limitations. While most works on catalysis in quantum information theory focus on fundamental limitations, throughout the review we don't distinguish between the two.
	Now imagine that our system $\ms S$ is prepared in a state $\rho_{\ms S}$, and we would like to bring it to a target state $\sigma_{\ms S}$:
	\begin{align}\label{eq:statetransition}
		\rho_{\ms S}\longrightarrow \sigma_{\ms S}. 
	\end{align}
	Then it may happen that, with the operations we can implement as described by the set $\mc O$, we are unable to perform the transformation. This means that there is no quantum channel $\mc F\in \mc O$ such that $\mc F[\rho_{\ms S}] = \sigma_{\ms S}$. 
	However, Eq.~\eqref{eq:statetransition} may become possible if we can act on additional degrees of freedom, represented by a quantum system $\ms C$. 
	The central idea of catalysis is that this type of activation may be possible even without disturbing the (quantum) state of subsystem $\ms C$, and hence one can \emph{reuse} $\ms C$ to ``catalyze'' the same state transition on another system $\ms S'$ prepared in the same state $\rho$. 
	In the ideal case, given a set of operations $\mc O$, we say that a state transition is achievable with a catalyst, and denote this by $ \rho_{\ms S}\AC{}~\sigma_{\ms S} $, 
	if there exists a state $ \omega_{\ms C} $ on $\ms C$ and a quantum channel $ \mc F\in\mc O $ such that 
	$	\mc F[\rho_{\ms S}\otimes\omega_{\ms C}] = \sigma_{\ms S}\otimes \omega_{\ms C}$.
	The central question is then:  

\begin{quote}
\centering{\textit{Which state transitions are made possible by catalysis?}}
\end{quote}	
	In approaching this question, we begin by presenting a simple Lemma which serves as a guideline for such considerations. 
	Imagine that we want to implement a quantum channel $\mc F$ using a unitary $U$ and auxiliary degrees of freedom $\ms E$ prepared in some state $\rho_{\ms E}$, as described in Section~\ref{sec:preliminaries}.  
	\begin{lemma}[Basic lemma]\label{lem:fundamental_nogo}
		Consider a finite dimensional Hilbert space $\mc H_{\ms{SEC}}=\mc H_{\ms{SE}}\otimes \mc H_{\ms C}$. 
		Let $U$ be a unitary on $\mc H_{\ms{SEC}}$, $\rho_{\ms{SE}} = \rho_{\ms S}\otimes\rho_{\ms E}\in \mc D(\ms{SE})$ a density operator on $\ms{SE}$ and $\omega_{\ms C}\in \mc D(\ms C)$ a density operator on $\ms C$. 
		If 
		\begin{align}\label{eq:fundamental_lemma}
			U(\rho_{\ms S}\ot \rho_{\ms E} \otimes \omega_{\ms C})U^\dagger = \sigma_{\ms{SE}}\otimes \omega_{\ms C},
		\end{align}
		for some density operator $\sigma_{\ms{SE}}\in \mc D(\ms{SE})$, then there exists a unitary operator $V$ on $\mc H_{\ms{SE}}$ such that $\sigma_{\ms{SE}} = V \rho_{\ms{SE}} V^\dagger$.
	\end{lemma}
	\begin{proof}
		Due to the tensor product structure of the output state, raising Eq.~\eqref{eq:fundamental_lemma} to any power $k\geq 0$ and taking the trace yields $\tr[\rho_{\ms{SE}}^k]=\tr[{\sigma}^k_{\ms{SE}}]$. This is only possible for two matrices if their eigenvalues (including multiplicities) coincide. Hence $ \rho_{\ms{SE}} $ and $ \sigma_{\ms{SE}}$  are related by a unitary transformation. 
	\end{proof}
	
	\begin{figure}[t]
		\includegraphics[width=\textwidth]{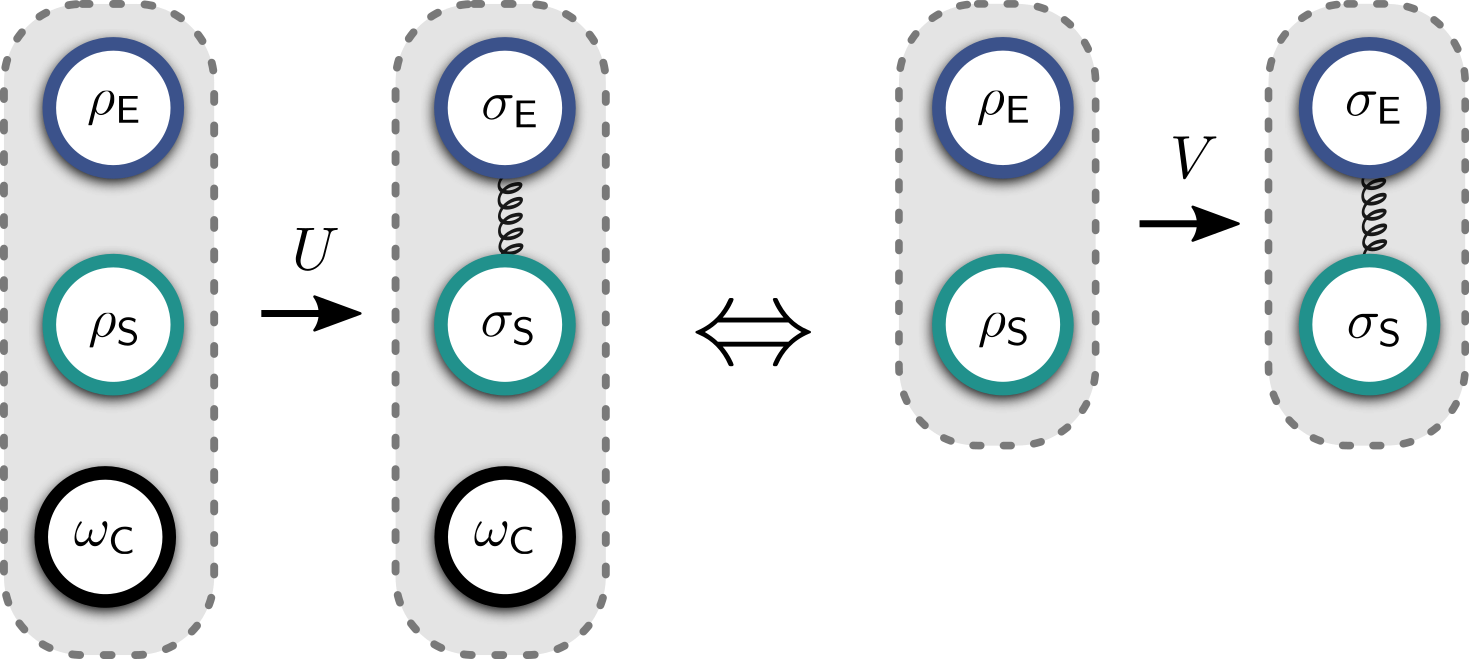}
		\caption{\textbf{Illustration of Lemma~\ref{lem:fundamental_nogo}}. Unitary $U$ uses $\ms E$ to implement a quantum channel $\mc F$ on $\ms{SC}$. If $\ms C$ is catalytic and remains uncorrelated to $\ms{ES}$, then the same state transition on $\ms S$ can be realized using a unitary $V$ acting only on $\ms{ES}$, implementing a quantum channel $\mc E$ on $\ms S$. The channels fulfill $\tr_C\circ \mc F[\rho_{\ms S}\otimes \omega_{\ms C}] = \mc E[\rho_{\ms S}] = \sigma_{\ms S}$.  }
		\label{fig:fundamental-lemma}
	\end{figure}
	
	Lemma \ref{lem:fundamental_nogo} is illustrated in Fig.~\ref{fig:fundamental-lemma}. It tells us that if we use a Stinespring dilation $(U, \rho_{\ms E})$ to implement $\mc F$ on $\ms{SC}$ such that the state on $\ms C$ does not change and remains uncorrelated to both $\ms S$ and $\ms E$, 
	then there exists a unitary operator $V$ on $\ms{SE}$ alone which leads to the same state transition on $\ms{SE}$ (and therefore on $\ms S$). In other words, the state transition $\rho_{\ms S}\rightarrow \sigma_{\ms S}$ can already be implemented using $\ms E$, and hence there is no need to use an additional catalyst $\ms{C}$. 	
	From this, we may observe three essential ways in which one can circumvent the assumptions of the lemma (and therefore enable catalysis):
	
	\begin{enumerate}[label=({\arabic*})]
		\item \label{cond:exactness} Catalysis is not exact, which means that
		\begin{align}
			\Tr_{\ms S}[\mc F[\rho_{\ms S}\otimes \omega_{\ms C}]]\neq \omega_{\ms C }.
		\end{align}		
		Therefore the state on $\ms C$ changes at least a little bit.
		\item \label{cond:correlation} The final state on $\ms C$ is correlated with $\ms{SE}$, that is
		\begin{align}
			U(\rho_{\ms S}\ot \rho_{\ms E} \otimes \omega_{\ms C})U^\dagger = \sigma_{\ms{SEC}} \neq \sigma_{\ms{SE}}\otimes \omega_{\ms C},
		\end{align}		
		where $\Tr_{\ms{SE}}\left[\sigma_{\ms{SEC}}\right] = \omega_{\ms{C}}$	. In particular, $\omega_{\ms{C}}$ is not pure.
		\item \label{cond:restrictions} The set of implementable operations $\mc O$ is sufficiently restricted. This means that, while the unitary $V$ in Lemma \ref{lem:fundamental_nogo} exists, the quantum channel it induces on ${\ms S}$,
		\begin{align}
			\mc E[\rho_{\ms S}]= \tr_{\ms E}[V \rho_{\ms S}\otimes\rho_{\ms E}V^\dagger],
		\end{align}
		is not an implementable operation, i.e. $\mc E\notin \mc O$.
	\end{enumerate}
	
	Different ways of relaxing Eq.~(\ref{eq:fundamental_lemma}), as captured by items $(1)$ and $(2)$ above, lead to different \emph{types of catalysis} (see Section \ref{subsec:types}). 	
	On the other hand, as captured by item $(3)$, different physical settings may lead to distinct sets of implementable operations $\mc O$. This, in turn, generally leads to different answers concerning the central question of catalysis. 
	In recent years, such operational restrictions have been formalized under the framework of \emph{resource theories}.
	In Sec.~\ref{sec:rt} we briefly discuss the main concepts of resources theories and describe three paradigmatic examples. 
	
	Before we close this section, let us discuss the \emph{reusability} of catalysts and how this is related to correlations catalysts establish with other systems. 
	As mentioned earlier, a catalyst $\omega_{\ms C}$ that catalyzes a state transition $\rho_{\ms S}\longrightarrow \sigma_{\ms S}$ via $\mc{F}[\rho_{\ms{S}} \ot \omega_{\ms{C}}] = \sigma_{\ms{S}} \ot \omega_{\ms{C}}$ may be reused to catalyze a further such state transition on a \emph{fresh copy} of the same state, in other words, on a \emph{different system} $\ms S'$ such that $\rho_{\ms S'}=\rho_{\ms S}$, and ${\ms S'}$ is uncorrelated to ${\ms C}$ prior to its interaction with $\ms C$.
	This is because the catalyst, by definition, is only guaranteed to act as a catalyst for a specific initial state $\rho$ and a specific operation $\mc F$. In other words, if we change the initial state from $\rho_{\ms{S}}$ to $\sigma_{\ms{S}}$, the catalyst will not, in general, remain catalytic, 
	\begin{align}
		\tr_{\ms S}\left[\mc F[\sigma_{\ms S}\otimes\omega_{\ms C}]\right] \neq \omega_C\quad\text{if}\quad \sigma_{\ms S}\neq \rho_{\ms S}.
	\end{align}
	In this sense a catalyst is always fine-tuned. However, a small error on the preparation only leads to a small error on the final states of $\ms S$ and $\ms C$. This will be discussed in more detail later.
	Since the main purpose of the catalyst is to \emph{change} the state of $\ms S$, there is not much sense in discussing the reusability of the catalyst on the same system. 
	Even so, reusing the catalyst on a different system also has its own drawbacks: As shown by Lemma~\ref{lem:fundamental_nogo}
	the catalyst can become correlated to either $\ms S$ or $\ms E$ in general.
	If the catalyst is now reused on $\ms S'$, these correlations, in general, will spread to $\ms{S}'$. Whether such an uncontrollable spread of correlations is problematic or not, in general depends on context. Furthermore, if ${\ms S}$ and ${\ms S'}$ were already initially correlated, then after the action of $\mc{F}$ on ${\ms {SC}}$, the resulting state $\hat\rho_{\ms {S'C}}$ may be correlated and therefore $	\tr_{\ms S}\left[\mc F[\hat\rho_{\ms {S'C}}]\right] \neq \omega_C$ even if $\hat\rho_{\ms{S'}} = \rho_{\ms{S'}}$ and $\hat\rho_{\ms{C}}=\omega_{\ms{C}}$.
	
	Despite the above discussion, one can easily imagine situations in which it is reasonable to demand that the catalyst remains invariant also for different states on $\ms{S}$. 
	For that, see Sec.~\ref{subsec:state-independent} where we discuss the closely related notion of \emph{state-independent catalysis}.

	\subsection{Resource theories}\label{sec:rt}
	We have seen above that catalysis is strongly related to the set of operations which obey certain constraints, that can be implemented on a physical system. 
	Such allowed operations can be formalized using the concept of a \emph{resource theory}, which we  introduce in this section, first in an abstract way and then illustrate them using three well-studied examples that will also be important in the remainder of this review. 
	An in-depth review on resource theories can be found at \cite{Chitambar2016}.
	
	A resource theory partitions the space of states and operations on physical systems into those that are either ``easy'' or ``hard'' to prepare or implement under given physical constraints. 
	The ``easy'' states and operations are typically called free. 
	We continue to use the symbol $\mathcal{O}$ to denote free operations, and use the symbol $\mc S$ for the set of free states. 
	A resource theory is specified by such a tuple $\mc R = (\mc S,\mc O)$. 
	For a fixed system $\ms S$ we denote by $\mc S(\ms S)= \mc{S} \cap \mc D(\ms S)$ the set of free states corresponding to that system. 
	One usually assumes the consistency condition that free operations always transform free states into free states, i.e
	\begin{align}
		\rho_{\ms S}\in \mc S,\ \mc F_{\ms S} \in \mc O\quad\Rightarrow\quad \mc F_{\ms S}[\rho_{\ms S}] \in \mc S.
	\end{align}
	Any state that is not a free state is interpreted as \emph{resourceful}, because it cannot be created using only operations from the set $\mc O$ acting on states in $\mc S$.  On the other hand, resourceful states can be used to implement non-free operations: Given a resourceful state  $\sigma_{\ms A}$ on an auxiliary system $\ms A$ and a free operation $\mc F_{\ms{SA}}$ on $\ms{SA}$, the operation defined by
	\begin{align}
		\mc E_{\ms S}[\rho_{\ms S}] = \tr_{\ms A}\left[\mc F_{\ms{SA}}[\rho_{\ms S}\otimes\sigma_{\ms A}]\right]
	\end{align}
	is typically not free, that is $\mc E_{\ms S}\notin \mc O$.	To denote when $\rho$ can be converted into $\sigma$, we use the standard notation
	\begin{align}
		\rho \xrightarrow[\mc O]{} \sigma\quad \Longleftrightarrow\quad  \exists\, \mc F\in \mc O~~ \text{such that} ~ \mc F[\rho] = \sigma. 
	\end{align}
	The relation ``$\xrightarrow[\mc O]{}$'' between quantum states is, in general, only a partial order, i.e. typically there are pairs of states for which neither $\rho \xrightarrow[\mc O]{}~ \sigma$ nor $\sigma\xrightarrow[\mc O]{}~\rho$ is true. Given a resource theory $\mc{R}$, the central question is: Which state transitions are made possible using free operations? 
	
	In general, it is difficult to characterize the partial order induced by a given resource theory. However, sometimes it is possible to specify a set of \emph{necessary} conditions for the state transition, in terms of so-called \emph{resource monotones}. A resource monotone $f$ with respect to a resource theory $\mc O$ is a function from quantum states to real numbers such that
	\begin{align}
		\rho\xrightarrow[\mc O]{}~ \sigma\quad \Rightarrow \quad f(\rho) \geq f(\sigma). 
	\end{align}
	In other words, since free operations cannot make the state more resourceful,  every resource monotone can be seen as measuring some abstract resource. Because of that, resource monotones are sometimes viewed as resource-theoretic analogues of the second law of thermodynamics.
	A typical example of a resource monotone is the (quantum) relative entropy with respect to the set of free states \cite{brandao_reversible_2015,berta_gap_2022}, i.e.
	\begin{align}\label{eq:relent-monotone}
		f(\rho_{\ms S}) = \inf_{\gamma_{\ms S}\in \mc{S}(\ms{S})} D(\rho_{\ms S}\| \gamma_{\ms S}), 
	\end{align}
	where $D(\rho\| \sigma) := \tr[\rho(\log(\rho)-\log(\sigma))]$
	is the relative entropy between two density matrices (see Eq.~\eqref{eq:relative-entropy} for a more formal discussion).
	Sometimes it is possible to find a \emph{complete} set of resource monotones, that is a family $\{f_\alpha\}$ of resource monotones which collectively characterizes both necessary and sufficient conditions for state transformations. More formally,
	\begin{align}
		\rho \xrightarrow[\mc O]{} \sigma \quad \Leftrightarrow \quad f_\alpha(\rho) \geq f_\alpha(\sigma)\quad\forall \alpha. 
	\end{align}
	It can be shown that, under some general assumptions on the resource theory $\mc{R}$, any complete set of resource monotones must be infinite \cite{Datta2022a}.
	This may be surprising, since in finite dimensions it should be clear that a finite number of parameters specify whether a state transformation is possible. 
		Neverthless the two statements are compatible, because a complete set of monotones allows to evaluate the possibility of state transitions by only checking whether none of the monotones increases.
		It is, however, often possible to determine the possibility of a state transition via more complicated functions of a finite number of monotones.	
	
	Most of the statements in this review are formulated for general quantum resource theories. However, in some cases we will need to restrict ourselves to a specific types of resource theories which we refer to as \emph{permutation-free} resource theories. In such resource theories permuting subsystems (i.e., physically swapping two subsytems) is allowed for free. In this case permuting subsystems can be considered simply as a relabelling of subsystems. To the best of our knowledge, the resource theories considered so far in the literature all fulfil this assumption either fully, or for specific subsystems, such as local subsystems in the context of LOCC discussed below.
		We emphasize, however, that permuting subsystems can not always simply be considered as a free relabelling: For example, swapping qubits in a quantum computer is an essential operations that requires non-trivial gates which pick up errors in general. In a fault-tolerant quantum computer we can expect that swapping \emph{logical qubits} may be done error-free using Clifford operations (see Sec.~\ref{subsec:computation}), but can still require significant computational time.

	With this general introduction in mind, let us turn to some examples.
	
	\subsubsection{Entanglement}
	\label{subsubsec:locc} The idea of resource theories originated with entanglement theory \cite{Bennett1996,vedral1997}, where it was realized that entanglement shared between multiple parties can be useful for certain tasks, and that it cannot be increased in the absence of quantum communication. The theory of quantum entanglement is now a mature field of study reviewed in \cite{Horodecki2009} and forms a central part of quantum information theory \cite{nielsen2002quantum,Wilde2009,watrous2018theory}. 
	The resource theory which captures these restrictions is known as the resource theory of local operations and classical communication (LOCC), denoted as $\mc{R}_{\mathrm{LOCC}} = (\mc{S}_{\mathrm{LOCC}}, \mc{O}_{\mathrm{LOCC}})$.
	The corresponding set of free operations $\mc O_{\mathrm{LOCC}}$ consists of any protocol that is composed of the following operations performed by the different parties:
	\begin{enumerate}[label=\roman*)]
		\item preparing arbitrary local quantum states, 
		\item applying local unitary operations and measurements, 
		\item exchanging classical messages between the parties, 
		\item discarding physical subsystems. 
	\end{enumerate}
	This choice of free operations encapsulates the idea that preparing quantum states locally and communicating classically is ``easy'', but exchanging (or communicating) quantum states in a coherent manner is ``hard''. 
	The set of free states $\mc S_{\mathrm{LOCC}}$ is simply given by all quantum states that can be prepared using free operations $\mc O_{\mathrm{LOCC}}$. 
	Clearly, all the free states can be prepared by using only classical shared randomness. 
	Consequently they are called \emph{classically correlated} states \cite{Werner_1989} or \emph{separable} states \cite{Horodecki_1996,Peres_1996}.
	Any state that is not classically correlated is called entangled.
	
	The amount of entanglement can be quantified using various resource monotones \cite{vedral1997,vidal2000entanglement,Horodecki2009}, which in this case are usually called \emph{entanglement monotones}. 
	For example, in the setting with only two parties Alice $\ms A$ and Bob $\ms B$, the monotone defined in Eq.~\eqref{eq:relent-monotone} translates into the \emph{relative entropy of entanglement}
	\begin{align}
		E_{\mathrm{rel.}}(\rho_{\ms A\ms B}) = \inf_{\gamma_{\ms A \ms B}\in \mc S_{\mathrm{LOCC}}(\ms A\ms B)} D(\rho_{\ms A\ms B}\| \gamma_{\ms A \ms B}).
	\end{align}
	In Sec.~\ref{subsec:entanglement} we review the role of catalysis in the resource theory of LOCC.
	
	\subsubsection{Asymmetry}
	\label{subsubsec:rt-asymmetry}
	Conservation laws place fundamental restrictions on the possible dynamics on physical systems. 
	To focus our attention, let us consider the case when the conserved quantity is the total energy. 
	Suppose we want to include in our quantum description all participating degrees of freedom, so that the dynamics is ultimately unitary. 
	The unitary operator $U$ that describes the dynamics must conserve total energy. 
	In other words, due to Noether's theorem, $U$ has to commute with the group of time-translations $t\mapsto t+s$, which on the Hilbert space is represented by $s\mapsto \exp(-\mathrm i s H)$ with $H$ denoting the total Hamiltonian of the system. 	
	
	The \emph{resource theory of asymmetry} with respect to a group $G$ ~\cite{Janzing2003,Marvian2012}, also known as the resource theory of reference frames  \cite{Bartlett2007,Gour2008, Vaccaro2008}, provides an idealized framework to study the restrictions imposed by such commutation conditions.
	In $\mc{R}_{\mathrm{asymm}} = (\mc{S}_{\mathrm{asymm}}, \mc{O}_{\mathrm{asymm}})$, every system $\ms S$ carries a (projective) unitary representation of the group $G$, given by $g\mapsto W_{\ms S}(g)$. 
	The free operations $\mc O_{\mathrm{asymm}}$ consist of all quantum channels that are \emph{covariant}: a quantum channel $\mc{F}: \mc{D}(\ms{S}) \rightarrow \mc{D}(\ms{S}')$  is covariant with respect to the given representations of the group $G$ if, $\forall g\in G$,
	\begin{align}
		\mc F[W_{\ms S}(g) \rho_{\ms S} W_{\ms S}(g)^\dagger] = W_{\ms S'}(g) \mc F[\rho_{\ms S}] W_{\ms S'}(g)^\dagger.
	\end{align}
	For two independent systems $\ms S_1$ and $\ms S_2$, we further require that they jointly carry the representation $g\mapsto W_{\ms S_1\ms S_2}(g) = W_{\ms S_1}(g)\otimes W_{\ms S_2}(g)$.
	The free states $\mc{S}_{\mathrm{asymm}}$ are all density matrices that are \emph{invariant} or \emph{symmetric} with respect to $G$, i.e.
	\begin{align}
		W_{\ms S}(g) \rho_{\ms S} W_{\ms S}(g)^\dagger = \rho_{\ms S}\quad\forall \rho_{\ms S}\in\mc S_{\mathrm{asymm}}(\ms S).
	\end{align}
	Thus the valuable resources in $\mc{R}_{\mathrm{asymm}}$ are all quantum states which are \emph{asymmetric} with respect to the group $G$.	
	In the particular $\mc{R}_{\mathrm{asymm}}$ with respect to time-translation, the free states are simply stationary states of the time evolution. Moreover, a state is resourceful if and only if it carries coherence between different energies.
	
	In the resource theory of asymmetry, Eq.~\eqref{eq:relent-monotone} translates into the \emph{relative entropy of asymmetry}, also known as the \emph{relative entropy of frameness} \cite{Vaccaro2008,Gour2009},
	\begin{align}\label{eq:REA}
		\mc A(\rho_{\ms S}) = \inf_{\gamma_{\ms S}\in \mc S_{\mathrm{asymm.}}}  D(\rho_{\ms S}\|\gamma_{\ms S}) =  H(\mc G[\rho_{\ms S}]) - H(\rho_{\ms S}),
	\end{align}
	where $H(\rho_{\ms S}) = -\tr[\rho_{\ms S}\log(\rho_{\ms S})]$ is the von~Neumann entropy \cite{vonNeumann1932}. 	
	The right-most equality holds if the infimum is attained, in which case $\mc G$ is the \emph{group-twirling channel} that maps a quantum state to the closest symmetric state. 
	For a compact group, this channel simply corresponds to a group average over its normalized, (unique) left-and right-invariant Haar measure $\mu$:
	\begin{align}
		\mc G[\rho_{\ms S}] = \int W_{\ms S}(g)\, \rho\, W_{\ms S}(g)^\dagger\, \mathrm{d}\mu(g).
	\end{align}
	Every covariant quantum channel can be realized by $(1)$ adding an ancillary system $\ms E$ prepared in a symmetric state, $(2)$ applying a covariant unitary (or, more generally, an isometry) $U$ mapping from $\mc{H}_{\ms{SE}}$ to $\mc{H}_{\ms{S'E}}$, and $(3)$ tracing out the ancillary system $\ms{E}$ \cite{Keyl1999}. This leads to a quantum channel of the form
	\begin{align}
		\mc F[\rho_{\ms S}]= \tr_{\ms E}[U\rho_{\ms S}\otimes \rho_{\ms E} U^\dagger],
	\end{align}
	with the conditions
	\begin{align}
		U W_{\ms S}(g)\otimes W_{\ms E}(g) &= W_{\ms S'}(g)\otimes W_{\ms E}(g) U, \\
		W_{\ms E}(g) \rho_{\ms E} W_{\ms E}(g)^\dagger &= \rho_{\ms E}. 
	\end{align}
	Moreover, the state $\rho_{\ms E}$  may be chosen to be pure.
	In Section~\ref{subsec:asymmetry} we discuss catalysis in the context of symmetries and conservation laws. 
	
	\subsubsection{Athermality}
	\label{subsubsec:athermality}
	Our last example is the resource theory of athermality $\mathcal{R}_{\mathrm{atherm}} = (\mc{S}_{\mathrm{atherm}}, \mc{O}_{\mathrm{atherm}})$, which can be seen as an idealized model of thermodynamics in the quantum realm, see e.g. Refs.  \cite{Janzing2000, Horodecki2013,brandao_resource_2013,brandao_second_2015,gour_resource_2015} and \cite{YungerHalpern2016}. 
	In this resource theory, every system $\ms S$ carries a Hamiltonian $\hat H_{\ms S}$, and $\mathcal{R}_{\mathrm{atherm}} $ is defined with respect to a fixed background temperature $T=\beta^{-1}$ (we set the Boltzmann constant $k_{\mathrm{B}}=1$). The free states are given by thermal (Gibbs) states, 
	\begin{align}
		\gamma_\beta(\hat H_{\ms S}) := \frac{\e^{-\beta \hat H_{\ms S}}}{Z_{\ms S}},\quad Z_{\ms S} = \Tr[\e^{-\beta \hat H_{\ms S}}].
	\end{align}
	Thus the set of free states of a system $\ms S$ consists of a single state, that is $\mc S_{\mathrm{atherm}}(\ms S) = \left\{\gamma_\beta(\hat H_{\ms S})\right\}$.
	There are multiple possibilities for choosing the set of quantum channels representing $\mc{O}_{\mathrm{atherm}}$. Importantly, all such channels must at the very least preserve $\mc{S}_{\mathrm{atherm}}$, and hence map thermal states to thermal states at the same temperature. Channels with this property are referred to as ``Gibbs-preserving maps" \cite{Faist2015}.
	Often it is additionally demanded that the free operations are covariant with respect to the group of time translations (see previous section). This set of operations was coined ``enhanced thermal operations'' \cite{Cwiklinski2015}.
	The above two classes of free operations lead to a valid resource theory, however, they lack a clear physical interpretation, and therefore a more restrictive set of free operations is commonly considered, known as \emph{thermal operations}. These are channels that can be written as 
	\begin{align}
		\mc F_{\ms S}[\rho_{\ms S}] = \tr_{\ms E}\left[U\big(\rho_{\ms S}\otimes \gamma_{\beta}(\hat H_{\ms{E}})\big) U^\dagger \right],
	\end{align}
	where $\gamma_{\beta}(\hat H_{\ms{E}})$ is a Gibbs state of inverse temperature $\beta$, and the unitary $U$ is strictly energy-preserving, i.e. it satisfies $[U, \hat H_{\ms S}+\hat H_{\ms E}]=0$. 
	It is known that thermal operations are a strict subset of enhanced thermal operations \cite{ding2021exploring}. Major progress in understanding the state-transition conditions for enhanced thermal operations was reported in \cite{gour_quantum_2018}, but no simple characterization is known.
	
	An important resource monotone in the resource theory of athermality is the \emph{non-equilibrium free energy}:
	\begin{align}\label{eq:free-energy}
		F^\beta (\rho_{\ms S},\hat H_{\ms S}) &:= \tr[\rho_{\ms S}\hat H_{\ms S}] - \frac{1}{\beta} H(\rho_{\ms S})\\
		&\phantom{:}=\frac{1}{\beta}\left[D\left(\rho_{\ms S}\|\gamma_\beta(\hat H_{\ms S})\right)  - \log Z_{\ms{S}}\right],
	\end{align}
	Since $\mc S_{\mathrm{atherm}}$ only consists of a single state, up to rescaling and a shift by the \emph{equilibrium free energy} $- \beta \log Z_{\ms S}$, the non-equilibrium free energy corresponds to the general monotone defined in Eq.~\eqref{eq:relent-monotone}. When free operations are given by (enhanced) thermal operations, the resource theory of athermality can be seen as the resource theory of asymmetry with respect to time translations, with the additional restrictions imposed on the sets of free states and free operations.
	In Secs.~\ref{subsub:CTOs} and \ref{subsubsec:batteries} we will discuss catalysis in the context of the resource theory of athermality.  
	
	We have discussed three major examples of resource theories. However, there exist many more, such as the resource theory of  contextuality (see Sec.~\ref{subsubsec:contextuality}), non-Gaussianity (see Sec.~\ref{subsubsec:CV}), or stabilizer operations (see Sec.~\ref{subsec:computation}) to name a few that we will encounter in this review. 
	
	\subsection{Basic mathematical tools}
	In this section we introduce some mathematical tools that are often used to describe the properties and the relationships between quantum states, e.g. to capture the partial order induced by a resource theoretic framework. These tools serve as a technical basis to study how such partial orders can change in the presence of catalysts.
	
	\subsubsection{Distinguishability measures}\label{subsubsec:data-processing}
	It is often necessary to measure how easy or difficult it is to distinguish two quantum states. In the context of this review, this will mainly be used for two purposes: $(i)$ quantifying how close the final state of the system $\ms S$ is to a given target state, and $(ii)$ quantifying how close the final state of the catalyst $\ms C$ is to its initial state.
	
	Suppose we receive two datasets (either classical or quantum) from two independent runs of some experiment, each of them derived from raw measurement data by the same post-processing technique. It is then intuitively clear that all the information that distinguishes the two datasets apart has already been present in the measurement data, i.e. before post-processing was applied. 
	In other words, post-processing two datasets in the same way cannot increase their distinguishability. Similarly, no post-processing, as represented generally by a quantum channel $\mc E$, should increase the distinguishability between two quantum states $\rho$ and $\sigma$. 
	Therefore, any operationally meaningful measure of distinguishability $\dist(\cdot, \cdot)$ between two quantum states should satisfy
	\begin{align}\label{eq:DPI}
		\dist(\mc E[\rho],\mc E[\sigma]) \leq \dist(\rho,\sigma)
	\end{align}
	for any two density operators $\rho$ and $\sigma$ and any quantum channel $\mc{E}$. The above inequality is called \emph{data-processing} inequality (DPI) and is a central concept in quantum information theory. An important measure which satisfies DPI is the \emph{trace distance} $\Delta(\cdot, \cdot)$, i.e.
	\begin{equation}\label{eq:tracedistance}
		\Delta (\rho,\sigma) := \frac{1}{2} \|\rho-\sigma\|_1,
	\end{equation}
	where $\norm{\cdot}_1$ is the Schatten $1$-norm. When the two density matrices commute, i.e. when $[\rho, \sigma]= 0$, trace distance reduces to the total variation distance (TVD) between the two probability vectors formed from the eigenvalues of $\rho$ and $\sigma$. Another relevant distinguishablity measure is the (Umegaki) \emph{quantum relative entropy} \cite{umegaki1962conditional}, which is defined as
	\begin{align}\label{eq:relative-entropy}
		D(\rho\| \sigma) := \tr[\rho(\log(\rho)-\log(\sigma))],
	\end{align}
	with $D(\rho\|\sigma)=+\infty$ if the support of $\rho$ is not contained in that of $\sigma$. For commuting density operators $\rho$ and $\sigma$, the quantum relative entropy reduces to the Kullback-Leibler divergence \cite{kullback1951information}. One can also consider generalizations of relative entropy, so called (quantum) R\'enyi divergences, see next section and Ref.~\cite{Tomamichel2016}. 
	While the trace distance is a metric, the quantum relative entropy is not. 
	More specifically, it is neither symmetric in its arguments, nor satisfies the triangle inequality.  
	Nevertheless, it satisfies the data processing inequality of Eq.~\eqref{eq:DPI}, and furthermore the quantum relative entropy has a strong operational relevance~\cite{hiai1991proper}: It gives rise to the logarithm of the minimal type-2 error in a quantum hypothesis testing scenario (involving $\rho$ and $\sigma$ as the null and alternative hypothesis), regularized in the asymptotic limit.
	The relative entropy can also be used to express the mutual information
	\begin{equation}\label{eq:mutual_info}
		I(\ms{A}:\ms{B})_{\rho} = D(\rho_\ms{AB}\|\rho_{\ms A}\otimes \rho_{\ms{B}}),
	\end{equation}
	which is a measure of the amount of correlations between subsystems $\ms{A}$ and $\ms{B}$.
	
	Another important measure of distinguishability is the \emph{fidelity} $F(\cdot, \cdot)$ defined as $	F(\rho,\sigma) = \norm{\sqrt{\rho}\sqrt{\sigma}}_1$,
	The classical counterpart of fidelity is known as Bhattacharyya distance \cite{bhattacharyya1943measure}. As opposed to previous distinguishability measures, fidelity fulfils a reverse DPI, i.e. 
	\begin{align}
		F(\rho,\sigma) \leq F(\mc E[\rho], \mc E[\sigma])
	\end{align}	 
	and is close to unity for states that are similar.
	Although fidelity is not a metric, other metrics derived from fidelity naturally satisfy the data-processing inequality \cite{Gilchrist_2005} and are therefore valid distinguishability measures.

	In the context of a resource theory $\mc R=(\mc S,\mc O)$, every distinguishability measure $\dist(\cdot, \cdot)$  that fulfils the DPI allows to define a resource monotone $f_\dist$ via \cite{Gonda_2023}
	\begin{align}
		f_\dist(\rho_{\ms S}) := \inf_{\sigma_{\ms S} \in \mc S(\ms S)} \dist(\rho_{\ms S},\sigma_{\ms S}).
	\end{align}
	The monotone $f_\dist(\rho_{\ms S})$ therefore measures how the state $\rho_{\ms S}$ can be distinguished from the free states $\mc{S}(\ms S)$ on $\ms S$ as measured by $\dist(\cdot, \cdot)$.

	\subsubsection{Entropic quantifiers}
	\label{subsubsec:renyi}
	Information encoded in physical systems can be conveniently characterized using various entropic quantifiers, most of them having well-established operational interpretations. Perhaps the most well-known entropic quantifier is the Shannon entropy $H(\bm{p})$ \cite{shannon1948mathematical}, which for a probability vector $\bm{p}$ is defined as $H(\vec p) := -\sum_i p_i \log(p_i)$, where $p_i$ denotes the $i$-th element of the vector $\bm{p}$.\footnote{{When $\bm{p}$ corresponds to a probability distribution, we will also sometimes use $p(i)$ to denote the $i$-th element of $\bm{p}$. }} Throughout the review we will use the logarithm of base $2$, i.e., $\log(2)=1$.
	The generalization of Shannon entropy to density operators is known as von~Neumann entropy \cite{vonNeumann1932}, 
	\begin{align}
		H(\rho_{\ms S}) = H(\eig(\rho)) \equiv -\Tr \rho \log \rho, 
	\end{align}
	where $\eig(A)$ denotes the vector of eigenvalues of $A$,  including multiplicities. Another generalization of the concept of entropy are the $\alpha$-R{\'e}nyi entropies \cite{renyi1961measures}, which, for a parameter~\footnote{More often the R{\'e}nyi entropies are defined in literature for $ \alpha\geq 0 $, since for negative values of $ \alpha $, $ H_\alpha(p) $ tends to infinity when $ p $ is not of full rank. However, this generalization proves to be useful in the context of catalysis as we will see later.} $ \alpha\in(-\infty,0)\cup(0,1)\cup(1,\infty)$, are defined as 
	\begin{equation}\label{eq:Renyi_entropy}
		H_\alpha (\vec p) := \frac{{\rm sign}(\alpha)}{1-\alpha} \log \sum_i p_i^\alpha.
	\end{equation}
	
	In the above, we use the convention that $\mathrm{sign}(0)=1$, $0^0=0$ and $0^{\alpha}=+\infty$ for $\alpha<0$.
	In particular, $H_\alpha(\vec p)=\infty$ for $\alpha<0$ if $\vec p$ has an entry $p_i=0$, whereas $H_\alpha(\vec p)$ for $\alpha>0$ only depends on the non-zero entries of $\vec p$.
	For $\alpha=\lbrace 0,1,\infty\rbrace$ the associated entropies are defined by continuity in $\alpha$, i.e.
	\begin{align}
		H_0(\vec p) &= \log(\mathrm{rank}(\vec p)),\\
		H_1(\vec p) &= -\sum_i p_i\log(p_i) = H(\vec p),\\
		H_\infty(\vec p) &= -\log(\max\{p_i\}),
	\end{align}
	where $\mathrm{rank}(\vec p)$ is defined as the number of non-zero elements of $\vec p$, so that $\mathrm{rank}(\eig(\rho_{}))=\mathrm{rank}(\rho)$ for any density operator $\rho$. 
	The $\alpha$-R\'enyi entropy of a density operator $\rho$
	is defined as $H_\alpha(\rho) = H_\alpha(\eig(\rho))$ \cite{Wehrl_1976,Thirring_1980,Ohya1993}, so that
	\begin{align}
		H(\rho) = \frac{\mathrm{sign}(\alpha)}{1-\alpha} \log(\tr[\rho^\alpha]).
	\end{align}
	The R\'enyi entropies are non-increasing in $\alpha$ so that, in particular, $H_0(\rho) \geq H_1(\rho) \geq H_{\infty}(\rho)$. All R\'enyi entropies derive from a family of parent quantities known as R\'enyi divergences \cite{renyi1961measures}. 
	These quantities generalise the Kullback-Leibler divergence (see below), hence they can also be seen as measures of distinguishability between probability distributions. R\'enyi divergences for two probability vectors $\vec p$ and $\vec q$, and parameter $\alpha \in \mathbb{R}$, are defined as
	\begin{align}
		D_\alpha(\vec p\| \vec q) := \begin{cases}
			\frac{\mathrm{sign}(\alpha)}{\alpha-1}\log(\sum_i q_i \big(\frac{p_i}{q_i}\big)^\alpha) &\text{if} \hspace{5pt} \vec p \ll \vec q , \\
			+\infty  &\text{otherwise},
		\end{cases}
	\end{align}
	where $\vec p \ll \vec q$ means that $q_i=0$ implies $p_i=0$ for all $i$. In the limit $\alpha\rightarrow 1$ we recover the Kullback-Leibler divergence \cite{kullback1951information}, or relative entropy
	\begin{align}
		D(\vec p\| \vec q) = \sum_i p_i \log(\frac{p_i}{q_i}).
	\end{align}
	R\'enyi entropies and R\'enyi divergences are linked via
	\begin{align}\label{eq:entropy-divergence}
		D_\alpha(\vec p\| \mm) = \mathrm{sign}(\alpha) \log(d) - H_\alpha(\vec p),
	\end{align}
	where $\mm \propto (1, 1, \ldots, 1)^\top/d$ is the uniform $d$-dimensional probability vector. In contrast to R\'enyi entropies, there is no unique quantum generalization for R\'enyi divergences. 
	Clearly, if $[\rho_{\ms S},\sigma_{\ms S}]=0$ (the quasi-classical case) we can simply diagonalize both density matrices in a common eigenbasis and consider $D_\alpha(\eig(\rho_{\ms S})\|\eig(\sigma_{\ms S}))$. 
	Perhaps the simplest and most commonly used candidate for the quantum generalization of R\'enyi divergences are the Petz-R\'enyi divergences \cite{petz1986quasi}, which for $\alpha\in[0,1)\cup(1,\infty)$ are defined as
	\begin{align}\label{eq:petz-divergence}
		D_\alpha(\rho\|\sigma) &= \begin{cases}
			\frac{1}{\alpha-1}\log(\tr[\rho^\alpha\sigma^{1-\alpha}])\quad&\text{if}\ \rho\ll\sigma\\
			+\infty &\text{otherwise},
		\end{cases}
	\end{align}
	where $\rho\ll\sigma$ means that the support of $\rho$ is contained in the support of $\sigma$, i.e. $\bra\psi\sigma\ket\psi=0$ implies $\bra\psi\rho\ket\psi=0$ for all vectors $\ket{\psi}$. 
	The Petz-R\'enyi divergences fulfill the data-processing inequality for $\alpha \in [0,2]$.
	Another commonly considered family of quantum R\'enyi divergences is the minimal (or sandwiched) R\'enyi divergence \cite{muller2013quantum}, defined as
	\begin{align}
		\tilde D_\alpha (\rho\|\sigma) &=\frac{1}{\alpha-1}\log\left\lbrace \tr[\left(\sigma^{\frac{1-\alpha}{2\alpha}}\rho\sigma^{\frac{1-\alpha}{2\alpha}}\right)^\alpha]\right\rbrace
	\end{align}
	for $\alpha\in(1/2,1)\cup(1,\infty)$, if $\rho\ll\sigma$ and $\tilde D_\alpha(\rho\|\sigma)=\infty$ otherwise\footnote{As is also the case for Petz-R\'enyi divergences, the condition $\rho\ll\sigma$ is only required for $\alpha>1$, which guarantees that the resulting quantity is finite.}. 
	The sandwiched R\'enyi divergence fulfills data-processing inequality~\eqref{eq:DPI} for the given range of $\alpha$ and, just as the Petz-R\'enyi divergence, coincides with the relative entropy $D(\rho_{\ms S}\|\sigma_{\ms S})$ in the limit $\alpha\rightarrow 1$.
	We recommend Ref.~\cite{Tomamichel2016} for detailed information about quantum generalizations of R\'enyi divergences.

	\subsubsection{Majorization}\label{subsubsec:maj}
	Majorization is a preorder between vectors: Given two vectors $ \vec p$ and $\vec q \in \mathbb R^d $, we say that $ \vec p $ majorizes $ \vec q $ if
	\begin{equation}\label{eq:majorization}
		\sum_{i=1}^k p^\downarrow_i \geq  	\sum_{i=1}^k q^\downarrow_i , \qquad \text{for all} \quad k \in \{1,\ldots, d\},
	\end{equation}
	and $\sum_{i=1}^d p_i = \sum_{i=1}^d q_i$,
	where $ \vec p^\downarrow$ and $\vec q^\downarrow $ denote vectors ordered non-increasingly. We will use $\vec p\succ \vec q$ to denote that $\bm{p}$ majorizes $\bm{q}$. For probability vectors, the normalization condition is automatically satisfied, and the partial sums appearing in Eq. \eqref{eq:majorization} can be interpreted as (discrete) cumulative distribution functions (CDFs) for $\vec p^\downarrow$ and ${\vec q}^\downarrow$. 
	Thus $\vec p\succ \vec q$ if and only if the CDF of $\vec p^\downarrow$ is at each point larger or equal than the CDF of ${\vec q}^\downarrow$. Moreover, any deterministic probability vector [e.g. $ \vec p = (1,0,\cdots,0)^\top $] majorizes all probability vectors, while the uniform one $ \mm = (1/d,\cdots,1/d)^\top $ is majorized by all probability vectors of dimension $d$. 
	
	Majorization can be extended to density matrices, in which case it can be seen as a preorder of their spectra. More specifically,  we say that $\rho\succ\sigma$ if $\vec \lambda(\rho) \succ \vec \lambda(\sigma)$, where $\vec \lambda(A)$ denotes the vector of eigenvalues of a matrix $A$ (including multiplicities). This seemingly simple preorder of density matrices can be used to characterize randomness/uncertainty in states, and 
	is tightly linked to the resource theories of (LOCC) entanglement (see Sec.~\ref{subsubsec:locc}) and noisy operations (see Sec.~\ref{sec:illustration}).
	A particularly useful theorem in the context of majorization is the Schur-Horn Theorem that uses majorization to relate the spectrum of a Hermitian matrix with its diagonal \cite{schur1923uber,horn1954doubly}.
	\begin{theorem}[Schur-Horn]\label{thm:schurhorm} 
		Let $H$ be any $d$-dimensional Hermitian matrix with a vector of eigenvalues $\vec \lambda (H)$. The following statements are equivalent:
		\begin{enumerate}
			\item There exists a unitary $U$ such that $\vec \lambda' = {\rm diag}(UHU^\dagger) $.
			\item $\vec \lambda(H) \succ \vec \lambda'$.
		\end{enumerate}
	\end{theorem}
	
	Majorization is an indispensable tool in the theory of statistical comparisons \cite{blackwell1953equivalent}. 
	Suppose that $\bm{p}$ and $\bm{q}$ describe the information encoded in a physical system (e.g. energy distribution). 
	Then we say that $\bm{p}$ is more informative than $\bm{q}$, in the absence of prior knowledge, when the latter can be obtained from the former by a bistochastic processing. 
	Due to the Hardy-Littlewood-Pólya theorem, this is equivalent to $\bm{p} \succ \bm{q}$ \cite{hardy1952inequalities}.  
	More generally, when prior knowledge is available in the form of probability distributions $\bm{p}'$ and $\bm{q}'$ (e.g., thermal distribution of energies), we say that a pair of probability
	distributions $(\bm{p}, \bm{p}')$, is more informative than $(\bm{q}, \bm{q}')$ when there exists a stochastic processing which maps $\bm{p}$ into $\bm{q}$, while also mapping $\bm{p}'$ into $\bm{q}'$. 
	When such a processing exists, then the first pair is said to relatively majorize the second \cite{hardy1952inequalities,ruch1978mixing}, see also \cite{Renes2016}. 
	Finally, we note that several extensions of relative majorization to density matrices have been proposed \cite{buscemi2017quantum,gour_quantum_2018}.
	
	An important generalization of majorization is approximate majorization  \cite{van_der_meer_smoothed_2017,horodecki2018extremal}. 
	Let $d(\cdot, \cdot)$ denote some distance measure between vectors. 
	Then, if $\bm{p} \succ \bm{q}$ does not hold, one can still ask whether $\bm{p}$ majorizes $\bm{q}$ approximately. 
	Formally, one asks whether there exists a sufficiently small error $\epsilon > 0$ and a probability vector $\bm{q}_{\epsilon}$ such that $d(\bm{q}, \bm{q}_{\epsilon}) < \epsilon$ and $\bm{p} \succ \bm{q}_{\epsilon}$. 
	One can also consider a related problem where the approximation error $\epsilon$ is located in the initial state, i.e. the existence of a probability vector $\bm{p}_{\epsilon}$ such that $d(\bm{p}, \bm{p}_{\epsilon})<\epsilon $ and $\bm{p}_{\epsilon} \succ \bm{q}$. 
	The two resulting relations are sometimes called, respectively, post and pre-majorization and are known to be equivalent \cite{Chubb2018,Chubb2019}. 
	One can also ask for the minimal approximation error $\epsilon$ such that either $\bm{p} \succ \bm{q}_{\epsilon}$
	or $\bm{p}_{\epsilon} \succ \bm{q}$ holds. 
	If $d(\cdot, \cdot)$ is the total variation distance, Ref.~\cite{horodecki2018extremal} proposed an algorithm which not only finds the optimal value of $\epsilon$, but also provides explicit constructions of the corresponding optimizers, i.e. so-called the \emph{steepest} and \emph{flattest} states. 
	These constructions can be generalized for relative majorization \cite{van_der_meer_smoothed_2017}, or used to address the majorization preorder in the asymptotic (i.i.d.) limit \cite{Chubb2018,boes_variance_2020}, as well as to address approximately catalytic transformations~\cite{ng_limits_2015,lipka-bartosik_all_2021}.

	\begin{figure}[t]
		\begin{tabular}{|C{0.23\textwidth}|C{0.13\textwidth}|C{0.13\textwidth}|C{0.13\textwidth}|C{0.13\textwidth}|C{0.13\textwidth}|}
			\cline{2-6}
			\multicolumn{1}{c|}{} & \hyperref[subsec:strict]{Strict} & \hyperref[sub:arb_strict]{Arb. Str.} & \hyperref[subsec:corr_cat]{Corr.} & \hyperref[subsec:app_cat]{Approx.} & \hyperref[subsubsec:embezzlement]{Emb.} \\ 
			\cline{2-6}
			\multicolumn{1}{c|}{}&$\AC{}{\hspace{10pt}}$ &$\AC{}{\,\,\,\,\text{arb.}\,\,\,\,}$&$\AC{}{\,\,\,\text{corr.}\,\,\,}$&$\AC{}{\text{approx.}}$& ---\\
			\hline
			\vspace{1pt}${\epsilon_\ms C} \hspace{3.5pt}>\hspace{3.5pt} 0$ &\xmark & \xmark & \xmark &\cmark &\cmark \\
			\vspace{1pt}$\epsilon_{\ms{A}}\hspace{3.5pt} \rightarrow \hspace{3.5pt} 0$ &\xmark &\cmark&\xmark&\nmark&\nmark\\
			\vspace{1pt}$I({\ms{S}}:\ms{C}) \hspace{1.5pt}> \hspace{3.5pt}0$ ~~ &\xmark &\xmark&\cmark&\xmark&\cmark\\
			\hline
		\end{tabular}
		\caption{\textbf{Overview of allowed errors and correlations for the catalytic types detailed in Section \ref{subsec:types}.} The condition specified in the left column is fulfilled for a given type of catalysis if \cmark appears and not fulfilled if \xmark\ appears. The errors are defined with respect to a state transformation $\rho_{\ms{S}} \ot \omega_{\ms{C}} \ot \omega_{\ms{A}} \rightarrow \sigma_{\ms{SCA}}$ as $\epsilon_{i} := d(\omega_{i}, \sigma_{i})$ for $i \in\{\ms{C}, \ms{A}\}$. 
			Subsystem ${\ms A}$ is only required for the formal definition of arbitrarily strict catalysis, but also can be incorporated into the categories of approximate catalysis and embezzlement by viewing $\ms A$ as part of $\ms C$ (indicated by \nmark). $I(\ms S:\ms C)$ denotes the mutual information between subsystems $\ms S$ and $\ms C$ in state $\sigma_{\ms{SC}}$. See Sec.~\ref{sub:arb_strict} for precise definitions of different catalytic types. }
		\label{fig:catalytic-types}
	\end{figure}
	
	\section{Catalytic types}\label{subsec:types}
	In the previous section we gave a general outlook on the concept of catalysis in quantum mechanics. We also described some basic mathematical tools 
	which will help us gain a better understanding of the mechanisms behind catalytic effects.
	Moreover, we saw in Lemma \ref{lem:fundamental_nogo} that catalytic effects can only emerge when at least one of the following conditions is satisfied: 
	\ref{cond:exactness} the catalyst becomes perturbed, \ref{cond:correlation} the catalyst develops some correlations with the other degrees of freedom, or \ref{cond:restrictions} the set of allowed operations is appropriately restricted. 
	The different ways in which these conditions can be combined and quantified lead to different types of catalysis, which we summarize in   
	Table~\ref{fig:catalytic-types}. Each type, in principle, induces a different set of transformation laws (i.e. what is possible and what is not) in a given resource theory. 
	While the variations between different types of catalysis may seem pedantic at first, we will see that these seemingly technical distinctions can give rise to a fundamentally different physical behaviour.
	
	In this section we define and discuss various types of catalysis arising from relaxing conditions \ref{cond:exactness} -- \ref{cond:restrictions} in the context of a general resource theory $\mc{R} = (\mc{S}, \mc{O})$. Later in Sec.~\ref{sec:applications}, we discuss their applications in different physical settings.

	\subsection{Strict catalysis}
	\label{subsec:strict}
	
	In the strictest formulation of catalysis, the catalyst $\ms{C}$ must be returned unperturbed and uncorrelated with main system $\ms{S}$ at the end of the process. Due to Lemma \ref{lem:fundamental_nogo}, non-trivial catalysis in this scenario is only possible if the set of allowed operations is sufficiently restricted. We refer to this category as \emph{strict catalysis}.
	
	\begin{definition}[Strict catalysis]
		A state transition from $\rho_{\ms S}$ to $\sigma_{\ms S}$ is called \emph{strictly catalytic} if there exists a free operation $\mc F\in \mc O$ and a quantum state $\omega_{\ms C}$ such that
		$\mc F[\rho_{\ms S}\otimes \omega_{\ms C}] = \sigma_{\ms S}\otimes \omega_{\ms C}$.
		We denote  a strictly catalytic state transition by
		\begin{equation} 
			\rho_{\ms S} \AC{\mathcal{O}}{} \sigma_{\ms S}.
		\end{equation}
	\end{definition}
	
	In this conservative type, catalysis results from the fact that the set of free operations $ \mathcal{O} $ is non-trivial, and therefore certain state transitions are forbidden. This allows to avoid the consequences of Lemma \ref{lem:fundamental_nogo} by relaxing, in particular, condition~\ref{cond:restrictions}. 
	Even so, one sees how correlations play a critical role:
	Firstly, despite the fact that the initial and final states of the system $\ms{S}$ and the catalyst  $\ms{C}$ are not correlated, they have to be correlated during the dynamical process that implements the state transformation.  In this sense, building temporary correlations between $\ms{S}$ and $\ms{C}$ is what actually enables the state transformation. 
	Secondly, when the set of free operations $ \mc{O} $ consists of non-unitary quantum channels (as is the case, e.g., in the resource theory of athermality), correlations can still build up between $\ms{S}$, $\ms{C}$ and an environment $\ms{E}$ which dilates the original non-unitary quantum channel. These correlations are of coursed formally lost upon discarding the environment.

	It is worth mentioning that any resource monotone $f$ with respect to $\mathcal{O}$ that is additive over tensor products, will also be a resource monotone for any strictly catalytic state transformation. More specifically,
	\begin{align}
		f(\rho_{\ms S}) + f(\omega_{\ms C}) = f(\rho_{\ms S}\otimes\omega_{\ms C})
		&\geq f(\sigma_{\ms S}\otimes\omega_{\ms C})\nonumber\\
		&=f(\sigma_{\ms S}) + f(\omega_{\ms C}).
	\end{align}
	Conversely, suppose $\mc O$ admits a \emph{complete family} of resource monotones: $f_\alpha(\rho_{\ms S}) \geq f_\alpha(\sigma_{\ms S})$ for all valid $\alpha$ implies $\rho_{\ms S} \xrightarrow[\mc O]{}  \sigma_{\ms S}$. 
	Then if the operations admit a \emph{single example} of non-trivial strict catalysis, at least one of the monotones $f_\alpha$ cannot be additive over tensor-products \cite{fritz_resource_2017}. 
	
	The earliest report of strict catalysis in a resource-theoretic framework was given in Ref.~\cite{jonathan_entanglement-assisted_1999} in the context of entanglement theory. 
	The authors affirmatively solved a conjecture of S. Popescu by showing, for the first time, that strict catalysis arises in LOCC (see Sec.~\ref{subsubsec:locc-proper}). 
	This result made clear that the mathematical structure of entanglement is much richer than previously expected.
	It was Ref.~\cite{jonathan_entanglement-assisted_1999} that started the systematic exploration of catalytic effects in entanglement theory and resource theories more generally.
	As discussed in our motivating example, previously a method of generating entangled states using strict catalysis was proposed in Refs. \cite{Phoenix_1993} and \cite{Cirac_1994}
	and connected to the term "catalyst" in \cite{Hagley_1997}. To the best of our knowledge, this is the first explicit appearance of the notion of catalysis in quantum information literature.
	\footnote{This prior example, however, falls outside of the resource theory of LOCC and is therefore little related to the structure of entanglement itself. Indeed, within LOCC it is not possible to generate new entanglement, not even using strict catalysis.}
	
	Naturally, one can also consider strictly catalytic transformations which produce the desired state on the system $\ms{S}$ \emph{approximately}. In this case, instead of transforming the system into the target state $\sigma_{\ms{S}}$, one transforms it into $\sigma_{\ms{S}}^{\varepsilon}$ which is $\varepsilon$-close to $\sigma_{\ms{S}}$. Typically, $\varepsilon$ is quantified using the trace distance in Eq.~\eqref{eq:tracedistance}, but, depending on the situation, other notions of distance can also be used.
	
	In general, relatively little is known on how to construct suitable catalyst states. Known constructions for various types of catalysis are discussed in Section~\ref{sec:constructions}.
	An extreme case arises with \emph{self-catalysis} \cite{duarte_self-catalytic_2016}, where $\omega_C = \rho_{\ms S}$, where a copy of the system catalyzes itself. 
	Finally, one can also consider variations of strict catalysis: 
	there are cases where two state-transitions impossible on their own mutually catalyze each other \cite{feng_mutual_2002}. 
	Instead of demanding that the catalyst is returned exactly, one could try to
	save as much of the resources that are being lost on $\ms S$ in the catalyst $\ms C$, thereby \emph{increasing} the resource content of $\ms C$ while facilitating a state conversion on $\ms S$. 
	This approach is called \emph{super-catalysis}, and was studied in the context of entanglement theory in \cite{Bandyopadhyay2002}. Both self-catalysis, super-catalysis as well as mutual catalysis remain little studied (in particular outside of LOCC), however, and provide an interesting avenue for further research.

	\subsection{Correlated catalysis}\label{subsec:corr_cat}

In strict catalysis, we allow for correlations between the environment $\ms E$ and the joint system $\ms {SC} $, but not between $\ms S$ and $\ms C$. Considering Condition \ref{cond:correlation}, a first step towards relaxing the requirements for strict catalysis is to allow for correlations between $\ms S$ and $\ms C$ to persist at the end of the process. 

	\begin{definition}[Correlated catalysis]
		A state transition from $\rho_{\ms S}$ to $\sigma_{\ms S}$ is called \emph{correlated catalytic} if there exists a free operation $\mc F\in \mc O$ and a quantum state $\omega_{\ms C}$ such that the state $ \sigma_{\ms{SC}} := \mc F[\rho_{\ms S} \ot \omega_{\ms C}]$ fulfils
		\begin{align}
			\tr_{\ms C} \left[\sigma_{\ms{SC}} \right] = \sigma_{\ms S}, \quad {\rm and~} 	\tr_{\ms S} \left[\sigma_{\ms{SC}} \right] = \omega_{\ms C}.
		\end{align}
		We denote a correlated catalytic state transition by
		\begin{equation} 
			\rho_{\ms S} \AC{\mathcal{O}}{\rm corr.} \sigma_{\ms S}.
		\end{equation} 
	\end{definition}
	
We will see in Sec. \ref{subsec:QM} that even when the set of operations $\mc{O}$ consists of all possible unitary operations, non-trivial catalytic effects persist.	

Allowing the catalyst to retain correlations with the system enables bypassing finite-size effects in state transitions 
	for many resource theories, see Sec.~\ref{subsec:partial-order-regularization}. 
	Specifically, the ability to correlate the system with the catalyst radically simplifies the state transition conditions, so that a single monotone is often sufficient to characterize all possible state transformations (see Sec.~\ref{subsubsec:noisy-corr}).
	If a catalyst $\ms C$ remains correlated with some system $\ms{S}$, and is then reused for the same state transition $\rho_{\ms S}\AC{\mc{O}}{\rm corr.} \sigma_{\ms S}$ on a different system $\ms S'$, then the resulting joint state of both systems $\sigma_{\ms{SS'}}$ will, in general, be correlated \cite{vaccaro_is_2018,boes_von_2019,boes_by-passing_2020}:
	\begin{align}
		\sigma_{\ms{SS'}} \neq \sigma_{\ms S}\ot \sigma_{\ms S'}.
	\end{align}
	As emphasized in Section~\ref{sec:concept}, whether such residual correlations are problematic, however, depends on the concrete physical context. In particular, one may envision situations where the ability to bypass finite-size effects overcomes the drawback of residual correlations. 
	In Sec.~\ref{subsec:partial-order-regularization} we will further see that, for a large class of resource theories, the residual correlations can be made arbitrarily small, when sufficiently large catalysts are being used. On the other hand, \cite{rubboli2021fundamental} shows that arbitrarily small residual correlations require arbitrarily large catalysts (as measured by resource content) in a wide class of resource theories. Recent works investigated the interplay between quantum and classical correlations that a catalyst $\ms{C}$ may develop with an external reference $\ms{R}$, under quantum channels that only have local access to $\ms{C}$. In particular, Ref.~\cite{lie2023catalysis} found that any such process must degrade genuinely quantum correlations between subsystems $\ms{C}$ and $\ms{R}$.
	
	
	Strict catalysis typically imposes strong constraints on state interconvertibility.
	A key guiding intuition of why correlations between $\ms S$ and $\ms C$ help to overcome these constraints is that, in any resource theory $\mc{O}$ which allows for permuting systems, we have that for any state $\rho_{\ms S_1 \ms S_2}$,
	\begin{align}
		\rho_{\ms S_1 \ms S_2}\AC{\mathcal{O}}{\rm corr.} \rho_{\ms S_1}\otimes \rho_{\ms S_2}.  
	\end{align}
	To see this, choose a copy of $\ms S_2$ as the catalyst, i.e. take $\omega_{\ms{C}} =\rho_{\ms{S}_2}$, and choose the free operation that simply swaps the catalyst $\ms{C}$ with $\ms S_2$. Thus, with the help of correlated catalytic transformations, one can freely decorrelate subsystems. Furthermore, any monotone for correlated catalytic transformations has to fulfil $f(\rho_{\ms S_1 \ms S_2})\geq f(\rho_{\ms S_1}\otimes \rho_{\ms S_2})$. 
	This often rules out constructions of monotones based on R\'enyi divergences, because the quantum relative entropy is the unique continuous and superadditive R\'enyi divergence.
	In other words, 
	\begin{equation} 
		D\left(\rho_{\ms A \ms B} \| \sigma_{\ms A} \otimes \sigma_{\ms B}\right) \geq D\left(\rho_{\ms A} \| \sigma_{\ms A}\right)+D\left(\rho_{\ms B} \| \sigma_{\ms B}\right)
	\end{equation} 
	holds for all states $ \rho_{\ms {AB}} $; while for any $D_\alpha$ with $ \alpha \neq 1 $, counter-examples to the inequality can be found. 
	
	Refs. 
	\cite{lostaglio_stochastic_2015,muller2016generalization} were the first works to make use of the fact above, in the context of a slightly different scenario called \emph{marginal-correlated catalysis}. In this situation, a multipartite catalyst is used, and a state transition is possible via marginal-correlated catalysis when there exists a free operation $ \mathcal{F} $, and an initially uncorrelated catalyst
	\begin{equation}
		\omega_{\ms{C_1\cdots C_n}} = \bigotimes_{i=1}^n \omega_{\ms{C_i}}
	\end{equation}
	such that
	$\mathcal{F}	[\rho_{\ms S}\otimes 	\omega_{\ms{C_1\cdots C_n}}   ] = \sigma_{\ms S} \otimes \omega_{\ms{C_1 \cdots C_n}}'$,
	where for each $ i = 1, \ldots, n$ we have $ \omega_{\ms{C_i}}' = \omega_{\ms{C_i}}$. 
	In other words, instead of allowing for final correlations to persist between system and catalyst, one allows for correlations to exist between different parts of the catalyst. 
	
	An immediate question arises as to the relationship between correlated catalysis and marginal-correlated catalysis. 
	If $\rho\AC{\mathcal{O}}{\rm corr.}\sigma$ using a catalyst $\omega_{\ms{C}_1}$, and the set $\mc O$ includes permutations of subsystems, then the state transition is also possible via marginal-correlated catalysis, using the catalyst $\widetilde\omega_{\ms C_1 \ms C_2} := \omega_{\ms C_1}\otimes \sigma_{\ms C_2}$.
	To see this, first use the $\ms{C_1}$ part of the catalyst in the same way as for the correlated-catalytic state transition $\rho_{\ms S} \AC{\mathcal{O}}{\rm corr.} \sigma_{\ms S}$, and then swap the system $\ms{S}$ with the second part of the catalyst ($\ms{C_2}$). Consequently, the set of marginal-correlated catalytic state transitions includes the set of correlated catalytic state transitions.
	Marginal-correlated catalysis often leads to similar (or even the same) state transition conditions as correlated catalysis. 
	However, its physical significance is less clear, since in general the catalyst cannot be reused, even when starting with a fresh copy of the system. 
	Moreover, in the resource theory of asymmetry for time-translation, marginal-correlated catalysis can essentially trivialize all state transition conditions \cite{Takagi2022}, see also Sec.~\ref{subsec:no-broadcasting}. This can be seen as a particular form of the general \emph{embezzlement} phenomenon, which we discuss in more detail in Sec.~\ref{subsubsec:embezzlement}.

	\subsection{Arbitrarily strict catalysis}\label{sub:arb_strict}
	So far, in our classification of catalysis, we have not allowed for any errors on the catalyst. However, from a physical point of view, this is likely a stringent restriction -- it is practically impossible for any physical system to undergo an evolution and be returned in \emph{exactly} the same state. Hence, for all practical purposes, it should be sufficient if the catalyst can be returned with a small perturbation. This, in view of Lemma \ref{lem:fundamental_nogo}, corresponds to relaxing condition \ref{cond:exactness}.
	
	A subtle point that can be easily missed is the distinction between two scenarios in which the error on the catalyst can be defined. In the first scenario, one allows for a small error for a \emph{fixed} state of the catalyst. In the second scenario, one first fixes the error, and then constructs the catalyst (which generally depends on the error). Perhaps surprisingly, these two scenarios are significantly different.
	In this Section we will discuss the former, and then address the latter in Sec.~\ref{subsec:app_cat}.
	
	For clarity of presentation we assume that the errors occur on a fixed subsystem of the catalyst.
	To emphasize this distinction, we denote this part with $ \ms A $ (for auxiliary system). 
	\begin{definition}[Arbitrarily strict catalysis] 
		A state transition from $\rho_{\ms S}$ to $\sigma_{\ms S}$ is called \emph{arbitrarily strictly catalytic} if there exists a quantum state $\omega_{\ms C}$, a sequence of quantum states $\lbrace\omega_{\ms{A}}^{(n)}\rbrace_{n=0}^\infty$ on a fixed, finite-dimensional quantum system $\ms A$, and a sequence of free operations $\mc F^{(n)}\in \mc O$ such that:
		\begin{enumerate}
			\item For each $n\geq 1$ we have
			\begin{align}
				\mc F^{(n)}[\rho_{\ms S}\otimes \omega_{\ms C}\otimes \omega_{\ms A}^{(0)}] = \sigma_{\ms S}\otimes \omega_{\ms C}\otimes \omega_{\ms A}^{(n)}.
			\end{align}
			\item The state of the auxiliary system $\ms A$ changes arbitrarily little as $n\rightarrow \infty$, i.e.
			\begin{align}
				\lim_{n\rightarrow \infty} \Delta(\omega_{\ms A}^{(n)},\omega_{\ms A}^{(0)})= 0.
			\end{align}
		\end{enumerate}
		We denote arbitrarily strictly catalytic state transition by
		\begin{equation} 
			\rho_{\ms{S}} \AC{~~~\mathcal{O}~~~}{\emph{arb.}} \sigma_{\ms{S}}.
		\end{equation}
	\end{definition}
	In the example of the resource theory of purity, an alternative, but equivalent, definition of arbitrarily strict catalysis was given in Ref.~\cite{gour_resource_2015} (see Def. $45$ therein). 
	Arbitrarily strict catalysis was first introduced as \emph{exact catalysis} in Ref.~\cite{brandao_second_2015}.
	It might be confusing as to why this catalytic type should be called exact, given that an error does occur on the catalyst. However, there is a good consensus that a state transition which achieves its target state with arbitrarily good precision should be referred to as being exact, if the systems involved are of fixed dimension  [see discussions in Refs. \cite{brandao_second_2015,gour_resource_2015}]. 
	Nevertheless, here we have decided to use the term "arbitrarily strict catalysis" to prevent misunderstandings.  
	
	It is natural to wonder when arbitrarily strict catalysis actually differs from strict catalysis. 
	For some resource theories, such as (LOCC) entanglement theory and athermality, the state transition conditions can be highly sensitive to the rank of the density matrices. Notice, however, that the rank is not a continuous function of the density operator. In particular, any density operator can be approximated arbitrarily well with a full rank operator. In this sense, arbitrarily strict catalysis often allows to regularize such instabilities. 
	To see this, consider two density matrices $\rho_{\ms{S}}$ and $\sigma_{\ms{S}}$, and let $\sigma^{\epsilon}_{\ms{S}}$ be a full-rank approximation of $\sigma_{\ms{S}}$, satisfying  $\Delta(\sigma_{\ms{S}},\sigma^{\epsilon}_{\ms{S}})  < \epsilon$. 
	Suppose that $\rho_{\ms{S}} \AC{\mathcal{O}}{} \sigma_{\ms{S}}^{\epsilon}$ holds for any arbitrarily small $\epsilon > 0$, but not for $\epsilon = 0$. 
	Then, whenever $\mc O$ allows for permuting subsystems, we have $\rho_{\ms S}\AC{\mathcal{O}}{\text{arb.}} \sigma_{\ms S}$. 
	More specifically, choose $\omega_{\ms{A}}^{(0)} = \sigma$ and $\omega_{\ms{A}}^{(n)} = \sigma^{\epsilon_n}$ with $\epsilon_n=1/n$ for all $n > 0$, so that  
	\begin{align}
		\rho_{\ms{S}} \ot {\sigma}_{\ms{A}}\quad \AC{\mathcal{O}}{} \quad\sigma_{\ms{S}}^{\epsilon_n} \ot \sigma_{\ms{A}} \quad\xrightarrow[\mathcal{O}]{} \quad\sigma_{\ms{S}} \ot \sigma_{\ms{A}}^{\epsilon_n},
	\end{align}
	where in the second step we swapped subsystem $\ms{A}$ with $\ms{S}$.  Thus, we see that with the help of arbitrarily strict catalytic transformations, one can avoid the instabilities resulting from changing the rank of a density operator.

	\subsection{Approximate catalysis}\label{subsec:app_cat}
	We now discuss the second scenario that corresponds to relaxation of \ref{cond:exactness} in Lemma \ref{lem:fundamental_nogo}. In this case, one fixes the magnitude of the allowed perturbation of the catalyst, and then construct its state. 
	Generally, this leads to dynamical constraints which are less tight than in the case of arbitrarily strict catalysis. In other words, approximate catalysis usually enables state transformations previously impossible under arbitrarily strict catalysis. Formally, we define
	\begin{definition}[Approximate catalysis]
		\label{eq:approx_cat_def}
		A state transition from $\rho_{\ms S}$ to $\sigma_{\ms S}$ is called \emph{$\epsilon$-approximate catalytic} with respect to the distance measure $\dist(\cdot, \cdot)$ if there exist two quantum states $\omega_{\ms C}$ and $\omega_{\ms C}'$, and a free operation $\mc F\in\mc O$ such that
		\begin{align}
			\label{eq:approx_cat_def2}
			\mc F[\rho_{\ms S} \ot \omega_{\ms C}] = \sigma_{\ms S} \ot \omega_{\ms C}' \quad \text{and}\quad \dist\left(\omega_C', \omega_C\right) &\leq \epsilon.
		\end{align}
		We refer to $\epsilon \geq 0$ as the ``smoothing parameter'' and denote an $\epsilon$-approximate catalytic state transition as
		\begin{align}
			& \rho_{\ms S} \AC{~~~\mathcal{O}~~~}{\emph{approx.}} \sigma_{\ms S}.
		\end{align}
	\end{definition}
	From a  physical perspective, approximate catalysis provides a more realistic framework for investigating catalytic effects, due to reasons argued in the earlier section. 
	Another perspective that motivates approximate catalysis is that some physical processes, due to fundamental reasons, must change the state of the catalyst (ancilla). 
	For example, a problem often studied in thermodynamics is that of minimizing the energy dissipated to the environment. 
	In this case, non-trivial dynamics on the system can only be achieved when the state of the environment changes. 
	From the perspective of approximate catalysis, this problem amounts to finding an appropriate catalytic environment that suffers minimal back-action from the system. (see Sec.~\ref{subsub:dissipation} for more details)
	
	From a mathematical perspective, approximate catalysis can be viewed as an interpolation between strict catalysis and generic activation phenomena. 
	Activation is a phenomenon demonstrating that quantum systems processed in assistance with another, are strictly more useful than when the systems are processed independently. 
	The standard example comes from thermodynamics where the so-called passive states, i.e. states which cannot perform any thermodynamic work, can provide work when processed collectively \cite{Lenard1978,Pusz1978,Alicki2013}. Similar examples can be found in entanglement theory, where ancillary states can reveal non-local properties of quantum states \cite{Masanes2006,Palazuelos2012,cavalcanti2013all,Yamasaki_2022}. 
	In this sense, approximate catalysis puts further restrictions on how the activator might be processed, with the aim of preserving the quality of the activator for future transformations. 
	The second extreme case of approximate catalysis is the already discussed case of (arbitrarily) strict catalysis (see Sec. \ref{sub:arb_strict}). 
	Therefore, for any $\epsilon \geq 0$, we obtain a different set of dynamic constraints, which translates into a different partial order between states. 
	In the language of resource theories this is often referred to as a ``regularization'' of the underlying partial order. In Sec.~\ref{subsec:partial-order-regularization} we describe this effect in more detail.
	We emphasize that even though the \emph{possible state-transitions} depend sensitively on the error $\epsilon$ in the chosen distance-measure, the realized final state of system and catalyst in a state-transformation is stable to small perturbations to their initial states. See also \cite{vidal_approximate_2000} for an early discussion of the robustness of catalysis in the context of LOCC.
	
	The choice of the distance measure $\dist(\cdot, \cdot)$ used in Def.~\ref{eq:approx_cat_def} is essential. Not only must it quantify the closeness of $\omega_{\ms C}'$ to $\omega_{\ms C}$, but should also preserve its catalytic properties (or resource content). 
	It is perhaps surprising that these two requirements are not always simultaneously satisfied. 
	For example, the trace-distance $\Delta(\cdot, \cdot)$ is often used as a default distance measure between two quantum states. This is because of its strong operational meaning in terms of distinguishing states via an optimal measurement \cite{Helstrom_1969,nielsen2002quantum}, and its desirable properties as a metric fulfilling the data-processing inequality (see Sec.~\ref{subsubsec:data-processing}). 
	However, even if an arbitrarily small (fixed) perturbation $\epsilon>0$ in trace-distance is allowed on the catalyst, then, in a large class of resource theories, \emph{any} state-transition can become possible. This phenomenon is known as "embezzlement" and we discuss it in Sec.~\ref{subsubsec:embezzlement} and Sec.~\ref{subsubsec:approx-cat-construction}. 
	One way of understanding this is that, although two states can be close to each other in terms of the trace-distance, extensive quantities can still grow with the logarithm of the Hilbert-space dimension, see also Sec.~\ref{subsubsec:non-continuity}. Therefore for any given target error $\epsilon$, one can always choose a large enough dimension such that the change in resource monotone is significant. 
	
	To prevent embezzlement, one can consider a stronger restriction, for example that the deviation in an extensive quantity must be small, e.g. $\dist(\rho,\tilde\rho) =|H(\rho) - H(\tilde\rho)|$.
	A more stringent candidate would be to ask for the largest deviation in a set of quantities (monotones) $\{f_\alpha\}_\alpha$, e.g.
	\begin{equation} 
		\dist(\rho,\tilde\rho) = \sup_{\alpha} |f_\alpha(\rho) - f_\alpha(\tilde\rho)|,
	\end{equation}
	or even a linear combination of several such distance measures.
	As an alternative to changing $d(\cdot, \cdot)$, one may also let the allowed error $\epsilon$ depend on different parameters such as the dimension of the state, see \cite{brandao_second_2015} and Sec.~\ref{subsubsec:noisy-approx}. 
	To summarize, the choice of $d(\cdot, \cdot)$ and allowed error $\epsilon$, in general, qualitatively change the landscape of possible state transitions. 
	
	When the state $\omega_{\ms C}$ of $\ms C$ changes, its ability to catalyze other transformations generally decreases, and therefore it becomes less useful. 
	Still, if the final state $\omega_{\ms C}'$ of the catalyst is sufficiently close to $\omega_{\ms C}$,
	it may be possible to re-use it without significantly affecting the transformation on ${\ms S}$. 
	Concretely, consider a free transformation $\mc F$ that leads to a perturbed catalyst $\omega_{\ms C}'$ satisfying Eq.~\eqref{eq:approx_cat_def2}, where $\dist(\cdot, \cdot)$ satisfies the triangle-inequality and data-processing inequality (see Sec.~\ref{subsubsec:data-processing}). Re-using the same catalyst on a new copy of the system $\ms S$ and $\mc{F}$ leads to a state $\widetilde\sigma_{\ms S} := \Tr_{\ms C}[\mc F[\rho_{\ms S} \ot \omega_{\ms C}']]$, where 
	\begin{align}
		\dist(\widetilde\sigma_{\ms S}, \sigma_{\ms S}) :&= \dist\left(\Tr_{\ms C} \mc F(\rho_{\ms S} \ot \omega_{\ms C}'), \sigma_{\ms S} \right) \nonumber\\
		&\leq \dist \left(\rho_{\ms S} \ot \omega_{\ms C}',\rho_{\ms S} \ot \omega_{\ms C}\right)\quad\leq \epsilon,
	\end{align}
	which follows from data-processing and triangle inequality, with Eq. (\ref{eq:approx_cat_def2}). 
	Suppose we now want to use the same catalyst $n$ times to implement $\mc F$ (each time with a fresh copy of $\rho_\ms{S}$). To look at the resulting state, let us introduce intermediate states for each step $0 \leq i \leq n$, 
	\begin{align}
		\sigma_{\ms S}^{(i+1)} &:= \Tr_{\ms C} \mc F\left[\rho_{\ms S} \ot \omega_{\ms C}^{(i)}\right], \\
		\omega_{\ms C}^{(i+1)} &:= \Tr_{\ms S} \mc F\left[\rho_{\ms S} \ot \omega_{\ms C}^{(i)}\right],
	\end{align}
	with identifications $\sigma_{\ms S}^{(1)} := \sigma_{\ms S}$ and $\omega_{\ms C}^{(0)}:= \omega_{\ms C}$. The state of the system and the catalyst after $k$ transformations satisfy
	\begin{align}
		\dist(\sigma_{\ms S}^{(k)}, \sigma_{\ms S}) \leq k \epsilon \quad \text{and} \quad \dist(\omega_{\ms C}^{(k)}, \omega_{\ms C}) \leq k \epsilon,
	\end{align}
	which is a again the consequence of triangle inequality, data-processing inequality, and Eq. (\ref{eq:approx_cat_def2}).  
	The catalyst can therefore be used approximately $k = \alpha/\epsilon$ times before it accumulates an error of size $\alpha$. 
	Intuitively, in approximate catalysis, we try to make $\epsilon$ as small as possible, so that the catalyst retains its catalysing properties. 
	From the above discussion it should be clear that quantifying the perturbation in the state of the catalyst, as well as identifying its usefulness for future applications, is at the core of approximate catalysis. 
	
	We close this section with a comment on nomenclature. 
	Of course one can also consider combinations of different types of catalysis that arise from lifting several restrictions at once. 
	For example, by allowing for an error on the catalyst $\ms{C}$ \emph{and} correlations between $\ms S$ and $\ms C$, one can study \emph{correlated approximate catalysis}. In fact, some authors simply use the term "approximate catalysis" for that setting, see for example Ref.~\cite{Datta2022}.  
	
	\subsection{Embezzlement}\label{subsubsec:embezzlement}
	In the last section, we saw that when the state of the catalyst is allowed to change, it is often possible to unlock transformations that are impossible under (arbitrarily) strict catalysis.
	This naturally comes at the cost of reducing the usefulness of the catalyst for future transformations. One might argue that in such cases, rather than using $\ms{C}$ catalytically, one is consuming it as a resource, non-free state.
	However, surprisingly sometimes it is possible for the ``catalyst" to completely lift all of the relevant dynamical constraints, but in the same time suffering almost no reduction of its catalytic capabilities. 
	This phenomenon was first reported in Ref.~\cite{van_dam_universal_2003}, where it was dubbed \emph{embezzlement}. 
	
	Embezzlement provides a mechanism for simulating forbidden dynamics on $\ms{S}$, using operations that are free on the composite system $\ms{SC}$, while perturbing $\ms{C}$ arbitrarily little in trace distance.
	Formally, it is defined as follows:
	
	\begin{definition}[Embezzlement]\label{def:emb}
		Given a system $ \ms S $ and a set of free operations $ \mathcal{O} $, a state $ \omega_{\ms C} $ is a $\delta$-embezzler if, for all $\rho_{\ms S}$ and $\sigma_{\ms S} \in \mathcal{D}(\ms{S})$, there exists $\mathcal{F} \in \mathcal{O}$ such that
		\begin{align}
			\label{eq:embezzlement_def_eq1}
			\Delta\left(\mathcal{F}[\rho_{\ms S} \ot \omega_{\ms C}],\sigma_{\ms S} \ot \omega_{\ms C}\right)\leq \delta.
		\end{align}
		An embezzling family is a sequence of $\delta_n$-embezzlers $ \lbrace \omega_{\ms C}^{(n)}\rbrace_{n}$ such that $\delta_n \rightarrow 0$ as $n \rightarrow \infty$.
	\end{definition}
	
	At first sight, embezzlement can be viewed as a particular instance of approximate catalysis (see Sec. \ref{subsec:app_cat}). However, it was initially surprising that approximate catalysis can fully trivialize \emph{arbitrary} state transitions, while using a single catalyst which is not fine-tuned for a particular initial system state.
	The very first example of an embezzlement family was introduced in $\mathcal{R}_{\text{LOCC}} $ in Ref.~\cite{van_dam_universal_2003}, see Sec.~\ref{subsubsec:LOSR}.
	In Sec.~\ref{subsubsec:approx-cat-construction}, we describe a general construction for embezzlers that works whenever permutations of subsystems are allowed. This construction first appeared in Ref.~\cite{Leung2013coherent}.

	Embezzlement results from the fact that trace distance is not sensitive enough to capture the difference in resource content between two quantum states. 
	In other words, when a quantum channel $\mathcal{E}$ transforms one state into another, it can happen that there is no optimal quantum measurement that can determine (signifcantly better than a random guess) whether the resource has been consumed during $\mathcal{E}$.  
	
	Let us describe a simple example illustrating the above idea. Consider a wave function of a harmonic oscillator that is spread over many energy levels, e.g. a coherent state of light with a large photon number. Such states are usually used to describe light emitted from a laser with a fixed frequency (see Sec.~\ref{subsec:circumventing_conservation_laws}). 
	Shifting the wave function in the energy space can change the expected energy (photon number) by a fixed amount, while in the same time perturbing the state arbitrarily little in terms of trace distance. In this case, even when a single photon disappears from the beam of light, there is no physical process that can determine with high probability whether this process actually happened. See Fig. \ref{fig:emb_energy} for details.
	
	\begin{figure}[t]
		\includegraphics[width=\linewidth]{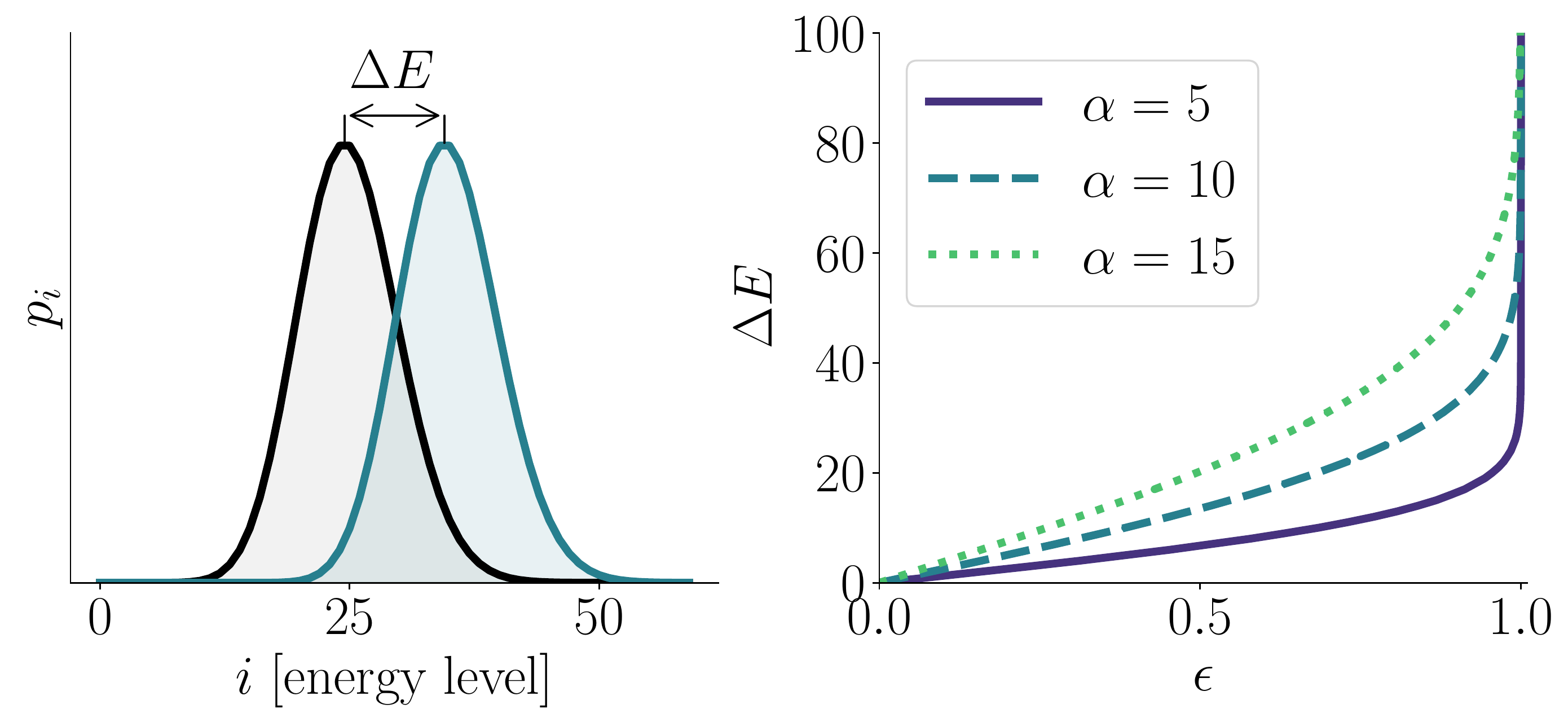}
		\caption{\textbf{Embezzlement of energy.} Consider a coherent state of light $\ket{\alpha}$ with energy levels $\ket{i}$ whose energy variance is given by $\text{Var}(E) = |\alpha|^2$ (see Eq. (\ref{eq:coh_state}) for details). The left panel shows the change in the occupation $p_i = \langle i |\alpha|i \rangle$ of the energy level $i$, obtained by shifting the state by $\Delta E$, i.e. $\ket{\alpha} \rightarrow \ket{\beta} = S \ket{\alpha}$ with $S = \sum_{i=0}^{\infty} \dyad{i+\Delta E}{i}$. The right panel shows the relationship between $\Delta E$ and trace distance $\epsilon := \frac{1}{2}\norm{\alpha-\beta}_1$ for different values of $\alpha$: A large energy variance allows to change the average energy signifcantly while maintaining a large overlap with the initial state.
			}
		\label{fig:emb_energy}
	\end{figure}
	
	In practice, the above mechanism can be used to implement (coherent) unitary operations on the system $\ms S$ that seemingly violate conservation of energy. This is achieved using an interaction between the system $\ms{S}$ and a laser field $\ms{C}$ which is globally energy-preserving, while perturbing the state of the field arbitrarily little. Indeed, as we discuss in Sec.~\ref{subsec:circumventing_conservation_laws}, whenever a large-dimensional system $\ms{C}$ is used to implement a unitary dynamics on $\ms{S}$, it can be done in such a way that the state on $\ms{C}$ is perturbed arbitrarily little. 
	Thus, in a sense, embezzlement is quite ubiquitous in physics, and occurs whenever a ``macroscopic'' quantum system is used to implement a controlled unitary dynamics on a smaller quantum system.
	In this case, the system $\ms C$ can be seen to act as a macroscopic \emph{reference frame} for the time-translation group, see also Sec.~\ref{subsec:asymmetry} for more results relating catalysis to quantum reference frames.
	In Sec.~\ref{subsubsec:non-continuity} we also discuss Ref.~\cite{Leung2019}, which shows how embezzlement 
	can lead to interesting \emph{mathematical} results in quantum theory, namely that there is a limit to how continuous extensive quantities can be. 

	\subsection{Infinite-dimensional catalysis}
	\label{subsec:infinite-dimensional}
	So far we have almost exclusively discussed quantum systems described by finite-dimensional Hilbert spaces. 
	While this is very natural from the perspective of quantum information theory, physics also requires infinite dimensional Hilbert spaces.
	This is exemplified by continuous-variable systems, such as quantum harmonic oscillators, or more generally, by quantum field theory.
	It is therefore reasonable to ask how the phenomenon of catalysis changes when we allow for infinite-dimensional catalysts. 
	
	In this discussion, we keep the system ${\ms S}$ finite-dimensional to highlight the key differences between finite and infinite-dimensional catalysis. For a discussion of catalytic effects in continuous-variable systems and quantum optics we refer to Sec.~\ref{subsec:optics}.  
	We also restrict our attention to strict and correlated catalysis, since embezzlement already shows that allowing for fixed errors (even though arbitrarily small) on the catalyst trivializes the state transition problem in infinite dimensions. 
	Furthermore, we will also discuss the notion of \emph{perfect embezzlement}, which sheds light on some fundamental questions about the formulation of quantum theory. 
	It was discovered early on in Ref.~\cite{daftuar_sumit_kumar_eigenvalue_2004} that there is a strict difference between catalysis with $(i)$ finite-dimensional catalysts, but of arbitrarily large dimension, and $(ii)$ infinite-dimensional catalysts. Specifically, in $\mathcal{R}_{\rm LOCC}$, there exist bipartite states $\ket\psi$ and $\ket\phi$ such that $\ket \phi$ is reachable from $\ket \psi$ when using $(ii)$, but not with $(i)$.
	Moreover, Ref.~\cite{aubrun_catalytic_2008,aubrun_stochastic_2009} showed that the target state $\ket{\phi}$ cannot even be reached with arbitrarily small error with finite-dimensional catalysts. The same reference also made progress towards characterizing the set of reachable states (of a fixed dimension) using infinite-dimensional catalysts. However, to our knowledge, no complete characterization of this set has been given. Since the results described above are all based on majorization, they immediately transfer to the resource theory of noisy operations (see Sec.~\ref{sec:illustration}).

	We close this section by discussing embezzlement using infinite-dimensional systems. 
	Based on the discussion in Sec.~\ref{subsubsec:embezzlement} one may be tempted to expect  
	that an infinite-dimensional catalyst ${\ms C}$ could be used to perform \emph{perfect} embezzlement -- no changes in the catalyst, while still any state on a finite-dimensional system ${\ms S}$ could be achieved (with an arbitrary accuracy). 
	However, as is usual with infinite-dimensional spaces, subtleties exist that must be carefully adressed. 
	Indeed, Refs. \cite{Luijk2024a,Luijk2024b,cleve_perfect_2017} show that \emph{local unitary} perfect embezzlement is impossible in the framework of LOCC if space-like separated
	parties are modelled by tensor products of Hilbert spaces, as is common in quantum information theory. 
	In contrast, perfect embezzlement becomes possible in a so-called \emph{commuting operator framework}, see also Sec.~\ref{subsubsec:non-locality}. 
	This is common in quantum field theory, where space-like separated parties are modelled by commuting operators on a single, infinite dimensional Hilbert space. 
	In this context, Ref.~\cite{cleve_perfect_2017} showed that, for every finite dimension $d$, there exists a quantum state $\ket{\Omega}$ on a separable and infinite-dimensional Hilbert space $\mc R$, such that, for every $\ket{\psi}\in \mathbb C^d\otimes \mathbb C^d$, there exist unitaries $u_\psi$ and $v_\psi$ on $\mathbb C^d\otimes \mathcal R$ and $\mathcal R\otimes \mathbb C^d$, respectively, which satisfy
	\begin{align}\label{eq:embezzling_infinite_cleve}
		(u_\psi\otimes \id)(\id\otimes v_\psi)\ket{0}\otimes \ket\Omega\otimes \ket{0} =  \ket{\psi}\otimes\ket{\Omega},
	\end{align}
	for some fixed $\ket{0}\in \mathbb C^d$ and such that $[u_\psi\otimes\id,\id\otimes v_\psi]=0$. Eq. \eqref{eq:embezzling_infinite_cleve} should be read with an implicit re-ordering of tensor factors.
	
		The construction by \cite{cleve_perfect_2017} has the drawback that in general $[u_\psi\otimes \id, \id\otimes v_\varphi]\neq 0$ if $\ket\psi$ and $\ket\varphi$ are different states (arising, e.g., because both parties attempt to embezzle different states). This shows that a proper bipartite structure is missing. 
		Refs.~\cite{Luijk2024a,Luijk2024b} show that if Minkowski spacetime is partitioned into two wedges $\mathsf L$ and $\mathsf R$, each of which is interpreted as one local system, then \emph{all pure states} in the vacuum representation of relativistic quantum field theories
		allow to embezzle \emph{any finite dimensional entangled state} $\ket{\psi}\in\mathbb C^d\otimes\mathbb C^d$ to \emph{arbitrary precision} in the sense that for every $\varepsilon>0$ there exists local unitaries $u_{\psi,\varepsilon},v_{\psi,\varepsilon}$ such that
		\begin{align}
			\norm{u_{\psi,\varepsilon}v_{\psi,\varepsilon}\ket{0}_{\mathsf A}\otimes\ket{0}_{\mathsf B} \otimes \ket{\Omega} - \ket\psi\otimes\ket\Omega} <\varepsilon,
		\end{align}
		where $u_{\psi,\varepsilon}$ acts on $\mathsf{LA}$ and $v_{\psi,\varepsilon}$ acts on $\mathsf{RB}$, respectively.
		Importantly, $[u_{\psi,\varepsilon},v_{\varphi,\varepsilon'}]=0$ for all states $\ket\psi,\ket\varphi$, because they are elements of commuting von Neumann algebras $\mathcal M_{\ms{LA}}$ and $\mathcal M_{\ms{RB}}$ associated to the different subsystems.  
		These results are shown by establishing a quantitative connection between embezzlement of entanglement and the classification of von Neumann algebras.

	In Sec.~\ref{subsubsec:non-locality}, we will see that the (non-)existence of perfect embezzlers is closely related to the fact that only certain kinds of correlations can exist in quantum theory. 
	Moreover, in Sec.~\ref{subsubsec:CV}, we also discuss catalysis in the context of Gaussian states and operations on continuous variable systems, where both the catalyst $\ms{C}$ and the system $\ms{S}$ are infinite-dimensional. Note, however, that the Gaussian framework is special because states and operations can be fully represented using finite-dimensional matrices.
	
	Finally, we are not aware of studies of perfect embezzlement outside of the pure-state LOCC framework. Still, it is possible to make the following general remark: If the resource theory in question has an additive monotone $f$, then it must diverge on a perfect embezzler: 
	Consider a transformation $\rho_{\ms S}\rightarrow \sigma_{\ms S}$ using a (hypothetical) perfect embezzler $\omega_{\ms C}$. Then 
	\begin{align}\label{eq:perfect-embezzlement-monotone}
		f(\rho_{\ms S}) + f(\omega_{\ms C}) \geq f(\sigma_{\ms S}) + f(\omega_{\ms C})
	\end{align}
	by monotonicity under free operations and additivity.
	But since we are considering embezzlement, we can choose $\rho_{\ms S},\sigma_{\ms S}$ so that $f(\rho_{\ms S}) < f(\sigma_{\ms S})$, which leads to a contradiction with \eqref{eq:perfect-embezzlement-monotone} unless $f(\omega_{\ms C})=\infty$.
	For example, a perfect embezzler in the resource theory of quantum thermodynamics would require an infinite amount of free energy if it were to exist (mathematically).
	
	\subsection{State-independent catalysis}\label{subsec:state-independent}
	
	It is natural to consider the possibility that a given catalyst $\ms{C}$ prepared in a state $\omega_{\ms C}$ remains a catalyst not only for a single state, but for a \emph{set} of input states. Given a free operation $\mc F: \mc D(\ms S\ms C)\rightarrow \mc D(\ms S'\ms C)$ in a resource theory $\mc{R}$, and a density operator $\omega_{\ms C}$, we can define the associated set of catalytic states $\mc C(\mc F,\omega_{\ms C})$ by
	\begin{align}
		\mc C(\mc F,\omega_{\ms C}):=\left\lbrace \rho_{\ms S} \big|	\tr_{\ms S'}\left[\mc F[\rho_{\ms S}\otimes\omega_{\ms C}]\right] = \omega_{\ms C}\right\rbrace.
	\end{align}
	The set $\mc C(\mc F,\omega_{\ms C})$ may of course be empty. The other extreme case occurs when $\ms C$ is a catalyst for any input density matrix. In this case we have the following definition.
	
	\begin{definition}[Catalytic quantum channel] 
		The pair $(\mc F,\omega_{\ms C})$ is called a \emph{catalytic quantum channel} if $\mc C(\mc F,\omega_{\ms C})=\mc D(\ms S)$.
	\end{definition}
	Despite the natural definition, state-independent catalysis has received comparably little attention so far. Ref.~\cite{vidal_catalysis_2002} show how certain unitary operations can be realized under the constraints imposed by LOCC, using a catalyst that cannot be realized otherwise. 
	Furthermore, Ref.~\cite{lie_randomness_2021} considers the special case where $\mc F[\cdot]=U(\cdot)U^\dagger$ is a unitary channel acting on $\ms S\ms C$. They refer to the effective quantum channel $\mc E(\cdot) := \Tr_{\ms{C}}[U((\cdot)_{\ms{S}} \ot \omega_{\ms{C}})U^{\dagger}]$ induced on $\ms{S}$ by the pair $(U,\omega_{\ms C})$ as a \emph{randomness-utilizing quantum channel}. 
	To see why this is an adequate description, note that such a randomness-utilizing quanum channel must be doubly-stochastic, i.e., it must leave the maximally-mixed state invariant.
	Moreover, it follows that $[U,\id_{\ms S}\otimes\omega_{\ms C}]=0$. 
	In particular, when $\omega_{\ms C}$ has no degenerate eigenvalues, i.e., all of its eigenspaces are one-dimensional, this implies that $U = \sum_i U_i \otimes \proj{i}_{\ms C}$, where $\ket{i}_{\ms C}$ is the eigenbasis of $\omega_{\ms C}$. 
	The resulting dynamics on $\ms S$ is then given by a mixed-unitary channel of the form
	\begin{align}\label{eq:mixed-unitary}
		\rho_{\ms S} \mapsto \sum_{i=1}^m p_i U_i \rho_{\ms S} U_i^\dagger,
	\end{align}
	where $p_i = \bra{i}\omega_{\ms C}\ket{i}$.  
	
	Channels of the form in Eq.~\eqref{eq:mixed-unitary} can be interpreted as a (non-selective) measurement of classical information represented by $\ms C$, and a unitary processing $U_i$ applied to $\ms S$, conditioned on the measurement-outcome $i$. They can also be interpreted as quantum channels that can be reversed by first measuring the environment appearing in the Stinespring dilation, and then performing a correcting unitary operation to recover $\rho_{\ms S}$ \cite{Gregoratti_2003}.

	It is currently an open problem how to characterize the set of all randomness-utilizing channels. Ref.
	\cite{lie_correlational_2021} show that they are a strict subset of doubly-stochastic quantum channels. Moreover, they also show that $U$ induces a randomness-utilizing quantum channel on $\ms{S}$ for some state $\omega_{\ms C}$ if and only if $U^{{T}_{\ms S}}$ is also a unitary, where ${T}_{\ms S}$ is a partial transpose on $\ms{S}$. In this case, the unitary $U$ also induces a (generally different) randomness-utilizing channel for the maximally-mixed catalyst state $\omega_{\ms C} = \id_{\ms C}/d_{\ms C}$.
	Moreover, it follows from a result in \cite{Haagerup_2011} that randomness-utilizing quantum channels are a strict superset of mixed-unitary channels \cite{Lie_private_2022}.
	
	\subsection{Illustrative example: Noisy Operations}\label{sec:illustration}

	\begin{figure}[h]
		\includegraphics[width=\linewidth]{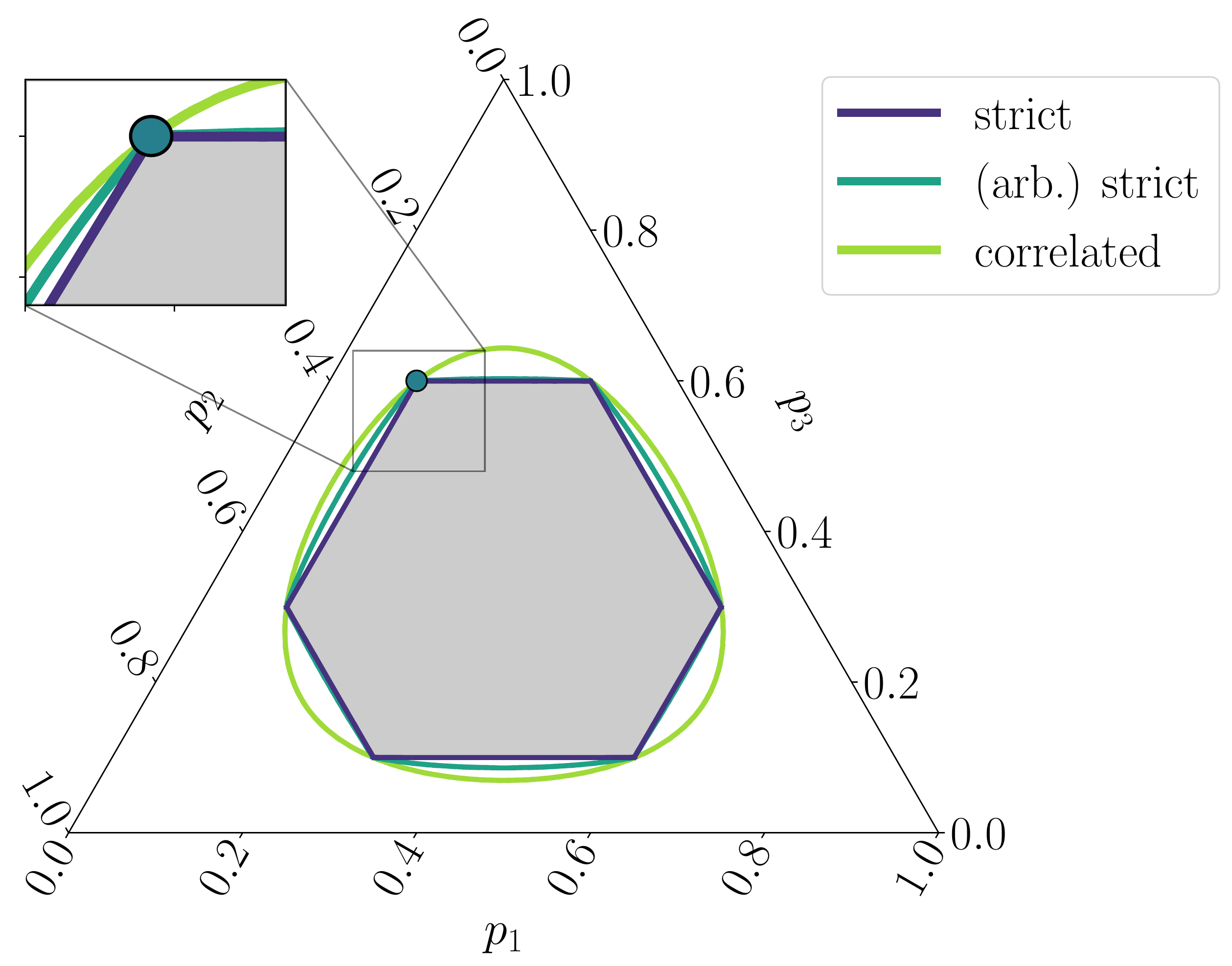}
		\caption{\textbf{States achievable with Noisy Operations under different catalytic types.} The three coloured sets indicate achievable states starting from a fixed state $\rho$ (turquoise dot) with eigenvalues $\eig(\rho) = [0.65, 0.2, 0.15]$ under various types of catalytic transformations. In this case ($d_{\ms{S}} = 3$) the set of states achievable under strict catalysis coincides with the set of states that are majorized by $\rho$, that is all states achievable without using a catalyst. Every state may be obtained to arbitrary precision with embezzlement.}
		\label{fig:noisy-example}
	\end{figure}
	
	In this section, we compare the various types of quantum catalysis that were formally defined in Sec. \ref{subsec:types}. We do so by focusing on a simple, yet illustrative example of majorization-based resource theories. In such resource theories the state transition conditions are fully characterized by majorization. 
	The paradigmatic example here is the resource theory of noisy operations (NO)~\cite{Horodecki2003a,Horodecki2003b},
	also known as the resource theory of purity or informational non-equilibrium \cite{gour_resource_2015}. In this resource theory the set of free operations $\mc{O}_{\mathrm{NO}}$ consist of all quantum channels that can be implemented with a maximally-mixed environment $\ms{E}$, i.e.
	\begin{align}
		T\in \mc{O}_{\mathrm{NO}} \quad \Rightarrow \quad	T(\rho) = \Tr_{\ms E} \left[U \left(\rho_{\ms S} \ot \frac{\mathbb{1}_{\ms E}}{d_{\ms E}}\right)U^{\dagger}\right]
	\end{align}
	for some dimension $d_{\ms E}$. The Schur-Horn lemma (Theorem \ref{thm:schurhorm}) implies that 
	\begin{equation}\label{eq:conditions-NO} 
		\rho \xrightarrow[\mc O_{\mathrm{NO}}]~\sigma \quad \Leftrightarrow\quad  \eig(\rho)\succ \eig(\sigma) .
	\end{equation}
	Below we summarize families of state-transition conditions that characterize different types of catalytic transformations. We will see that these conditions can be formulated in terms of progressively fewer entropic conditions. To illustrate this, 
	Figure~\ref{fig:noisy-example} shows quantum states that can be obtained using different types of catalysis when starting from a fixed initial state. 
	Nielsen's theorem (see Theorem~\ref{thm:nielsen}) implies that all the results equivalently apply to the case of LOCC operations restricted to pure states. In fact many of the results have first been obtained in this setting.

	For the following, note that the R\'enyi entropies $H_\alpha$ are \emph{anti-monotone}, as they are Schur-concave, i.e.:
	\begin{align}
		\eig(\rho)\succ \eig(\sigma)\quad\Rightarrow\quad H_\alpha(\rho) \leq H_\alpha(\sigma).
	\end{align}
	\eqref{eq:conditions-NO} then implies that the R\'enyi entropies can only \emph{increase} under free operations.
	
	\emph{Strict catalysis.} A strictly catalytic state-transformation requires a (finite-dimensional) density operator $\omega_{\ms C}$ such that
	\begin{align}
		\eig(\rho_{\ms S})\otimes\eig(\omega_{\ms C}) \succ \eig(\sigma_{\ms S})\otimes\eig(\omega_{\ms C}).
	\end{align}
	Conditions for strict catalysis were first derived in Ref.~\cite{klimesh_inequalities_2007,turgut_catalytic_2007} (see also the more recent works \cite{pereira_dirichlet_2013,kribs_trumping_2015,pereira_extending_2015-1}), and are known as \emph{trumping conditions}. Here we provide a simplified, but still equivalent, set of conditions expressed in terms of the R\'enyi entropies $H_\alpha$, as proposed in Ref.~\cite{brandao_second_2015}. More specifically, unless $\rho_{\ms S}$ and $\sigma_{\ms S}$ are unitarily equivalent, we have
	\begin{equation}\label{eq:noisy-strict} 
		\rho_{\ms S} \AC{\mathcal{O}_{\mathrm{NO}}}{} \sigma_{\ms S}\quad\Leftrightarrow\quad \begin{cases}
			H_\alpha(\rho_{\ms S}) < H_\alpha(\sigma_{\ms S})\quad \forall \alpha\in\mathbb R\setminus\{0\},\\
			H_0(\rho_{\ms S}) \leq H_0(\sigma_{\ms S}).
		\end{cases}
	\end{equation}
	If an arbitrarily small error on $\ms S$ (but not $\ms C$) is allowed, then the above strict inequalities 
	relax to non-strict ones. Interestingly, 
	strict catalysis can only enable new transformations if the dimension of the system is large enough, i.e. $d_{\ms S}>3$ \cite{jonathan_entanglement-assisted_1999}. 
	General bounds on the catalyst (e.g., its dimension) have been obtained by \cite{sanders_necessary_2009,grabowecky_bounds_2019}.
	
	\emph{Arbitrarily strict catalysis.} Since the R\'enyi entropies $H_\alpha$ for $\alpha>0$ are continuous, the conditions in Eq.~\eqref{eq:noisy-strict} are stable under small perturbations in the state of the catalyst, which is the case for arbitrarily strict catalysis. 
	The conditions corresponding to $\alpha<0$, on the other hand, may be removed. This can be achieved by introducing a qubit $\ms{A}$ initially prepared in a pure state $\omega_{\ms A}=\proj{\psi}$ that is returned in a (full-rank) state $\omega_{\ms{A}}'$ arbitrarily close to $\proj{\psi}$. 
	Now, observe that any pure state has $H_\alpha(\omega_{\ms A})=\infty $ for $\alpha<0$, and in the same time, any full-rank state must have finite R\'enyi entropies, i.e. $H_\alpha(\omega_{\ms A}') < \infty$. This simple observation effectively removes all entropic conditions corresponding to $\alpha<0$ in Eq. (\ref{eq:noisy-strict}). Hence, unless $\rho_{\ms S}$ and $\sigma_{\ms S}$ are unitarily equivalent, we have
	\begin{equation}\label{eq:noisy-arb-strict} 
		\rho_{\ms S} \AC{\mathcal{O}_{\mathrm{NO}}}{\emph{\rm arb.}}  \sigma_{\ms S}\quad\Leftrightarrow\quad \begin{cases}
			H_\alpha(\rho_{\ms S}) < H_\alpha(\sigma_{\ms 
				S})\quad \forall \alpha>0,\\
			H_0(\rho_{\ms S}) \leq H_0(\sigma_{\ms S}).
		\end{cases}
	\end{equation}
	\emph{Correlated catalysis.}\label{subsubsec:noisy-corr} The R\'enyi entropies $H_0$ and $H_1$ are the only R\'enyi entropies that are sub-additive, i.e., fulfil $H_{0/1}(\rho_{\ms S_1\ms S_2}) \leq H_{0/1}(\rho_{\ms S_1})+H_{0/1}(\rho_{\ms S_2})$.
	It is therefore easy to see that they cannot decrease under a correlated-catalytic state transformation. 
	Conversely any monotone under correlated catalysis that is additive over tensor-products must also be super-additive, see Sec.~\ref{subsec:corr_cat}.
	Ref.
	\cite{muller_correlating_2018} first showed that the simultaneous increase of both $H_0$ and $H_1 \equiv H$ is indeed the sole 
	criterion for the existence of a correlated catalytic transformation 
	in the case of noisy operations. 
	Specifically, if $\rho_{\ms S}$ and $\sigma_{\ms S}$ are not unitarily equivalent, then
	\begin{equation} 
		\rho_{\ms S} \AC{\mathcal{O}_{\mathrm{NO}}}{\rm corr.} \sigma_{\ms S}\quad \Leftrightarrow\quad \begin{cases}
			H_0(\rho_{\ms S})\leq H_0(\sigma_{\ms S})\quad \text{and}\\
			H(\rho_{\ms S}) < H(\sigma_{\ms S}).
		\end{cases}
	\end{equation} 
	Since any state can be approximated (up to an arbitrary accuracy) by a state with a full rank, then if state-transformations up to arbitrarily small error (on $\ms S$, not $\ms C$) are considered, the only condition that remains is $H(\rho_{\ms S})\leq H(\sigma_{\ms S})$.
	In fact, we will see in Sec.~\ref{subsec:QM} that already correlated-catalytic unitary operations (i.e. without the environment $\ms{E}$) yield the same set of state-transitions \cite{boes_von_2019,wilming_entropy_2020,Wilming2022a}. 
	
	\begin{figure}[h]
		\includegraphics[width=\linewidth]{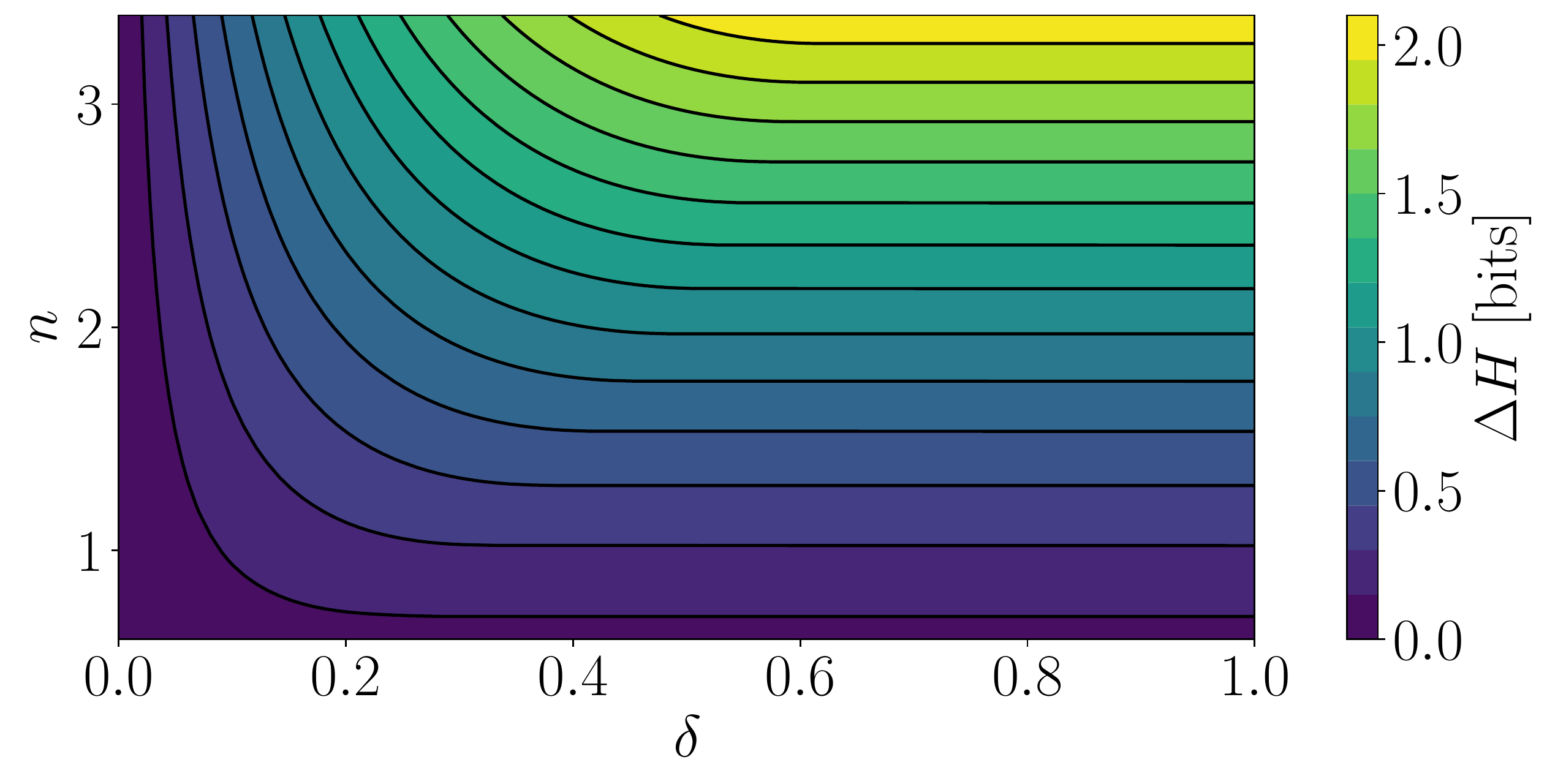}
		\caption{\textbf{Embezzlement in noisy operations}.  The maximal amount of entropy change $\Delta H$ between the state $\omega^{(n)}_{\ms{C}}$ (see Eq. (\ref{eq:van_dam_emb}) and its $\delta$-perturbation for different dimensions $n$.}
		\label{fig:emb_ent}
	\end{figure}
	
	\emph{Approximate catalysis.}\label{subsubsec:noisy-approx} As approximate catalysis depends on the distance measure used, no general statement can be made. If a finite error $\epsilon>0$ in terms of trace distance is allowed on the catalyst ${\ms C}$, then any state transition becomes possible (see Sec.~\ref{subsubsec:approx-cat-construction}). On the other hand, if an error of order $1/\log(d_{\ms C})$ is allowed, with $d_{\ms C}$ being the dimension of the Hilbert space of the catalyst, then $H(\rho_{\ms{S}}) < H(\sigma_{\ms{S}})$ is the only remaining condition for state transitions \cite{brandao_second_2015}. In other words, we essentially obtain the same state transition conditions as for correlated catalysis. The interplay between the error $\epsilon$ on the catalyst, and the resulting simplification of the state transition conditions was analysed in Refs. \cite{ng_limits_2015} and \cite{lipka-bartosik_all_2021}.
	
	\emph{Embezzlement.} Per definition of embezzlement, any state may be reached to arbitrary accuracy via embezzlement. A general construction of embezzling families that achieve this is given in Sec.~\ref{subsubsec:approx-cat-construction}.
	The first example of an embezzling family $\{\omega^{(n)}_{\ms{C}}\}$ was presented in Ref.~\cite{van_dam_universal_2003} in the context of LOCC. In the context of NO it takes the form
	\begin{align}\label{eq:van_dam_emb}
		\omega^{(n)}_{\ms{C}} = \frac{1}{C_n}\sum_{j=1}^n \frac{1}{j} \proj{j}, \quad C_n = \sum_{j=1}^n \frac{1}{j},
	\end{align}
	for some orthonormal basis $\{\ket j\}$, see also Figure~\ref{fig:emb_ent}. Ref.~\cite{george_revisiting_2023} provides a detailed discussion on the optimality of this construction.

	\emph{Infinite-dimensional catalysis.} The same remarks apply here as given in Sec.~\ref{subsec:infinite-dimensional} for the case of pure state LOCC.

	\emph{State-independent catalysis. } To our knowledge, the set of catalytic quantum channels for noisy operations has not been characterized. However, the randomness-utilizing channels from Sec.~\ref{subsec:state-independent} clearly constitute a subset.

	\section{Constructing catalysts}
	\label{sec:constructions}
	The transformation laws for quantum catalysis rarely tell us anything about the state of the catalyst that enables a desired transformation. 
	Therefore, most current results treat catalytic transformations as enigmatic ``black boxes''. While it is theoretically possible to achieve certain transformations with some catalyst, it is unclear how to identify the suitable catalyst state, or determine the {actual} transformation. This limitation severely hinders the practical applicability of catalysis. In this section, we describe the few existing results on the explicit constructions of catalysts.
	
	\subsection{From multi-copy transformations to strict catalysis}
	\label{subsubsec:multi-copy}
	In many resource theories, denoted with $\mc R=(\mc S,\mc O)$, one commonly {adopted assumption} is {the ability} to condition operations on classical randomness. 
	Consider a system $\ms A$ described by a state $\omega_{\ms A}$ diagonal in some fixed basis $\ket{i}_{\ms A}$~--- we view this system as classical due to the distinguished basis. 
	Let $\ms S$ be another system in state $\rho_{\ms S}$ and let $\lbrace\mc F_i\rbrace_i$ be a set of free operations on $\ms S$, i.e. $\mc F_i\in \mc O$ for every $i$. 
	Then the transformation 
	\begin{align}\label{eq:conditional}
		\rho_{\ms S}\otimes\omega_{\ms A}\mapsto \sum_i \mc F_i[\rho_{\ms S}]\otimes \proj{i}\omega_{\ms A}\proj{i}
	\end{align}
	is also free operation on $\ms S\ms A$. 
	The resource theories presented in Section~\ref{sec:rt} all permit such operations {for free}. Note that if $\omega_{\ms A}$ is diagonal in {the basis of} $\ket{i}_{\ms A}$, then system $\ms A$ retains its marginal density operator (but potentially builds up correlations with $\ms S$). 
	If the resource theory in question additionally allows to prepare classical randomness (i.e. any state diagonal in $\ket{i}_{\ms A}$ of arbitrary dimension) for free, and allows to discard subsystems, then the sets $\mathcal{S}$ and $\mathcal{O}$ must be convex.
	
	Now, consider the situation where a state $\rho_{\ms S}$ cannot be transformed to $\sigma_{\ms S}$, but for some sufficiently large $n\in\NN$, the multi-copy state $\rho_{\ms S}^{\otimes n}$ may be transformed to $\sigma_{\ms S}^{\otimes n}$. 
	This can be seen as a form of activation (see section~\ref{subsec:app_cat}) {and was first observed in the context of LOCC in Ref.~\cite{bandyopadhyay_classification_2002}}. A key observation is that if $\mathcal{R}$ allows to permute identical subsystems and condition operations on classical information, then $\rho_{\ms S}$ may be transformed to $\sigma_{\ms S}$ via strict catalysis: 
	\begin{align}\label{eq:multi-copy-implies-catalytic}
		\rho^{\otimes n}_{\ms{S}} \xrightarrow[\mc{O}]{} \sigma_{\ms{S}}^{\otimes n} \quad\Rightarrow \quad \rho_{\ms{S}} \AC{\mc O}{} \sigma_{\ms{S}}.
	\end{align}
	Therefore, {(strict)} catalysis can be used to reduce multi-copy transformations to catalytic transformations.
	{The converse of this statement is false in general, as demonstrated in Ref.~\cite{feng_relation_2006}, see also Refs. \cite{aubrun_stochastic_2009,Gupta2022}.}

	\begin{figure}
		\includegraphics[width=\linewidth]{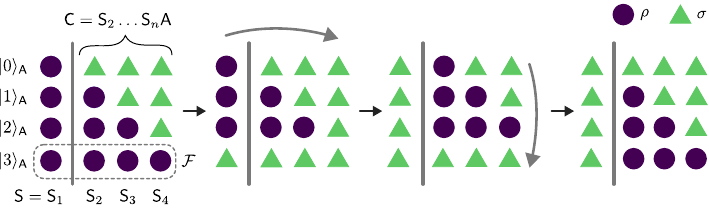}
		\caption{\textbf{Strict catalysis emulating multi-copy transformations.} Illustration of the operation implementing the transformation $\rho_{\ms S}\otimes \omega^{(n)}_{\ms C}(\rho,\sigma) \rightarrow \sigma_{\ms S}\otimes \omega^{(n)}_{\ms C}(\rho,\sigma)$. The dark purple dots signify $\rho$, the lighter green triangles $\sigma$, rows correspond to the states $\ket{i}_{\ms A}$. The free transformation $\mc F$ is applied to the dashed box followed by the cyclic permutations indicated by the grey arrows. }\label{fig:duan-state}
	\end{figure}

	The relation~\eqref{eq:multi-copy-implies-catalytic} is shown by providing a general construction for the required catalyst. This construction first appeared {in the context of LOCC transformations} in Ref.~\cite{duan_multiple-copy_2005}. Consider the following density operator
	\begin{align}
		\omega^{(n)}_{\ms C}(\rho,\sigma) := \frac{1}{n}\sum_{i = 1}^{n} \rho^{\ot i-1} \ot \sigma^{\ot n-i} \ot \dyad{i}_{\ms A},
	\end{align}
	acting on $n-1$ tensor-copies of the Hilbert space  $\mc{H}_{\ms S}$ and an $n$-dimensional Hilbert space $\mc{H}_{\ms A}$ representing classical information.  Thus, $\ms C= \ms S_2\cdots \ms S_n \ms A$, and we will use the convention $\ms S=\ms S_1$ below.
	We may envision the state as corresponding to $n$ distinguishable boxes, where the $i$-th box contains $i-1$ particles in the state $\rho$ and $n-i$ particles in the state $\sigma$. The boxes are distinguished by the classical label $i$, which is unknown.
	
	Suppose we can perform a free transformation $\mathcal{F} \in \mathcal{O}$ such that $\mathcal{F}[\rho^{\ot n}] = \sigma^{\ot n}$ for some $n \in \mathbb{N}$. Then the following protocol transforms a single copy of $\rho_{\ms S}$ into $\sigma_{\ms S}$ using $\omega^{(n)}(\rho,\sigma)$ as a strict catalyst  (see Fig.~\ref{fig:duan-state} for an illustration):
	\begin{enumerate}
		\item Apply a conditional operation as in \eqref{eq:conditional}: Perform $\mc F$ on systems $\ms S_1\cdots \ms S_n$ if $\ms A$ is in state $\proj{n}$ and the identity otherwise. 
		\item Cyclically permute the quantum part $\ms S_1\cdots \ms S_n$ of the system $\ms{SC}$, so that 
			\begin{align*}
				\rho^{\ot i} \ot \sigma^{\ot n-i} \otimes \proj i_{\ms A} \rightarrow \sigma \ot \rho^{\ot i} \ot \sigma^{\ot n-i-1}\otimes \proj i_{\ms A}.
			\end{align*}
			\item Cyclically relabel the classical register $\ms A$ of $\ms{C}$, i.e. map
			$\ket{i} \rightarrow \ket{i+1}$ for $i < n$ and $\ket{n} \rightarrow \ket{1}$.
	\end{enumerate}
	
	As a result, the joint state {of the system $\ms{SC}$} transforms as 
	\begin{align}
		\rho_{\ms S} \otimes \omega^{(n)}_{\ms C}(\rho,\sigma) \rightarrow \sigma_{\ms S}\otimes \omega^{(n)}_{\ms C}(\rho,\sigma).
		\label{eq:subroutine}
	\end{align}
	The dimension of the catalyst $\omega_{\ms C}^{(n)}(\rho,\sigma)$ increases exponentially with the number of necessary copies $n$ for multi-copy activation. Furthermore, the catalyst must also be finely tuned with respect to the initial and target states on ${\ms S}$. Despite these disadvantages, the construction still serves as a key technical tool in various proofs of catalysis.

	\subsection{From asymptotic transformations to correlated catalysis}
	\label{subsec:partial-order-regularization}

	The main problem addressed in any resource theory is determining whether there exists a free operation $\mathcal{F}\in \mathcal{O}$ such that $\mathcal{F}(\rho) = \sigma$.  This is generally a difficult problem that has to be addressed on a case-by-case basis. In general, there exists no systematic way of characterizing which transitions are allowed in a given resource theory. One potential approach to better understand what is achievable within a particular resource theory is to employ reasonable relaxations.

	The first relaxation we discuss is allowing for a small error in the final system state, {which is also called smoothing}. Specifically, given a fixed error $\epsilon$, we allow the channel to produce any state $\sigma_{\epsilon} = \mathcal{F}[\rho]$ $\epsilon$-close in trace distance to the target state $\sigma$. 
	Due to the definition of trace distance, all observables of $\sigma_{\epsilon}$ will not differ much from $\sigma$. 
	The second relaxation is {assuming an i.i.d. distribution}, that is, instead of considering a direct transformation from $\rho$ to $\sigma$, we ask if it is possible to convert $\rho^{\ot n}$ into $\sigma^{\ot {\lfloor r n \rfloor}}$ for some $r > 0$ and a large enough $n\in \mathbb{N}$. One key insight is that for large $n$, the multi-copy state $\rho^{\ot n}$ {when smoothed,} is almost indistinguishable from the maximally-mixed state spanned over its typical subspace \cite{Wilde2009}. 
	Therefore, characterizing state transformations between such states is much easier than solving the corresponding single-shot problem. 
	The standard quantity studied in this smooth, asymptotic limit is the conversion rate, which for any {set of free operations $\mc O$} can be defined as
	\begin{align}
		R_{\epsilon}^{n}(\rho, \sigma) := \sup \Big\{r \Big| \inf_{{\mc F\in \mc O}} {\Delta(\mc F[\rho^{\ot n}], \sigma^{\ot \lfloor n r \rfloor}) \leq \epsilon}\Big\}.
	\end{align}
	The rate $R_{\epsilon}^{n}(\rho, \sigma)$ represents how many copies of the target state can be (approximately) generated per copy of the initial state, when using $n$ copies of the initial state. Of particular importance is the asymptotic rate $R_{\epsilon}^{n}$ in the limit of infinite copies, i.e.
	\begin{align}
		R^{\infty}(\rho, \sigma) = \lim_{n \rightarrow \infty} R_{\epsilon}^{n}(\rho, \sigma).
	\end{align}
	Importantly, $R^{\infty}(\rho,\sigma)$ does not depend on $\epsilon$ as long as $\epsilon>0$. Note that the asymptotic rate $R_{\epsilon}^{\infty}$ allows to fully characterizing the approximate conversion problem in the asymptotic limit. 
	Indeed, when $R^{\infty}(\rho,\sigma) > 1$,  then there exists an $n\in\NN$ and a free operation mapping $\rho^{\ot n}$ into a state that is $\epsilon$-close to $\sigma^{\ot n}$. 
	Conversely, when $R^{\infty}(\rho,\sigma) < 1$, we are guaranteed that there is no free operation that can ever map $n$ copies of $\rho$ into $n$ copies of $\sigma$ for an arbitrary small error $\epsilon$. In this sense, the partial order of states becomes regularized, or simplified, in the approximate i.i.d. limit. All possible state transformations can now be characterized by a single quantity, the rate, as opposed to the single-shot case where they are usually characterized by a set of monotones.
	
	It was first observed in Ref.~\cite{shiraishi_quantum_2021} in the context of thermodynamics that asymptotic state transformations and correlated-catalytic state transformations are closely related using a generalization of the construction in Ref.~\cite{duan_multiple-copy_2005} presented in Sec.~\ref{subsubsec:multi-copy}.
	More specifically, the following lemma can be proven, which can be seen as a generalized summary of the results in Refs.~\cite{shiraishi_quantum_2021,wilming_entropy_2020,Kondra2021,lipka-bartosik_catalytic_2021,Char2021,Takagi2022,Datta2022a}.
	
	\begin{lemma}\label{lemma:partial-order-regularization}
		Let $\rho_{\ms S}$ and $\sigma_{\ms S}$ be two density operators and $\mc R=(\mc S,\mc O)$ a resource theory that allows for classical conditioning and permutation of identical subsystems. {Then
			\begin{align}
				R_\epsilon^{n}(\rho_{\ms S},\sigma_{\ms S}) \geq 1\quad \Rightarrow \quad \rho_{\ms S} \AC{\mathcal{O}}{\rm corr.} \sigma_{\ms S}^\epsilon
			\end{align}
			with $\Delta(\sigma_{\ms S}^\epsilon,\sigma_{\ms S})\leq \epsilon$. 
			In fact, the catalyst $\omega_{\ms C}$ can be chosen such that $\Delta(\eta_{\ms S\ms C},\sigma_{\ms S}\ot \omega_{\ms C})\leq 2\epsilon$, where $\eta_{\ms S \ms C}$ is the final state on $\ms S \ms C$.}
	\end{lemma}
	
	The lemma shows that whenever it is possible to transform $n$ copies of $\rho_{\ms S}$ into $n$ copies of $\sigma_{\ms S}$ with a given accuracy, there also exists a correlated catalytic transformation that converts a single copy of $\rho_{\ms S}$ into $\sigma_{\ms S}$ with the same accuracy. 
	When the transformation can be performed without error, i.e. when $\epsilon = 0$, then the mapping becomes strictly catalytic and no correlations are established between catalyst and system, recovering the 
	setting of section~\ref{subsubsec:multi-copy}.
	
	Let $\mc F$ be a free operation such that $\Delta(\mc F[\rho_{\ms S}^{\otimes n}],\sigma_{\ms S}^{\otimes n})\leq \epsilon$, which exists whenever $R_\epsilon^{n}(\rho_{\ms S},\sigma_{\ms S})\geq 1$. 
	To prove Lemma~\ref{lemma:partial-order-regularization} by construction, one follows the same strategy as in section~\ref{subsubsec:multi-copy} but replaces the {states $\sigma_{\ms S}^{n-i}$ in the definition of $\omega^{(n)}_{\ms C}$ with the reduced state $\tr_{> n - i}[\mc F[\rho_{\ms S}^{\otimes n}]]$ on the first $n-i$} marginals. The result then follows by making use of the data-processing inequality for the trace-distance. 
	
	In certain resource theories (e.g. majorization-based resource theories), it is known that $R^\infty$ can be directly computed and expressed using an appropriate resource monotone. 
	Let $\mathcal{R}$ be a resource theory whose asymptotic rate is given by
	\begin{equation}
		\label{eq:asymp_rate_entr}
		R^{\infty}(\rho_{\ms S},\sigma_{\ms S}) = \frac{f(\rho_{\ms S})}{f(\sigma_{\ms S})}\qquad \text{for any}\,\, \epsilon > 0,
	\end{equation}
	where $f$ is a monotone of the resource theory. 
	For example, as we will see in Sec.~\ref{subsubsec:entanglement-distillation}, $f$ corresponds to the entanglement entropy in the context of pure state entanglement.
	Then Lemma~\ref{lemma:partial-order-regularization} shows that $f(\rho_{\ms S}) > f(\sigma_{\ms S})$ implies that a correlated-catalytic state transformation is possible. 
	{As far as we know, it} is currently unknown if the possibility of a correlated-catalytic state transition also implies that the asymptotic rate fulfills $R^{\infty}(\rho_{\ms S},\sigma_{\ms S})\geq 1$. 
	
	The above result not only shows that correlated catalytic state transformations are strictly more powerful than the corresponding {single-shot} free operations, but also that they sometimes have a simple mathematical characterization in terms of the asymptotic resource monotone determining the rate $R^{\infty}$. 
	In Sec.~\ref{sec:applications} we will review and describe applications of the above results in different physical settings. 
	
	\subsection{Constructions for approximate catalysts and embezzlers}
	\label{subsubsec:approx-cat-construction}
	In section \ref{subsec:partial-order-regularization} we showed that correlated catalysis enables state transitions that are otherwise possible only in an approximate sense in the asymptotic limit.
	In this section we describe a generic construction for approximate catalysis that allows to map between \emph{arbitrary states}, whenever the set of free operations $\mc O$ allows to permute identical subsystems. 
	We then generalize it to construct an embezzler for state transitions on a fixed finite-dimensional system $\ms S$. 
	The construction explicitly demonstrates that a fixed error in trace-distance is not sensitive enough to prevent embezzlement. 
	We follow Ref.~\cite{Leung2013coherent}, but the construction is closely related to the construction of strict catalysts to emulate multi-copy transformations presented in section~\ref{subsubsec:multi-copy}. 
	
	Suppose we wish to implement the state transition $\rho_{\ms S}\rightarrow \sigma_{\ms S}$ using a catalyst $\omega_C$ that is perturbed by at most $\epsilon$ in trace-distance.  
	Consider the states
	\begin{align}\label{eq:universal_catalyst}
		\omega_{\ms C} = \frac{1}{n-1} \sum_{k=1}^{n-1} \rho_{\ms S}^{\otimes k} \otimes {\sigma_{\ms S}}^{\otimes n-k},\\
		\omega'_{\ms C} = \frac{1}{n-1}\sum_{k=2}^n \rho_{\ms S}^{\otimes k} \otimes {\sigma_{\ms S}}^{\otimes n-k}.
	\end{align}
	By applying a cyclic permutation of subsystems represented by the unitary $U_\pi$, we find $U_\pi(\rho_{\ms S} \otimes \omega_{\ms C}) U_\pi^\dagger = \sigma_{\ms S}\otimes \omega'_{\ms C}$. 
	Since
	\begin{equation}\label{eq:omegaprime_and_omega}
		\omega'_{\ms C} - \omega_{\ms C} = \frac{1}{n-1}\left[\rho_{\ms S}^{\otimes n} - \rho_{\ms S} \otimes {\sigma_{\ms S}}^{\otimes n-1}\right]
	\end{equation}
	we find that $\Delta(\omega_{\ms C},\omega'_{\ms C})\leq 1/(n-1)$. 
	Thus, for fixed $\epsilon$ and sufficiently large $n$ any state transition is possible using the given construction for approximate catalysts. 
	
	{Note that the construction in Eq.~\eqref{eq:omegaprime_and_omega} currently depends on $\rho_{\ms{S}}$ and $\sigma_{\ms{S}}$. However,} using it, we can construct a universal embezzler for any fixed system $\ms S$ with finite-dimensional Hilbert space:
	Since the set of density operators on a finite-dimensional Hilbert-space is compact, for any $\delta>0$ there exists a finite collection of density operators $\{\chi^{(i)}_{\ms S}\}_{i=1}^N$ such that
	for any density operator $\rho_{\ms S}$ we have $\Delta(\rho_{\ms S},\chi^{(i)}_{\ms S})\leq \delta/4$ for some $1\leq i\leq N$. {The size of this set, $N$, depends on $\delta$.}
	For any $1 \leq i,j\leq N$, let $\omega^{(i,j)}_{\ms C_{ij}}$ be the approximate catalyst above implementing the state transition $\chi^{(i)}_{\ms S}\rightarrow \chi^{(j)}_{\ms S}$. 
	Then the embezzler consisting of $N^2$ copies $\ms C_{ij}$, with state
	\begin{align}
		\omega^{(\delta)}_{\ms C}= \bigotimes_{i,j=1}^N \omega^{(i,j)}_{\ms C_{ij}}
	\end{align}
	is a $\delta$-embezzler: 	Suppose we wish to transform $\rho_{\ms S}$ to $\sigma_{\ms S}$. 
	Then {we need only to identify $\chi^{(i)}_{\ms S},\chi^{(j)}_{\ms S}$} such that
	\begin{equation}
		\Delta(\rho_{\ms S},\chi^{(i)}_{\ms S})\leq \delta/4,\qquad\Delta(\sigma_{\ms S},\chi^{(j)}_{\ms S})\leq \delta/4.
	\end{equation}
	Now, by using $	\omega^{(\delta)}_{\ms C}$ as a catalyst, there is a free operation $\mc F$ consisting only of permutations involving $\ms S$ and the sub-catalyst $\ms C_{ij}$ such that
	\begin{align}
		\Delta(\mc F[\rho_{\ms S}\otimes \omega^{(\delta)}_{\ms C}],\sigma_{\ms S}\otimes \omega_{\ms C}^{(\delta)}) \leq \delta.
	\end{align}
	A related construction was used in  \cite{datta2022entanglement} to construct a correlated catalyst for all state transitions in LOCC between pure states where the entanglement entropy is non-increasing. In the context of LOCC, the structure of embezzling families has been studied in much more detail and can be characterized comprehensively, see \cite{van_dam_universal_2003,leung_characteristics_2014,zanoni_complete_2023}.
	
	We close the section with a word of caution on nomenclature: 
	According to our classification of catalysis the states $\omega_{\ms C}$ constructed above are approximate catalysts, and we reserve  the term universal embezzler for $\omega^{(\delta)}_{\ms C}$, because only the latter allow to implement arbitrary state transitions with {vanishing error (as $N\rightarrow\infty$)}. 
	Yet, sometimes already the states $\omega_{\ms C}$ are referred to as universal embezzlers in the literature, with the idea that the notion of "embezzlement" refers to the fact that approximate catalysis in terms of a fixed trace-distance allows for arbitrary state transitions. 
	
	In Sec.~\ref{subsubsec:non-continuity} we discuss how the above construction can be used to constrain continuity properties of resource monotones. 
	
	\subsection{Numerical construction of catalysts}
	In this section, we discuss a simple, yet useful, method for computing the exact state of the catalyst. More specifically, for a fixed quantum channel 
	$\mathcal{E}_\ms{SC}$ and a fixed input state on ${\ms{S}}$, this method determines the corresponding state of the catalyst, namely, the fixed point of the described quantum channel on $\ms{C}$. To the best of our knowledge, this approach was first used in the context of catalysis in Ref.~\cite{boes_by-passing_2020}.
	
	Let $\mathcal{E} \in \mathcal{L}(\mathcal{H}_{\ms S} \ot \mathcal{H}_{\ms C})$ be an arbitrary quantum channel acting on $\ms{SC}$, and let $\rho_{\ms S} \in \mathcal{D}(\mathcal{H}_{\ms S})$ be a fixed density operator. The effective channel acting on  $\ms{C}$ can be written as
	\begin{align}\label{eq:fixed-point}
		\mc E_{\ms C}(\cdot) = \Tr_{\mathsf S} \mc E[\rho_{\ms{S}} \ot (\cdot)_{\ms{C}}].
	\end{align} 
	
	Suppose we want to implement a correlated-catalytic transformation discussed in Sec. \ref{subsec:corr_cat}. This means that we have to find a density operator $\omega_{\ms{C}}$ which satisfies $\mathcal{E}_{\ms C}(\omega_{\ms C}) = \omega_{\ms C}$. This is equivalent to finding fixed points of the map $\mathcal{E}_{\ms C}$, {which exist due to Brouwer's fixed point theorem} \cite{wolf2012quantum}. This problem can be formulated as a semidefinite program (SDP), i.e.
	\begin{align} \label{eq:sdp_cc}
		\min_{X} \hspace{15pt} &0 \\
		\text{subject to} \hspace{15pt}& \mc{E}_{\ms{C}} [X] = X, \nonumber \\
		&X \geq 0, \hspace{5pt} \Tr X = 1.\nonumber
	\end{align}
	The solution of the above problem is a positive semidefinite operator $X$, which can be interpreted as the quantum state of the correlated-catalyst, i.e. $\omega_{\ms{C}} = X$. The main advantage of semidefinite programs stems from the fact that they can be efficiently solved numerically \cite{boyd2004convex}, for example using the modelling language CVX \cite{cvx}. Moreover, in certain cases, some essential features of the solution can be even inferred analytically, see e.g. \cite{Cavalcanti_2016,Napoli_2016,bavaresco2021strict}.
	To give an example in the context of catalysis, notice that the output state $\sigma_{\ms S} = \Tr_{\ms C}\mc E[\rho_{\ms S}\otimes \omega_{\ms C}]$, is, a priori, arbitrary. 
	However, if some additional information about $\mc E$ is available, it may be possible to deduce some essential features of $\sigma_{\ms S}$ even analytically. 
	In Ref.~\cite{boes_by-passing_2020} this technique was successfully used in the context of work extraction from multipartite systems, {to show that correlating catalysis allows for the surpassing of stringent conditions imposed by Jarzynski's inequality}. 
	
	The formulation in Eq.~(\ref{eq:sdp_cc}) allows to easily add additional constraints on the catalyst state, as long as they can be formulated in terms of semidefinite constraints. This freedom can be used to investigate other regimes of catalysis. To see this, consider a more general version of the problem in Eq.~(\ref{eq:sdp_cc}):
	\begin{align}
		\min_{X} \hspace{15pt} & F(X) \label{eq:general_SDP_construction}\\
		\text{subject to} \hspace{15pt}& \norm{\mc{E}_{\ms{C}}(X) - X}_1 \leq \epsilon, \nonumber \\
		&X \geq 0, \hspace{5pt} \Tr X = 1.\nonumber 
	\end{align}
	where $F$ is a linear function of $X$ and $\epsilon > 0$. The above problem is still an SDP, as can be seen by using a standard reformulation of the trace norm $\norm{\cdot}_1$ in terms of a semidefinite constraint \cite{vandenberghe1996semidefinite}. Eq.~\eqref{eq:general_SDP_construction} allows for determining approximate catalysts discussed in Sec. \ref{subsec:app_cat}, by taking $F(X) = 0$ and fixing $\epsilon > 0$ allows for finding the state of an approximate catalyst.

	Interestingly, one can also consider non-linear functions $F(\cdot)$, as long as they can be itself expressible by semidefinite programs. An example of such a function is the trace distance, i.e. $\norm{\cdot}_1$, see Ref.~\cite{watrous2009semidefinite}. Taking $F(X) = \norm{\mc{E}(\rho_{\ms{S}} \ot X_{\ms{C}}) - \sigma_{\ms{S}} \ot X_{\ms{C}}}_1$ allows to determine the state of the catalyst which ends up being the least correlated with the system $\ms{S}$.  When the SDP problem defined according to this recipe achieves $F(X^*) = 0$, the resulting optimal variable $X^*$ corresponds to a strict catalyst discussed in Sec. \ref{subsec:strict}. 
	
	At this point, we should note that from a resource-theoretic perspective, the method above is very limited, as it requires $\mathcal{E}_{\ms{SC}}$ and $\rho_{\ms{S}}$ to be fixed. 
	Therefore, it does not allow for finding the catalyst that would enable a given state transformation on the system. However, in many realistic applications, the available joint quantum channel $\mathcal{E}_{\ms{SC}}$ is either fixed, or can be parametrized using a small number of parameters. This happens, e.g., in the experimentally-relevant models of light-matter interactions like the Jaynes-Cummings \cite{Jaynes1963} or Dicke models \cite{hepp1973superradiant}. Moreover, due to experimental capabilities, the set of states $\rho_{\ms{S}}$ prepared is usually also restricted and efficiently parametrized. In such special cases, the method of determining the catalyst as described above performs reasonably well. This approach was used, e.g. in Ref.~\cite{deOliveira2023} to demonstrate the effect of (catalytic) activation of Wigner negativity (see also Sec. \ref{subsec:wigner_neg}) and in Ref.~\cite{lipka2023operational} to show a catalytic enhancement in cooling or heating using a single-mode optical cavity. 
	
	Finally, we would like to note that beyond the regime of analytical constructions and convex optimization, the avenue for determining whether $\rho\rightarrow \sigma$ is possible using a certain type of catalyst remains a space for much potential exploration. 
	For example, by training a neural network to learn majorization, machine learning techniques have been implemented \cite{acacio2022analysis} to see if they can successfully identify whether a transition may be achieved catalytically. 
	It is anticipated that similar techniques could be used to learn how to identify a catalyst that activates the desired transformation.

	\section{Applications of catalysis}
	\label{sec:applications}
	
	\subsection{Unitary quantum mechanics}
	\label{subsec:QM}
	One of the simplest resource theories one could imagine is given by unitary quantum mechanics $\mc R_{\mathrm{QM}} = (\mc S_{\mathrm{QM}}, \mc O_{\mathrm{QM}})$. The set of free states $\mc S_{\mathrm{QM}}$ is empty (every state is resourceful) and the free operations $\mc O_{\mathrm{QM}}$ are simply all unitary transformations. That means that all free operations are reversible and randomness (even classical) is costly. In particular, the resource theory is not convex, and the tools of Section~\ref{subsec:partial-order-regularization} do not directly apply. 
	While at first sight physically maybe less plausible, $\mc R_{\mathrm{QM}}$ serves as an interesting testbed to distinguish different types of catalysis.
	First, the law for state transitions is very simple: Two states $\rho_{\ms S}$ and $\sigma_{\ms S}$ on the same Hilbert space can be converted into each other if and only if their eigenvalues (including multiplicities) are identical.
	Second, Lemma~\ref{lem:fundamental_nogo} implies that strict catalysis is useless in this scenario. 
	On the other hand, correlated catalysis was shown to enlarge the set of states that can be reached in this scenario. More specifically, Refs.~\cite{boes_von_2019,wilming_entropy_2020,Wilming2022a} showed that the set of achievable states is described by a statement akin to the second law of thermodynamics.
	\begin{theorem} \label{thm:CEC}
		Let $\rho_{\ms S}$ and $\sigma_{\ms S}$ be density matrices on a finite-dimensional Hilbert space that are not unitarily equivalent. Then the following two statements are equivalent:
		\begin{enumerate}
			\item There exists a finite-dimensional density operator $\sigma_{\ms C}$ and a unitary operator $U$ on $\ms S\ms C$ such that
			\begin{align}\label{eq:unitary-catalysis}
				\tr_{\ms C}[U\rho_{\ms S}\otimes \sigma_{\ms C} U^\dagger] &= \sigma_{\ms S}\\
				\tr_{\ms S}[U\rho_{\ms S}\otimes\sigma_{\ms C} U^\dagger] &=\sigma_{\ms C}. \label{eq:unitary-catalysis2}
			\end{align}
			\item $H(\rho_{\ms S})< H(\sigma_{\ms S})$ and $\mathrm{rank}(\rho_{\ms S})\leq \mathrm{rank}(\sigma_{\ms S})$. 
		\end{enumerate}
	\end{theorem}
	The same results hold in the classical case where density matrices are replaced by probability vectors and unitary transformations by permutations.
	By continuity of von Neumann entropy, the theorem implies that any state $\sigma_{\ms S}$ can be reached from $\rho_{\ms S}$ to arbitrary accuracy using correlated catalysis if and only if $H(\sigma_{\ms S})\geq H(\rho_{\ms S})$.
	Thus we find a complete operational characterization of von Neumann entropy without any reference to thermodynamics, entanglement or information theory. 
	Previously, \cite{muller_correlating_2018} showed that Thm.~\ref{thm:CEC} holds if one allows for general mixed-unitary quantum channels instead of reversible unitary channels.
	\cite{boes_von_2019} used this to generalize a classic result of \cite{Aczel1974} for Shannon entropy to the quantum case:  
	If a continuous function on density matrices is a) invariant under unitary transformations, b) additive over tensor-products, and c) sub-additive, then it is given by von~Neumann entropy up to rescaling and (dimension-dependent) shift of the origin. 
	In particular, if appropriate normalization for maximally mixed states and pure states is required, this singles out von Neumann entropy uniquely. 
	
	In passing, we remark that the conditions in Eq.~\eqref{eq:unitary-catalysis} and ~\eqref{eq:unitary-catalysis2} also appear in Deutsch's analysis of quantum mechanics near closed time-like curves \cite{Deutsch1991}, where the condition that the state on $\ms C$ remains the same is a consistency requirement to prevent certain logical paradoxes.
	Deutsch already observed $H(\sigma_{\ms S})\geq H(\rho_{\ms S})$ and that \emph{some} $\sigma_{\ms C}$ and $\sigma_{\ms S}$ fulfilling Eq.~\eqref{eq:unitary-catalysis} and \eqref{eq:unitary-catalysis2} always exist once $\rho_{\ms S}$ and $U$ are specified. 
	
	Where does the entropy increase on $\ms S$ come from? It corresponds precisely to the correlations build up between system $\ms S$ and catalyst $\ms C$ when measured in terms of mutual information:
	\begin{align}
		H(\sigma_{\ms S}) - H(\rho_{\ms S}) = I({\ms S}:{\ms C})_{U\rho_{\ms S}\otimes \sigma_{\ms C} U^\dagger}.
	\end{align}
	
	How large does the system $\ms C$ need to be? When $\rho_{\ms S}\succ \sigma_{\ms S}$, the catalyst $\ms C$ can be chosen to be a maximally mixed state of dimension at most $\lceil \sqrt{d_{\ms S}}\rceil$ \cite{boes_catalytic_2018}, corresponding to a regime where the entropy increase on $\ms S$ can be considered large. However, as $H(\sigma_{\ms S})-H(\rho_{\ms S})$ becomes small, it can be shown that there exists catalytic state transitions requiring arbitrarily large catalyst dimensions \cite{boes_variance_2020}. 
	In particular if $H(\rho_{\ms S})=H(\sigma_{\ms S})$ either the two states are unitarily equivalent or there does not exist a finite-dimensional catalyst. 
	It is currently unknown whether an infinite-dimensional catalyst can be used in this case.  
	
	In Section~\ref{subsubsec:cooling} we discuss some applications of Theorem~\ref{thm:CEC} in the context of quantum thermodynamics.
	{\cite{gallego_what_2018} further used the previous result in \cite{muller_correlating_2018} to show that, in principle, a single catalyst may be used to bring arbitrary many many-body systems that are initially in equilibrium permanently out of equilibrium. }

	
	\subsection{Entanglement theory}
	\label{subsec:entanglement}
	Entanglement is perhaps the most striking manifestation of the non-classical nature of quantum mechanics \cite{EPR1935}. After its first experimental demonstration \cite{Aspect1982}, it was realized that entanglement may be used as a resource enabling new types of protocols, including new communication tasks \cite{Bennet1993} and unconditionally secure cryptographic schemes \cite{BB84,ekert1991quantum}.

	The most common approach to study entanglement is the so-called ``distant-lab'' paradigm, briefly introduced in Sec.~\ref{subsubsec:locc}. Consider the scenario where a multipartite quantum system is distributed to spatially-separated parties, who are restricted to act locally on their respective subsystems by performing local quantum operations. 
	A common assumption is that the parties can exchange classical information in order to enhance their measurement strategies. 
	Quantum operations implemented in this manner are known as local operations and classical communications (LOCC). 
	When communication is not allowed, e.g. in Bell non-locality, the relevant set of operations involves local operations and shared randomness (LOSR). 
	In this paradigm parties cannot communicate, however, they are allowed to share classical randomness. This captures the natural restrictions encountered in Bell-like experiments and non-local games \cite{Buscemi2012}. 
	In both cases, the set of free states $\mc S_{\mathrm{LOCC/LOSR}}$ is given by classically correlated states. 
	In the bipartite scenario involving Alice $({\ms A})$ and Bob $({\ms B})$, they take the form
	\begin{align}
		\rho_{\ms A\ms B} = \sum_i p_i\, \rho_{\ms{A}}^i \ot \rho_{\ms{B}}^i
	\end{align} 
	for some probability distribution $p_i$. Since LOSR operation is an LOCC protocol without communication, the LOSR set is a strict subset of LOCC. By definition, neither LOSR nor LOCC operations can create entanglement. 
	Both of these classes have a fairly intuitive physical description, however, they are known for being notoriously
	hard to characterize mathematically \cite{chitambar2014everything}.
	In this section, we review the role of catalysis LOCC (Sec.~\ref{subsubsec:locc-proper}) and LOSR (Sec.~\ref{subsubsec:LOSR}) transformations.
	
	\subsubsection{Local operations and classical communication (LOCC)}
	\label{subsubsec:locc-proper}
	A physical operation is called a $1$-way LOCC operation from $\ms{A}$ to $\ms{B}$
	when it can be implemented by applying arbitrary local quantum {operations} by Alice $(\ms{A})$ and Bob $(\ms{B})$, and a single round of classical communication from $\ms{A}$ to $\ms{B}$. 
	More generally, $n$-LOCC operations involve exchanging $n$ rounds of classical communication.
	The class of LOCC operations $\locc(\ms A:\ms B)$ between Alice and Bob is the union of all $n$-LOCC operations.
	It is known that $(n+1)$-LOCC operations are strictly more powerful than $n$-LOCC \cite{chitambar2014everything}.
	In general, the outcome of an LOCC operation on an initial state $\rho$ consists of classical measurement outcomes $x$, and the associated conditional quantum states $\sigma_x$.
	In other words, an LOCC operation corresponds to a \emph{quantum instrument} $\{\mc F_x\}$, where $\mc F_x$ are completely positive maps such that $\mc F:= \sum_x \mc F_x$ is a quantum channel and $\sigma_x = \mc F_x[\rho]/p_x$ with $p_x := \Tr[\mc F_x[\rho]]$.
	We say that an LOCC operation converts $\rho$ to $\sigma$ (deterministically) if $\sigma_x = \sigma$ for all $x$. 
	The success probability  to convert $\rho$ to $\sigma$ by the LOCC operation $\{\mc F_x\}$ is the total probability of the events $x$ that have $\sigma$ as outcome:
	\begin{align}
		p_{\text{succ}}(\rho\rightarrow \sigma|\{\mc F_x\}) := \sum_{x : \sigma_x=\sigma} p_x.  
	\end{align}
	The optimal success probability $p_{\text{succ}}(\rho\rightarrow \sigma)$ to convert $\rho$ into $\sigma$ via LOCC is obtained by maximizing the success probability over all LOCC instruments $\{\mc F_x\}$. 
	A subtlety easy to miss is the following: LOCC allows for discarding subsystems and in particular discarding subsystems storing classical measurement records, thereby effectively averaging over these outcomes. 
	Thus if $\sigma = \sum_x p_x \sigma_x$ is the average outcome of a given LOCC protocol acting on $\rho$, then there is a different LOCC protocol (represented by the quantum channel $\mc F = \sum_x  \mc F_x$) that converts $\rho$ into $\sigma$ deterministically. If $\sigma$ is pure, then all the $\sigma_x$ already need to be identical to $\sigma$.
	
	Interestingly, when restricted to bipartite pure states, then the most general LOCC transformation requires only one-way communication \cite{lo_concentrating_2001}. This is because any bipartite pure state admits a Schmidt decomposition, which is symmetric (up to local unitaries) under exchange of parts $\ms{A}$ and $\ms{B}$. Therefore, without loss of generality, an LOCC protocol involving pure states can always be executed by a single measurement on one party, followed by a local unitary operation on the second party conditioned on the measurement outcome.  Consequently, while characterizing general LOCC transformations for generally mixed states is difficult, for bipartite pure states a simple characterization exists and is known as Nielsen's theorem \cite{nielsen_conditions_1999}.
	\begin{theorem}[Nielsen's theorem]
		\label{thm:nielsen}
		State $\ket{\psi}$ can be converted into state $\ket{\phi}$ by means of LOCC if and only if
		\begin{align}
			\label{eq:ent_nielsen_majorization}
			\eig(\psi_{\ms A}) \prec \eig(\phi_{\ms A}),
		\end{align}
		where $\eig(\psi_{\ms A})$ is the vector of Schmidt coefficients of $\ket \psi$, i.e., the eigenvalues of the reduced state $\psi_{\ms A} := \Tr_{\ms{B}}\left[\dyad{\psi}_{\ms{AB}}\right]$. 
	\end{theorem}
	
	Nielsen's theorem provides an elegant connection between entanglement transformations and the theory of majorization \cite{marshall_inequalities_2011}. Moreover, the conditions given in Eq. (\ref{eq:ent_nielsen_majorization}) are easy to check numerically, and therefore provide a powerful tool for determining when one pure bipartite state can be converted into another via LOCC.
	Nielsen's theorem directly implies that there exist incomparable states, in the sense that neither $ \psi$ nor $ \phi $ can be directly transformed into another using LOCC. 
	To address this interconversion barrier, Vidal generalized Nielsen's work by characterizing probabilistic transformations in the LOCC framework \cite{vidal_entanglement_1999}.
	
	\begin{theorem}[Vidal's theorem]
		State $\ket{\psi}$ can be (conclusively) converted into $\ket{\phi}$ with probability $\mu$ by means of LOCC if and only if
		\begin{align}
			\eig(\psi_{\ms A}) \prec^w \mu \eig(\phi_{\ms A}).
		\end{align}
	\end{theorem}
	The symbol $\prec^w$ denotes a general form of majorization called \emph{weak} majorization \cite{bhatia2013matrix} to compare unnormalized distributions. Vidal also presented an optimal protocol demonstrating that the transformation from $\ket{\psi}$ to $\ket{\phi}$ is always possible probabilistically, with success probability
	\begin{align}
		\label{eq:ent_vidal_conds}
		p_{\text{succ}}(\ket{\psi} \rightarrow \ket{\phi}) = \min_{1 \leq k \leq d} \frac{1 - \mathcal{L}_k(\psi)}{1 - \mathcal{L}_k(\phi)},
	\end{align}
	where $\mathcal{L}_k(\psi) := \sum_{i=1}^{k-1}\lambda^{\downarrow}(\psi_A)_i$ with $\lambda(\psi_A)_0 \equiv 0$. 
	See Example \ref{example:probabilistic} for an application.
	
	\begin{example} Consider the following bipartite pure states
		\label{example:probabilistic}
		\begin{align}
			\!\!\ket{\psi} &= \sqrt{0.4} \ket{00} + \sqrt{0.4} \ket{11} + \sqrt{0.1} \ket{22} + \sqrt{0.1} \ket{33}\!,\! \nonumber \\
			\label{eq:ent_state2}
			\!\!\ket{\phi} &= \sqrt{0.5} \ket{00} + \sqrt{0.25}\ket{11} + \sqrt{0.25} \ket{22}.
		\end{align}
		The eigenvalues of $\psi_{\ms A}$ and $\phi_{\ms A}$ are given by
		\begin{align}
			\eig(\psi_{\ms A}) &= [0.4, 0.4, 0.1, 0.1],\\
			\eig(\phi_{\ms A}) &= [0.5, 0.25, 0.25, 0].
		\end{align}
		Since $\eig(\psi_{\ms A}) \nprec  \eig (\phi_{\ms A})$, Nielsen's theorem precludes the existence of a deterministic LOCC transformation $\ket{\psi} \rightarrow \ket{\phi}$. Still, if one attempts to transform $\ket{\psi}$ into $\ket{\phi}$ probabilistically, Vidal's theorem from Eq. (\ref{eq:ent_vidal_conds}) implies that it is possible to achieve it with probability $p_{\text{succ}}(\ket{\psi} \rightarrow \ket{\phi}) = 4/5$.
	\end{example}
	
	It was later observed in Ref.~\cite{jonathan_entanglement-assisted_1999} that a catalytic version of LOCC can be considered, where Alice and Bob are allowed to use pre-shared entanglement in the form of auxiliary bipartite pure states. Importantly, one can do so in a way such that the auxiliary state is returned exactly in its initial state after the transformation. 
	To demonstrate this, Ref.~\cite{jonathan_entanglement-assisted_1999} proposed a particular example of this type of transformation, see Example \ref{example:jonathan_plenio} for details. 
	
	\begin{example}
		\label{example:jonathan_plenio}
		Consider the two states defined in Eq. (\ref{eq:ent_state2}) and a two-qubit catalyst prepared in the state 
		\begin{align}
			\ket{\omega} = \sqrt{0.6} \ket{00} + \sqrt{0.4}\ket{11}.
		\end{align}
		Nielsen's theorem applied to the transformation $\ket{\psi} \ot \ket{\omega} \rightarrow \ket{\phi} \ot \ket{\omega}$ implies that the transition can be realized deterministically via LOCC. 
	\end{example}
	
	The example demonstrates that an appropriately chosen catalyst can increase the success probability of a transformation, even from $ p<1 $ to $ p=1 $. This type of transformation was dubbed by Ref.~\cite{jonathan_entanglement-assisted_1999} as entanglement-assisted LOCC (ELOCC), and is the first example of strict catalysis in LOCC. The authors also show that two bipartite pure states $\ket{\psi}$ and $\ket{\phi}$ that are inter-convertible via ELOCC, i.e. each state is convertible to the other, are already equivalent up to local unitaries. In other words, strict catalysis does not arise for the question of inter-convertibility between pure states (cf. Sec.~\ref{subsubsec:dichotomies}). 
	The mathematical characterization of strict catalysis in the framework of LOCC was later examined in Ref.~\cite{daftuar_mathematical_2001}. 
	It was discovered that there does not exist an upper bound on the dimension of catalysts that should be considered. 
	More specifically, for most initial states, the set of final states achievable with catalysts of local dimension $d_{\ms{C}}$ is strictly larger than the set of states achievable with catalysts of smaller dimension, $d'_{\ms C} < d_{\ms C}$. 
	Moreover, there exist state-transitions that cannot be implemented with any finite-dimensional catalyst, but may be implemented with an infinite dimensional catalyst \cite{daftuar_sumit_kumar_eigenvalue_2004}.
	This poses serious difficulty in characterizing which state transformations are possible under ELOCC. 
	Furthermore, Ref.~\cite{daftuar_mathematical_2001} proved that any entangled state which has at least two non-zero and non-equal Schmidt coefficients can serve as a catalyst for some entanglement transformation. 
	It is clear that a separable state cannot be useful as a catalyst. However, counter-intuitively, the maximally entangled state cannot catalyze any transition as well, since taking a tensor product with this state would preserve the majorization structure.
	
	In general, not much is known about the structure of strict catalysis in LOCC in the case when the dimension of the catalyst is bounded. An exception to this unsatisfactory state of the art is Ref.~\cite{anspach2001two} which derived the necessary and sufficient conditions for the existence of a catalyst in the case when the systems $\ms{A}$ and $\ms{B}$ have local dimension $d_{\ms{A}} = d_{\ms{B}} = 4$ and the catalyst local dimension is $d_\ms{C} = 2$.  Ref.~\cite{xiaoming_sun_existence_2005} also presented a polynomial time algorithm to decide whether a given entanglement transformation can be deterministically catalyzed, by a pure bipartite catalyst with local dimension $k$. 
	
	The first full characterization of strict catalysis in entanglement theory (ELOCC) for the case of unbounded catalyst dimension was by Ref.~\cite{turgut_catalytic_2007}, using an infinite set of entropic conditions. 
	Independently, Klimesh \cite{klimesh_inequalities_2007} proposed a different description of a complete set of inequalities characterizing ELOCC. 
	These two sets of conditions are in fact equivalent, see also \cite{pereira_dirichlet_2013,kribs_trumping_2015,pereira_extending_2015-1}.
	Finally, the closure of the set containing all states reachable from a fixed entangled state $\ket{\psi}$ was studied in \cite{kribs_trumping_2015}, where it was found that the set can be fully described using a generalization of the majorization relation called \emph{power majorization} \cite{allen1988power}.
	The mathematical structure of probabilistic entanglement transformations was studied by Ref.~\cite{feng_catalyst-assisted_2005} who gave a necessary and sufficient condition for the existence of strict catalysts, that can increase the success probability of a state transformation{, see also the previous works \cite{feng_when_2004,duan_trade-off_2005}. \cite{feng_catalyst-assisted_2005} showed} that the maximal probability of success $p_{\text{succ}}$ for transformation $\ket{\psi} \ot \ket{\omega} \rightarrow \ket{\phi} \ot \ket{\omega}$ optimized over $\ket{\omega}$ depends only on the minimal Schmidt coefficients of $\ket{\psi}$ and $\ket{\phi}$.  Moreover, the authors of \cite{feng_catalyst-assisted_2005} have also derived necessary and sufficient conditions for a given $k$-dimensional bipartite pure catalyst state $\ket{\omega_k}$ to enable a given transformation $\ket{\psi} \rightarrow \ket{\phi}$.
	
	One of the main limitations of studying catalysis in LOCC arises from the complete state transition conditions being known only for bipartite pure states. Ref.~\cite{eisert_catalysis_2000} made first steps in generalizing catalysis to the case of mixed states. Specifically, they found an explicit example of two bipartite mixed states $\rho, \sigma \in \mathcal{D}(\mathcal{H}_{\ms A} \ot \mathcal{H}_{\ms B})$ such that $\rho\AC{\mathcal{O}_{\rm LOCC}}{}~\sigma$, but there exists no LOCC transformation turning $\rho$ into $\sigma$.
	
	\subsubsection{Entanglement distillation and formation}
	\label{subsubsec:entanglement-distillation}
	Suppose that Alice and Bob would like to communicate quantum information over a noisy quantum channel. The quality of their communication depends crucially on their ability of maintaining a high degree of entanglement \cite{lloyd1997}. Entanglement distillation and formation constitute two fundamental protocols that aim to enhance quantum communication by collectively processing multiple copies of quantum states \cite{Bennett1996, D_r_2007,Horodecki2009}. 
	
	In the paradigmatic setting of entanglement distillation, Alice and Bob share $n \gg 1$ copies of a bipartite state $\rho$, and the aim is to transform them into $m$ copies of the maximally-entangled state $\phi$ via LOCC, where the rate is given by $R := m/n$ \cite{Bennett1996}. Formally, the \emph{asymptotic distillation cost} $E_D(\rho)$ is defined as the largest rate $R^*$ such that 
	\begin{align}
		\lim_{ n \rightarrow \infty} \quad \inf_{\mathcal{E} \in \locc(\ms{A}: \ms{B})}   \norm{\mathcal{E}[\rho^{\ot n}] - \phi_+^{\ot m}}_1  = 0
	\end{align}
	While $E_D(\rho)$ measures the maximum rate of distilling $\phi$ from $\rho$ via LOCC, one can also ask about the inverse process: what is the optimal rate of creating $n$ copies of $\rho$ from $m$ maximally entangled states $\phi_+$. This has been termed \emph{entanglement
		formation} \cite{Bennet1996_mixedstate}, and the associated optimal conversion rate is the \emph{asymptotic formation cost} $E_F(\rho)$, i.e. the largest rate $R^* = n/m$ such that
	\begin{align}
		\label{eq:ent_form}
		\lim_{ n \rightarrow \infty} \quad \inf_{\mathcal{E} \in \locc(\ms{A}: \ms{B})}   \norm{\mathcal{E}[\phi_+^{\ot m}] - \rho^{\ot n}}_1  = 0
	\end{align}
	Generally $E_D(\rho) \neq E_F(\rho)$, leading to the fundamental irreversibility of entanglement transformations even in the asymptotic limit. However, for bipartite pure states, the two entanglement quantifiers coincide and are equal to the entanglement entropy of the state, i.e. for a pure state $\psi_{\ms{AB}}$,
	\begin{align}
		E_D(\psi_{\ms{AB}}) = E_F(\psi_{\ms{AB}}) = H(\psi_{\ms{A}}).
	\end{align}
	Outside of the asymptotic limit, i.e. when $n < \infty$, the quantifiers $E_D(\rho)$ and $E_F(\rho)$ lose their operational significance. 
	Recently it was discovered that {this operational significance is recovered when correlated catalysis is allowed.} 
	This approach was addressed in Ref.~\cite{Kondra2021} where a variant of Theorem \ref{thm:nielsen} was proven for correlated catalytic LOCC with bipartite pure states. More specifically, Ref.~\cite{Kondra2021} proved the following theorem, which can be seen as an instance of Lemma~\ref{lemma:partial-order-regularization}.
	
	\begin{theorem} Given bipartite pure states $\ket{\psi}_{\ms A\ms B}$ and $\ket{\phi}_{\ms A\ms B}$, then $H(\psi_{\ms A}) \geq H(\phi_{\ms A})$ iff
		$\ket{\psi}_{\ms A\ms B} \AC{\mathcal{O}_{\rm LOCC}}{\rm corr.}\ket{\phi}_{\ms A\ms B}$ to arbitrary accuracy in terms of trace-distance.
	\end{theorem}
	More generally, with the help of a correlated-catalytic LOCC operation, it is possible to convert \emph{any} distillable state $\rho_{\ms {AB}}$ (to arbitrary accuracy) into a pure state $\psi_{\ms {AB}}$ with entanglement entropy $H(\psi_{\ms A}) \leq E_D(\rho_{\ms {AB}})$.
	Conversely, via correlated-catalytic LOCC it is also possible to create any bipartite quantum state $\rho_{\ms {AB}}$ from a pure entangled state $\psi_{\ms {AB}}$ with entanglement entropy $H(\psi_{\ms A}) \geq E_{\ms C}(\rho_{\ms {AB}}).$ 
	
	Furtermore, the authors of Ref.~\cite{Kondra2021} used similar techniques to show that in the tripartite setting, i.e. when three spatially separated parties $\ms{A}$, $\ms{B}$ and $\ms{C}$ share a pure quantum state $\ket{\psi}$, the analog of distillable entanglement extended by the use of correlated-catalysis is given by $\min(H(\psi_{\ms A}), H(\psi_{\ms B}))$, in close analogy with the asymptotic setting \cite{Smolin2005entanglement}.

	{The above result uses a catalyst to convert an asymptotic result to the single-shot setting. 
		Previsouly, \cite{bennett_exact_2000} defined a notion of catalysis directly in the asymptotic setting, i.e., also involving asymptotically many copies of the catalyst and allowing for an arbitrarily small error. 
		The authors showed that such transformations can even simulate LOCC transformations which additionally allow for a sub-linear amount of quantum communication. \cite{vidal_irreversibility_2001} argue that this class of catalytic transformations nevertheless cannot be used to distil \emph{bound entanglement} (see also Sec.~\ref{sec:PPT}). Bound entangled states are entangled states with $E_D(\rho)=0$ and hence their entanglement cannot be extracted using LOCC. 
		Recently, \cite{lami_catalysis_2023} conclusively shows that not even correlated or marginal-correlated catalysis can be used to distil bound entanglement, while \cite{ganardi_catalytic_2023} shows that the distillable entanglement of distillable states cannot be increased using correlated catalysis.} 
	
	\subsubsection{Quantum state merging}
	\label{subsubsec:state-merging}
	
	Quantum state merging, sometimes called coherent state transfer, is a communication task which enables transferring a known quantum state to a distant party who already holds part of the state \cite{Horodecki2005partial,Horodecki2006}. More specifically, consider 
	three parties: Alice ($\ms{A}$), Bob ($\ms{B}$) and Referee ($\ms{R}$) who share asymptotically many copies of some pure state $\ket{\psi_{\ms{ABR}}}$. The goal is to send the $\ms{A}$ part of the state to $\ms{B}$, while preserving all correlations with $\ms{R}$. It is further assumed that Alice and Bob can perform any LOCC protocol and on top of that, they may share an arbitrary number of maximally-entangled states. As was shown in \cite{Horodecki2005partial}, the rate at which maximally-entangled states have to be supplied in order to accomplish this process is given by the quantum conditional entropy, i.e.
	\begin{align}
		H(\ms A|\ms B)_{\psi} := H(\psi_{\ms A\ms B}) - H(\psi_{\ms B}).
	\end{align}
	In other words, for $H(\ms A|\ms B)_{\psi} > 0$ state merging is possible when consuming singlets at a rate $H(\ms A|\ms B)_\psi$. 
	For $H(\ms A|\ms B)_{\psi} < 0$, not only is state merging possible, but in addition, the process generates entangled pairs at a rate given by $-H(\ms A|\ms B)_{\psi}$. 
	This provides an operational meaning to the quantum conditional entropy. 
	Variants of the state merging task were analysed, including its single-shot version \cite{Berta2009} and an extension to multiple parties \cite{Dutil2010}.
	
	An extension of the state merging protocol allowing correlated catalysis was proposed in Ref.~\cite{Kondra2021}. The setup is similar to the single-shot setting in Ref.~\cite{Berta2009}: the parties $\ms{A}$, $\ms{B}$ and $\ms{R}$ share a single copy of $\ket{\psi}_{\ms{ABR}}$ and are allowed to use any entangled state as a correlated catalyst. 
	In this extended protocol, state merging can be performed as long as $H(\ms A|\ms B)_{\psi} > 0$. On the other hand, if $H(\ms A|\ms B)_{\psi} < 0 $, then catalytic state merging can be performed not only without extra entanglement, but also with Alice and Bob gaining an additional pure state with entanglement entropy $ - H(\ms A|\ms B)_{\psi}$.  Importantly, both procedures are optimal in the sense that state merging is not possible if a pure state with a smaller entanglement entropy is provided (when $H(\ms A|\ms B)_{\psi} > 0$) and state merging with entanglement gain exceeding $-H(\ms A|\ms B)_{\psi}$ is not possible when $H(\ms A|\ms B)_{\psi} < 0$. This procedure provides an operational interpretation of the quantum conditional entropy in the single-shot regime. 
	
	\subsubsection{Quantum teleportation}
	\label{subsubsec:teleportation}
	Quantum teleportation enables transferring quantum states using pre-shared entanglement and classical communication \cite{Bennet1993}. Alice ($\ms{A}$) and Bob ($\ms{B}$) share an entangled state $\rho_{\ms {AB}}$ 
	Alice is additionally given a quantum state $\varphi_{\ms A'}$,
	which is unknown to both parties. They then attempt to transfer the state on $A'$ from Alice to Bob using a quantum channel $\mathcal{E} \in \locc(\ms{A'A}:\ms{B})$, and the entangled state they share.  The goal of the protocol is to simulate a noiseless quantum channel from Alice to Bob, i.e. the identity map $\text{id}_{\ms{A'} \rightarrow \ms{B}}$. The quality of teleportation is usually quantified using the average fidelity of teleportation introduced in Ref.~\cite{Popescu1994}: 
	\begin{align}
		\label{eq:fid_tel}
		\langle F \rangle_{\rho} := \max_{\mathcal{E}} & \, \int \langle \varphi|_{\ms B}\tr_{\ms{A'A}}\mathcal{E}(\varphi_{\ms A'} \ot \rho_{\ms{AB}})|\varphi\rangle_{\ms B} \, \text{d} \varphi_{\ms B} \nonumber \\
		\text{s.t.}& \quad \mathcal{E}\in {\locc}(\ms{A'A}:\ms{B}).
	\end{align}
	The integral in Eq. (\ref{eq:fid_tel}) is computed over a uniform distribution of all pure input states $\varphi_{\ms B} = \dyad{\varphi}_{\ms B}$ according to a normalised Haar measure $\int \text{d} \varphi_{\ms B} = \mathbb{1}_{\ms B}$.   It can be easily verified that $0 \leq \langle F \rangle_{\rho} \leq 1$, where $\langle F\rangle_{\rho} = 1$ corresponds to perfect teleportation that is possible if and only if $\rho$ is maximally-entangled. As shown in \cite{Horodecki1999}, the fidelity of teleportation (\ref{eq:fid_tel}) can be conveniently expressed as
	\begin{align}
		\langle F \rangle_{\rho} = \frac{f(\rho) d_{\ms{A'}}+1}{d_{\ms{A'}}+1},
	\end{align}
	where $f(\rho) := \max \{{\langle \phi^+ |}_{\ms{AB}} \mathcal{E}(\rho_{\ms{AB}})  {|\phi^+\rangle_{\ms{AB}}}|\ \mathcal{E} \in {\rm LOCC}(\ms{A}:\ms{B})\}$ is called the entanglement fraction, $d_{\ms{A'}}$ is the dimension of ${\ms A}'$, and {$\ket{\phi^+}_{\ms{AB}}$} is a maximally entangled state on $\ms{AB}$.
	
	The protocol for quantum teleportation has been extended to the case when ancillary entanglement is used catalytically \cite{lipka-bartosik_catalytic_2021}. For that, let us assume that Alice and Bob, in addition to $\rho_{\ms {AB}}$, also have access to a quantum system $\ms {C_A C_B}$ prepared in $\omega_\ms{C_A C_B}$.  This additional system is distributed such that Alice has access only to $\ms{C}_{\ms A}$, and Bob only to $\ms{C_B}$. Alice is then given an unknown state $\varphi_{\ms{A'}}$, and the parties perform a protocol $\mathcal{E} \in \locc({\ms {A'AC_A}:\ms{BC_B}})$ acting on both shared systems and the input. For the protocol to be catalytic, one further demands that the local state of the system $\ms{C_A C_B}$ remains the same after $\mathcal{E}$, although it can become correlated with $\ms {AB}$. The protocol $\mathcal{E}$ can then be viewed as an instance of correlated catalysis in the LOCC framework that aims at performing teleportation of ${\ms A}'$.
	
	If we have freedom to choose the state of the catalyst $\omega_{\ms{C_A C_B}}$, it is then natural to define a benchmark that optimizes over all possible states of the catalyst. This leads to the average fidelity of catalytic teleportation defined as
	\begin{align}
		\nonumber
		\langle F_{\text{cat}} \rangle_{\rho} = \max_{\mathcal{E}\!,\, \omega}& \int \langle \varphi|\tr_{\ms{A'ACC'}}\mathcal{E}(\varphi_{\ms A'} \!\ot \!\rho_{\ms {AB}} \ot \omega_{\ms{C_AC_B}})|\varphi\rangle \, \text{d} \varphi \! \\
		\text{s.t.}& \quad \tr_{\ms{A'AB}} \mathcal{E}(\varphi_{\ms A'} \ot \rho_{\ms {AB}} \ot \omega_{\ms{C_A C_B}}) = \omega_{\ms{C_A C_B}}, \nonumber\\
		&\quad \omega_{\ms{C_A C_B}} \geq 0, \quad \tr \omega_{\ms{C_A C_B}} = 1, \nonumber \\ 
		&\quad \mathcal{E} \in \locc(\ms{A'AC}:\ms{BC'}).
		\label{eq:fid_cat_tel}
	\end{align}
	The main result of Ref.~\cite{lipka-bartosik_catalytic_2021} is the following (achievable) lower bound,
	\label{thm1}
	\begin{align}
		\label{eq:f_cat}
		\langle F_{\text{cat}} \rangle_{\rho}  \geq  \frac{f_{\text{reg}}(\rho) d_{\ms{A'}} + 1}{d_{\ms{A'}} + 1},
	\end{align}
	where $f_{\text{reg}}(\rho)$ is a regularization of entanglement fraction, $f_{\text{reg}}(\rho) := \lim_{n \rightarrow \infty} [f_n(\rho^{\ot n})/n]$, with $f_n(\sigma)$ given by 
	\begin{align}
		f_{n}(\sigma) :=  \max_{\mathcal{E}}& \,\,\, \sum_{i=1}^n \langle\phi^+ | \tr_{/i}\mathcal{E}(\sigma)|\phi^+\rangle, \nonumber \\
		\text{s.t.}&\,\,\, \mathcal{E} \in {\rm LOCC} (\ms{A_1 \ldots A_n}:\ms{B_1\ldots B_n}),
		\label{eq:fn_ent_frac}
	\end{align}
	where $\tr_{/i}(\cdot)$ is the partial trace performed over local parties $1\ldots i-1, i+1\ldots n$. Notice that by taking a sub-optimal guess $\mathcal{E} = \mathcal{E}_1 \ot \mathcal{E}_2 \ot \ldots \ot \mathcal{E}_n$ with $\mathcal{E}_1 = \mathcal{E}_2 = \ldots = \mathcal{E}_n$ we have $f_{\text{reg}}(\rho) \geq f(\rho)$ for all density operators $\rho$.
	
	Importantly, the proof of Eq. (\ref{eq:f_cat}) is constructive, that is, there exists a protocol $T \in {\rm LOCC}(\ms{A'AC_A}:\ms{BC_B})$ and a catalyst $\omega_{\ms{C_AC_B}} $ that achieves the bound from Eq. (\ref{eq:f_cat}). Moreover, for most pure bipartite states $\psi_{\ms{AB}}$ we have $f_{\text{reg}}(\psi) > f(\psi)$, implying that $\langle F_{\text{cat}}\rangle_{\psi_{\ms{AB}}} > \langle F\rangle_{\psi_{\ms{AB}}}$. As a consequence, the quality of teleportation can be improved when using ancillary entanglement catalytically. 
	
	\subsubsection{Local operations and shared randomness (LOSR)}
	\label{subsubsec:LOSR}
	
	Let us denote the set of all LOSR operations between $\ms{A}$ and $\ms{B}$ with $\losr(\ms{A}:\ms{B})$. Any operation $\mathcal{E} \in \losr(\ms{A}:\ms{B})$ can be written in Kraus representation as
	\begin{align}
		\mathcal{E}[\rho_{\ms{AB}}] = \sum_{i} p_i (A_i \ot B_i) \rho_{\ms{AB}} (A_i \ot B_i)^{\dagger}
	\end{align}
	where $\{A_i\}$ and $\{B_i\}$ are sets of local Kraus operators and $\lbrace p_i \rbrace_i$ a probability distribution. 
	In analogy with Nielsen's theorem, transformations between pure bipartite states via LOSR can be characterised using majorization. 
	
	\begin{theorem}
		\label{thm:losr_conversion}
		\cite{schmid2020understanding} A bipartite quantum state $\ket{\psi}$ shared between two parties can be converted into $\ket{\phi}$ by means of LOSR if and only if there is a bipartite state $\ket{\xi}$ such that
		\begin{align}
			\label{eq:app_losr_1}
			\bm{\lambda}(\psi_{\ms{A}})^{\downarrow} = [\bm{\lambda}(\phi_{\ms{A}}) \ot \bm{\lambda}(\xi_{\ms{A'}})]^{\downarrow}
		\end{align}
	\end{theorem}
	The above theorem already implies that strict catalysis cannot provide an advantage for LOSR between pure bipartite states. 
	More specifically, Ref.~\cite{schmid2020understanding} shows that if $\ket{\psi}$ cannot be converted into $\ket{\phi}$ via LOSR, then $\ket{\psi} \ot \ket{\omega}$ also cannot be converted into $\ket{\phi} \ot \ket{\omega}$ for any bipartite state $\ket{\omega}$. 
	To see this, note that due to Theorem \ref{thm:losr_conversion}, {$\ket{\psi}  \AC{\mathcal{O}_{\rm LOSR}}~ \ket{\phi} $} is possible if and only there exist exist a pure bipartite state $\ket{\xi}$ with Schmidt coefficients $\schmidt(\xi_{\ms{A'}})$ such that  
	\begin{align}
		\label{eq:app_losr_2}
		[\schmidt(\psi_{\ms{A}}) \ot \schmidt(\omega_{\ms{C_A}})]^{\downarrow} = [\schmidt(\phi_{\ms{A}}) \ot  \schmidt (\omega_{\ms{C_A}}) \ot \schmidt(\xi_{\ms{A'}})]^{\downarrow}.
	\end{align}
	We can read this equation as $[\vec p\otimes \vec r]^{\downarrow} = [\vec q\otimes \vec r]^{\downarrow}$ with probability vectors $\vec p = \schmidt(\psi_{\ms{A}}),\vec q = \schmidt(\phi_{\ms A})\ot \schmidt(\xi_{\ms A'}), \vec r = \schmidt(\omega_{\ms{C_A}})$.  
	But this is only possible if $\vec p^\downarrow = \vec q^\downarrow$: The largest entry gives $p^\downarrow_1 r^\downarrow_1=q^\downarrow_1 r^\downarrow_1$, so that $p^\downarrow_1= q^\downarrow_1$. We can now remove all equations involving $p^\downarrow_1$. But then we similarly obtain $p^\downarrow_2 = q^\downarrow_2$. Repeating these steps we get $\vec p^\downarrow = \vec q^\downarrow$, which corresponds to \eqref{eq:app_losr_1}. 
	This means that any LOSR transformation that can be achieved via strict catalysis can also be achieved without using any catalyst{, in strong contrast to LOCC.}
	We will also see in Sec.~\ref{subsec:no-broadcasting} that a similar distinction in terms of catalysis arises when considering different ways to formalize the concept of coherence as a resource theory. 
	
	The above reasoning relies crucially on two crucial assumptions: the tensor product structure of strictly catalytic transformations and the purity of the catalyst state. 
	Therefore, the proof ceases to hold when applied directly to either correlated-catalytic transformations or when one is dealing with mixed states.
	It is an interesting open question whether correlated catalysis might be suitable for the LOSR framework. On the other hand, embezzling has been studied very early on in LOSR -- {\cite{van_dam_universal_2003} start from the state transition conditions of LOCC, i.e. majorization of Schmidt vectors in Eq.~\eqref{eq:ent_nielsen_majorization}, and show that embezzling not only trivializes the state transition conditions (i.e. enabling the preparation of a maximally entangled state from a product state for free), but also does this in such a way that the process can be accomplished using only LOSR. All this is achieved only at the cost of a disturbance $\delta > 0$ induced on the embezzler (in terms of trace distance), which can be made as small as possible by choosing a large enough dimension. }

	\subsubsection{(Non-)closure of quantum correlations}
	\label{subsubsec:non-locality}
	Bell's theorem \cite{Bell1964} states that no local hidden-variable theory can account for all predictions of quantum theory under the assumption that the hidden variables are statistically independent from the measurement settings. 
	At the core of this theorem is the notion of a Bell experiment: a source distributes two physical systems to distant observers, Alice and Bob, {so that they both share a bipartite quantum state described by a density matrix $\rho_{\ms{AB}}$}. Each observer performs a randomly chosen measurement on their {part of the} system, labelled by $x,y$ and obtains outcomes $a,b$. 
	The experiment is characterized by the joint distribution $p(ab|xy)$.

	We say that the experiment has a local (hidden-variable) model if there exists a hidden random variable $\lambda$ with an associated distribution $p(\lambda)$, and two local response functions, $p(a|x, \lambda)$ for Alice and $p(b|y, \lambda)$ for Bob, such that $p(ab|xy)$ can be achieved using $\lambda$,
	\begin{align}\label{eq:lhv}
		p(ab|xy) =  \sum_{\lambda} p(\lambda) p(a|x, \lambda) p(b|y, \lambda).
	\end{align}
	The distribution $p(ab|xy)$ is called non-local when it cannot be described in the form of Eq.~\eqref{eq:lhv}. Every possible action taken by Alice and Bob can be viewed as an LOSR operation and vice versa \cite{Buscemi2012}. 
	
	Bell non-locality is often formalized in terms of \emph{non-local games}, 
	which have been extensively studied in computer science, where they are a special instance of interactive proof systems \cite{Cleve2004}.
	In such games, Alice ($\ms{A}$) and Bob ($\ms{B}$) play against a referee $(\ms{R})$. 
	The referee chooses a question $x \in \mathcal{X}$ for Alice and $y \in \mathcal{Y}$ for Bob according to some probability distribution $p(x,y): \mathcal{X} \times \mathcal{Y} \rightarrow [0, 1]$, where $\mathcal{X}$ and $\mathcal{Y}$ denote finite sets of questions. 
	Without communicating, Alice (Bob) returns an answer $a \in \mathcal{A}$ $(b \in \mathcal{B})$ from a finite set of possible answers $\mathcal{A}$ $(\mathcal{B})$. 
	In quantum mechanics, this corresponds to applying local measurements $\{M_{a|x}\}$ for Alice and $\{N_{b|y}\}$ for Bob. 
	Based on the received answers and according to a pre-arranged set of rules, the referee decides whether the players win or lose the game. 
	These rules are typically expressed using a function $V:$ $\mathcal{A} \times \mathcal{B} \times \mathcal{X}\times \mathcal{Y} \rightarrow \{0, 1\}$, where $V_{ab}^{xy} = 1$ if and only if Alice and Bob win the game when answering $a$ and $b$ for questions $x$ and $y$. 
	The probability of Alice and Bob winning the game $\mathcal{G} = \{p(x,y), V\}$, maximized over all  measurement strategies, is given by
	\begin{align}
		\label{eq:bell_beh}
		p_{\text{guess}}(\mathcal{G},\rho_{\ms A\ms B}) =\max\sum_{a,b,x,y} p(x,y)  p(a,b|x,y)V_{ab}^{xy},
	\end{align}
	where 	$p(a,b|x,y) = \tr[\left(M_{a|x} \ot N_{b|y}\right)\rho_{\ms A\ms B}]$.
	Here, we made explicit that we work in the \emph{tensor-product framework}, where Alice's and Bob's quantum system are described by local Hilbert spaces $\mc H_{\ms A}$ and $\mc H_{\ms B}$, respectively, and the joint Hilbert space is given by their tensor-product. 
	In the more general \emph{commuting operator framework} it is simply assumed that Alice's and Bob's measurement operators are defined on a global Hilbert space $\mc H_{\ms A\ms B}$ and that $[M_{a|x},N_{b|y}]=0$, see also Sec.~\ref{subsec:infinite-dimensional}.
	The tensor-product framework and the commuting operator framework are equivalent for finite-dimensional Hilbert spaces.
	
	Bell inequalities provide upper bounds on Eq.~\eqref{eq:bell_beh} with which the players can win the game using the best classical strategy, i.e. when $\rho_{\ms A\ms B}$ is separable.
	A violation of a Bell inequality indicates there is a quantum strategy which outperforms the best classical strategy. 
	In other words, Bell’s theorem \cite{Bell1964} asserts that there exist games where players who share entanglement can outperform
	players who do not. 
	The most famous example of this is the CHSH game $\mathcal{G}_{\text{CHSH}}$ \cite{Clauser1969}, where Tsirelson's theorem provides an upper bound on the guessing probability $p_{\text{guess}}(\mathcal{G}_{\text{CHSH}}, \rho_{\ms{AB}})$ for classical strategies \cite{cirel1980quantum}. 
	For a more complete account of non-local games see Ref.~\cite{watrous2018theory}.
	
	An immediate application of non-local games is the task of witnessing the dimension of entanglement. 
	More specifically, a non-local game $\mathcal{G}$ can be used as an (entanglement) dimension witness if a certain $p_{\text{guess}}(\mathcal{G}, \rho_{\ms{AB}})$ can only be achieved with entangled states $\rho_{\ms{AB}}$ on a Hilbert space with a given minimal dimension \cite{Brunner2008testing}. 
	In this spirit, \cite{coladangelo_two-player_2020} proposed a non-local game $\mathcal{G}_{\text{Coladangelo}}$ in which the players’ maximum achievable guessing probability increases monotonically with the allowed Hilbert space dimension. 
	In the following, let
	\begin{align}
		p^*_{\text{Coladangelo}} := \sup_{\rho_{\ms A\ms B}}  p_{\text{guess}}(\mathcal{G}_{\text{Coladangelo}}, \rho_{\ms{AB}}),
	\end{align}
	where the maximization is performed over quantum states $\rho_{\ms A \ms B}$ on Hilbert spaces with finite (but arbitrarily large) dimension.
	On one hand, \cite{coladangelo_two-player_2020} shows that
	$p_{\text{guess}}(\mathcal{G}_{\text{Coladangelo}}, \rho_{\ms{AB}}) < p^*_{\text{Coladangelo}}$
	within the tensor-product framework for both finite-dimensional and infinite-dimensional quantum states $\rho_{\ms A\ms B}$ \footnote{The set of correlations achievable by infinite-dimensional density matrices in the tensor-product framework is included in the closure of the set of correlations achievable by finite-dimensional quantum states \cite{scholz_tsirelsons_2008}.}.
	The proof of this statement is based on choosing $\rho_{AB}$ to be an (approximate) embezzler, see Sec.~\ref{subsubsec:embezzlement}. Moreover, there it is further shown that a strategy achieving $p^*_{\text{Coladangelo}}$ requires a perfect embezzler, which requires a commuting operator framework, see also \ref{subsec:infinite-dimensional}.
	An approximate embezzler allows Alice and Bob to simulate an approximately coherent strategy, i.e. perform measurements in a way that retains a high degree of coherence between the parties. 
	High coherence of the process, on the other hand, leads directly to a high score in the non-local game. 
	The main idea behind using embezzlers in such a coherent state exchange process is described in Example \ref{ex:coherent_state_exchange}. 
	\begin{example}[Coherent state exchange]
		\label{ex:coherent_state_exchange} 
		Consider a bipartite pure state $\ket{\phi}_{\ms{AB}}$ shared between Alice $(\ms A)$ and Bob $(\ms B)$. The goal is to transform $\ket{\phi}_{\ms{AB}}$ into $\ket{\psi}_{\ms{AB}}$ under two assumptions: ($i$) no communication is allowed, and ($ii$) the process must be implemented in a coherent way, that is
		\begin{align}
			\label{eq:coh_state_exch_example}
			\alpha \ket{00}_{\ms{A'B'}}\ket{\gamma}_{\ms{AB}} &+ \beta \ket{11}_{\ms{A'B'}} \ket{\phi}_{\ms{AB}} \nonumber\\
			&\rightarrow \alpha \ket{00}_{\ms{A'B'}} \ket{\gamma}_{\ms{AB}} + \beta \ket{11}_{\ms{A'B'}} \ket{\psi}_{\ms{AB}},\nonumber
		\end{align}
		where {$\ms{A'B'}$} represent control systems, one held by each player. 
		Nielsen's theorem implies that Eq.~\eqref{eq:coh_state_exch_example} is impossible when $\schmidt (\psi_{A}) \succ \schmidt (\phi_{A})$. However, this is no longer true when catalysis is involved: 
		using an auxiliary state
		\begin{align}
			\ket{\omega_n} = \frac{1}{\sqrt{n}} \sum_{i=1}^n \ket{\phi}^{\ot i} \ket{\psi}^{\ot (n-i+1)},
		\end{align}
		and cyclically permuting local subsystems transforms $\ket{\phi}_{\ms{AB}}\ket{\omega_n}$ into $\ket{\psi}_{\ms{AB}} \ket{\omega_n'}$ with $\braket{\omega_n'}{\omega_n} = 1-\frac{1}{n}$, so that the auxiliary state is approximately preserved. 
		This protocol is almost identical to the universal embezzlement construction in Sec.~\ref{subsubsec:approx-cat-construction}.
		That the auxiliary system is approximately catalytic is the key property which ensures that the transformation is done coherently. Controlling the cyclic permutation on the control systems yields 
		\begin{align}
			\alpha \ket{00}_{\ms{A'B'}}&\ket{\gamma}_{\ms{AB}}\ket{\omega_n} + \beta \ket{11}_{\ms{A'B'}} \ket{\phi}_{\ms{AB}}\ket{\omega_n} \nonumber \\ &\rightarrow \alpha \ket{00}_{\ms{A'B'}} \ket{\gamma}_{\ms{AB}}\ket{\omega_n} + \beta \ket{11}_{\ms{A'B'}} \ket{\psi}_{\ms{AB}}\ket{\omega_n'}\nonumber
		\end{align}
		and hence the parties can approximately implement the coherent state exchange since $\ket{\omega_n'}\approx \ket{\omega_n}$. 
		If the state $\ket{\omega_n'}$ was far from $\ket{\omega_n}$, then discarding the ancilla would effectively decohere the primary system. 
		Indeed, the requirement $(ii)$ can only be satisfied if the process does not change much the auxiliary state. 
		The players cannot, e.g. swap $\ket{\phi}_{\ms{AB}}$ with an initially shared copy of $\ket{\psi}_{\ms{AB}}$ without losing coherence.
	\end{example}
	
	From the perspective of Bell non-locality, \cite{coladangelo_two-player_2020} proposes an embezzlement-based Bell inequality that cannot be maximally violated within the tensor-product framework, but where the limiting value can be achieved arbitrarily well. This shows that the set of quantum correlations from a Bell experiment described in the tensor-product framework is not closed. In fact, the existence of a suitable game has been proven first using representation theory of finitely-presented groups in Refs.~\cite{slofstra_tsirelsons_2019,Slofstra2019set}, with subsequent alternative proofs in Refs.~\cite{dykema_non-closure_2019,musat_non-closure_2020}. 
	In the case of \cite{coladangelo_two-player_2020}, however, embezzlement-based techniques allowed for a simpler proof using only basic linear algebra and an explicit non-local game. Previously, \cite{cleve_perfect_2017} showed a similar result for \emph{coherent embezzlement games} \cite{regev_quantum_2013}. 
	Ref. \cite{coladangelo_two-player_2020} transparently demonstrates the relationship between the dimension of the used entangled state and the associated score in a non-local game.  A similar technique involving embezzlement was used to prove that any number of parties can coherently exchange any pure quantum state for another, without communication, given prior shared entanglement \cite{Leung2013coherent}.

	\subsubsection{Embezzlement and the Reverse Coding Theorem}
	The (classical) channel capacity is the central concept in information theory. It quantifies the maximum rate at which classical data can be transmitted in the limit of many uses of the channel \cite{cover1999elements}. The Noisy Channel Coding Theorem is a milestone result that provides a closed formula for the channel capacity \cite{shannon1948mathematical}. It states that the (classical) capacity $C$ of a memoryless classical communication channel is given by the supremum of mutual information between the input and output of the channel. More specifically, let $X$ be a random variable distributed according to a probability distribution $\bm{p}_X$ and $\mathcal{E}$ be a classical and memoryless channel. Then the classical capacity of $\mathcal{E}$ is given by
	\begin{align}
		C(\mathcal{E}) = \sup_{\bm{p}_X}\,\, I(X:Y)_{\mathcal{E}[\bm{p}_X]},
	\end{align}
	where $Y \sim \mathcal{E}[\bm{p}_X]$ denotes the output of the channel and $I(X:Y)_{\mathcal{E}[\bm{p}_X]} := H(X) + X(Y) - H(X, Y)$. After the appearance of Shannon's theorem, a closely related (classical) Reverse Coding Theorem was proven \cite{bennett2002entanglement}. The theorem states that every classical channel $\mathcal{E}$ can be simulated exactly using  $C(\mathcal{E})$ bits of classical communication and free shared randomness between the sender and the receiver. 
	In other words, when given access to shared randomness and local operations, any classical channel $\mathcal{N}$ can be simulated using $\mathcal{E}$ at a rate given by $C(\mathcal{E})/C(\mathcal{N})$. Therefore, every classical channel is completely characterized by its capacity, in the limit of many uses of the channel. 
	
	The communication problem becomes more complex in quantum theory, where channels are known to exhibit different types of capacities \cite{holevo1998capacity, shor2003capacities, devetak2005private}. Presumably the most natural analogue of classical channel capacity in the quantum regime is the \emph{entanglement-assisted classical capacity}, $C_E$. It is defined as the highest rate at which \emph{classical} information can be transmitted when the sender and receiver share unlimited, noiseless entanglement \cite{Bennett_1999}. An analogue of the (classical) Noisy Channel Coding Theorem was proven in \cite{bennett2002entanglement} and provides a closed-form formula relating $C_E$ to the quantum mutual information (see Eq. \eqref{eq:mutual_info}),
	\begin{align}
		\label{eq:ent_ass_capacity}
		C_E(\mathcal{E}) = \max_{\rho \in \mc D(\ms X)} I(\ms Y:\ms X')_{(\mathcal{E}\ot \id)\Phi_{\rho}},
	\end{align}
	where $\mc E$ is a quantum channel from $\ms X$ to $\ms Y$, $\Phi_{\rho}$  is a purification of $\rho$ on $\ms{XX'}$ and the optimization ranges over all input states $\rho$ on $\ms X$. Ref.~\cite{bennett2002entanglement} conjectured a quantum version of the Reverse Coding Theorem, which was subsequently proven \cite{Bennett_2014}. 
	
	The Quantum Reverse Coding Theorem states that any quantum channel $\mathcal{E}$ can be simulated using $C_E(\mathcal{E})$ bits of classical communication under unlimited shared entanglement. 
	As a consequence, $C_E$ suffices to characterize the quantum channel when pre-shared entanglement is allowed for free. 
	The Reverse Coding Theorem conveys the idea that entanglement cannot be easily discarded when access to communication is limited. 
	Specifically, when two entangled states $\ket{\phi_1}_{\ms{AB}}$ and $\ket{\phi_2}_{\ms{AB}}$ are part of a superposition, i.e. $\ket{\psi}_{\ms{ABC}} = (\ket{\phi_1}_{\ms{AB}}\ket{0}_{\ms C} + \ket{\phi_2}_{\ms{AB}}\ket{1}_{\ms{C}}) / \sqrt{2}$, then it is impossible to retain the entanglement in one branch of the superposition and remove it from the other, without either using classical communication or causing decoherence. This can also be understood as a consequence of the fact that changing the \emph{entanglement spread} of an entangled pure quantum state always requires communication \cite{harrow2010entanglement}. 
	
	In quantum information theory, free entanglement usually takes the form of maximally entangled states. Surprisingly, even if one has an infinite supply of maximally entangled states, the optimal entanglement-assisted capacity $C_E$ from Eq. (\ref{eq:ent_ass_capacity}) cannot be reached \cite{Bennett_2014}. Specifically, a communication protocol can only attain Eq.~\eqref{eq:ent_ass_capacity} if it can be implemented coherently,
	similar to the coherent state exchange protocol in Example \ref{ex:coherent_state_exchange}. 
	Interestingly, the Quantum Reverse Coding Theorem can be achieved by using entanglement in the form of embezzlers, instead of maximally entangled states.  
	This allows to overcome the constraints imposed by entanglement spread, and ultimately leads to the rate of communication achieving the capacity $C_E$. 
	
	{In the above example, embezzlement was used to generalize a classical information processing task to the quantum case. 
		Recently, \cite{george_one-shot_2023} also used embezzlement to generalize the classical task of distributed source simulation to a fully quantum setting, which they call \emph{embezzling source simulation}.}
	
	\subsubsection{Cryptography}
	In many situations, communication protocols require an authentication step to ensure secure communication between the two parties (Alice and Bob). This is necessary to protect against potential security breaches from an impostor (Eve). To authenticate Bob to Alice, the following protocol involving strict catalysis was proposed in Ref.~\cite{barnum1999quantum}:
	\begin{enumerate}
		\item Alice and Bob select a pair of incomparable states, e.g. $|\psi\rangle_{\ms{AB}}$ and $|\phi\rangle_{\ms{AB}}$. Alice then prepares $\ket{\psi}_{\ms{AB}}$ in her lab.
		\item They also share a strict catalyst $|\omega\rangle_{\ms{A'B'}}$ that enables the forbidden transition $|\psi\rangle_{\ms{AB}} \rightarrow |\phi\rangle_{\ms{AB}}$.
		\item  Alice sends one part of the state $|\psi\rangle_{\ms{AB}}$ to Bob.
		\item Alice and Bob use the catalyst to perform the transition $|\psi\rangle_{\ms{AB}} \otimes |\omega\rangle_{\ms{A'B'}} \rightarrow |\phi\rangle_{\ms{AB}} \otimes |\omega\rangle_{\ms{A'B'}}$.
		\item Bob sends his part of $|\phi\rangle_{\ms{AB}}$ to Alice, who measures it.
	\end{enumerate}
	The security of this protocol relies on the incomparability of the states $|\psi\rangle_{\ms{AB}}$ and $|\phi\rangle_{\ms{AB}}$. If Eve intercepts the communication and receives the $\ms{B}$ part of $|\psi\rangle_{\ms{AB}}$, she will not be able to transform it into $|\phi\rangle_{\ms{AB}}$ without access to the catalyst. 
	
	A second cryptographic scenario involving catalysis has been discussed in Ref.~\cite{boes_catalytic_2018}. Here, Alice wants to communicate a (unknown) quantum state $\rho_{\ms D}$ on a data-register $\ms D$ to Bob over a public quantum channel, in such a way that no information about $\rho_{\ms D}$ is revealed to the public. Alice and Bob are assumed to share a number of maximally entangled qubits (ebits). 
	Alice now applies a suitable unitary (independent of $\rho_{\ms D}$) on $\ms D$ and her part of the shared ebits with the result that the state on $\ms D$ is maximally mixed.
	
	Afterwards she sends $\ms D$ to Bob using a public quantum channel.
	By applying a suitable unitary, Bob can recover the initial state on $\ms D$ and, moreover, restore the ebits to their initial state.
	The catalytic ebits act as a secret key that is used to encrypt the quantum data for the transmission over the public quantum channel.

	For a thorough review on quantum cryptographic schemes we refer to Refs. \cite{Gisin2002} and \cite{Pirandola_2020}, as well as the more recent review \cite{Portmann2022}.
	
	\subsubsection{Beyond local operations and classical communication}
	\label{sec:PPT}
	So far we have discussed two classes of quantum operations in entanglement theory: LOCC and LOSR. The relevance of these sets of operations is two-fold: first, they reflect typical physical restrictions imposed by many basic protocols in quantum information theory \cite{alber2003quantum}, and second, they characterize operationally entangled states as those that cannot be prepared via LOCC/LOSR \cite{chitambar2014everything}. 
	
	Entanglement transformations are challenging to study, due to the fact that the mathematical structure of LOCC/LOSR is far from being well-understood. This motivates the investigation of other operations that are potentially more easily characterized. One such example is PPT-preserving (PPTP) operations \cite{Rains1999}, defined as quantum channels that map the set of states with a positive partial transpose back to itself. It is known that PPTP operations are strictly more powerful than LOCC, e.g. they allow for creating bound entangled states from product states \cite{Eggeling2001}. See Refs. \cite{Rains2001,Ishizaka2005} for a study of the properties of PPTP operations.
	The significance of PPTP operations lies mainly in the fact that they can be efficiently characterized in terms of semi-definite constraints. Since PPTP operations form a strict superset of LOCC operations, one can then define and determine relevant quantifiers under PPTP operations, such as entanglement cost \cite{audenaert2003entanglement} or distillable entanglement \cite{Ishizaka2005}. This approach allows to obtain meaningful bounds on the corresponding quantifiers under LOCC. Similarly, one can consider catalytic PPT-preserving operations in order to understand the limitations of catalytic LOCC. 
	
	In this line, Ref.~\cite{matthews_pure-state_2008} investigated transformations of pure bipartite states under PPTP operations. In particular, they demonstrated that, in analogy to LOCC, strict catalysis also enlarges the set of possible transformations under PPTP operations. {Surprisingly, in contrast to the LOCC paradigm, even a maximally-entangled quantum state can be a useful catalyst under strictly catalytic PPTP operations.} 
	More specifically, Ref.~\cite{matthews_pure-state_2008} shows that a PPTP channel between two pure bipartite states $\ket{\psi}_{\ms {AB}}$ and $\ket{\phi}_{\ms {AB}} $ can exist only if 
	\begin{align}
		H_{\alpha}(\schmidt (\psi_{\ms A})) \geq H_{\alpha}(\schmidt (\phi_{\ms A})) \hspace{5pt} \text{for} \hspace{5pt} \alpha \in \left\{\frac{1}{2}, 1, \infty\right\}.
	\end{align}
	Ref.~\cite{matthews_pure-state_2008} also found necessary and sufficient condition for transforming a rank-$k$ maximally entangled state into any other pure state via PPTP operations, assisted with a maximally entangled catalyst. That is, they found that a transformation
	
		\begin{align}
			\ket{\Phi^{+,k}}_{\ms {AB}}  \AC{\mathcal{O}_\textsc{PPT}}{}~ \ket{\phi}_{\ms {AB}}
		\end{align}
is possible if and only if the Renyi entropy of order $\alpha = 1/2$ (strictly) decreases,
\begin{align}\label{eq:renyi2_pptp}
	H_{\frac{1}{2}}(\schmidt(\Phi_{\ms{A}}^{+,k})) > H_{\frac{1}{2}}(\schmidt(\phi_{\ms{A}})).
\end{align}
Note that here, both the initial system state and catalyst are maximally entangled.
Because of Eq.~\eqref{eq:renyi2_pptp}, it is possible to reach states that \emph{increase} the Shannon entropy of Schmidt coefficients (i.e the Renyi entropy of order $\alpha = 1$), in stark contrast with strictly catalytic LOCC operations that can never increase any of the Renyi entropies, or correlated-catalytic LOCC operations that never increase the Shannon entropy. As a consequence, catalytic PPTP operations can increase the asymptotic entanglement content of quantum states and therefore loses some of its operational significance when catalysis is allowed.

\subsubsection{Multipartite entanglement}
\label{subsubsec:multipartite}
So far our discussion of entanglement was mainly restricted to scenarios involving two parties (bipartite scenarios). For pure states, these scenarios can relatively easily be studied using majorization, due to the Nielsen's Theorem~\ref{thm:nielsen}. However, this not the case for genuine multipartite scenarios, i.e. involving more than two parties. In this case, we do not have a simple criterion to decide which state-transitions are possible. Consequently, studying catalysis in these scenarios is challenging. 

One of the rare results for multipartite scenarios can be found in Ref.~\cite{chen2010tensor}. 
Given a $N$-partite pure quantum state $\ket{\psi}\in\bigotimes_{i=1}^N \mathcal{H}_i$, the \emph{tensor rank} $\mathrm{rk}(\psi)$ of $\ket{\psi}$ is defined as the smallest number of product states $\lbrace \bigotimes_{i=1}^N \ket{\psi'}_i \rbrace_{j=1}^{\mathrm{rk}(\psi)}$ such that its linear span contains $\ket{\psi}$. 
For $N=2$, the tensor rank reduces to the Schmidt rank; however, for $N > 2$, computing this quantity is generally NP-hard \cite{johan1990tensor}. Importantly, it is known that the tensor rank is a monotone under LOCC operations \cite{bengtsson_zyczkowski_2006}. 
The converse, however, is only true for certain classes of states; in particular, it is true for states which are reversibly interconvertible w.r.t. the $N$-partite GHZ state $\ket{{\rm GHZ}_N^d} =\frac{1}{\sqrt{d}} \sum_{i = 1}^d \ket{i}^{\otimes N}$. 
For such states, the possible state transitions are fully characterized by the tensor rank. 
Ref.~\cite{chen2010tensor} made use of this knowledge to study catalysis in multipartite entanglement. They found examples of non-trivial strict catalysis in the multipartite scenario under probabilistic LOCC. 
More specifically, they found that strict catalysis allows for an increase of the success probability of a state transition from zero to a strictly positive value. 
More recently, Ref.~\cite{neven2021local} reported the first examples of strict catalysis under deterministic LOCC in a genuinely multipartite scenario. 

An important distinction between bipartite and multipartite scenarios is that the latter leads to multiple (inequivalent) classes of entanglement: States of one class cannot be converted to states of the other class with a non-zero success probability \cite{dur_three_2000}. In this context, Ref.~\cite{ghiu_entanglement-assisted_2001} showed that strict catalysis does not help to convert between distinct classes.

\subsubsection{Contextuality}
\label{subsubsec:contextuality}
The notion of contextuality aims to characterize the property of quantum theory that it is impossible to assign pre-determined values (hidden variables)
to observables, such that the functional relationship between compatible observables is conserved \cite{Bell_1966,Kochen_1967}.
More specifically, let $M$ be a set of $n$ observables associated to a physical setting. 
In general, not all of these observables will be jointly measurable, but for any subset $C=\{A_1,\ldots,A_m\} \subset M$ of jointly measurable observables, called \emph{measurement context},
there is an associated probability distribution $p_C(a_1,\ldots,a_m)$
for the outcomes $a_j$ of the measurement of observable $A_j$.
A set of observables is jointly measureable in quantum theory if and only if the observables arise as a coarse-graining of a parent observable \cite{guhne2021incompatible}.
A \emph{non-contextual hidden-variable model} for a physical setting is given by a random variable $\lambda$ with distribution $q(\lambda)$ and an assignment of outcomes $a_1,\ldots,a_n$ to all $n$ observables with probability $q(a_1,\ldots,a_n|\lambda)$ such that:
\begin{enumerate}
	\item  For each value of $\lambda$ and every context $C$, the assignment of outcomes preserves the functional relationship between the observables. 
	That is, if $A_i,A_j\in C$ and $A_i = f(A_j)$ for some function $f$, then $a_i = f(a_j)$.
	\item The assignment reproduces the distributions $p_C$,
	\begin{align}
		p_C(a_1,\ldots,a_m) = \sum_\lambda q_C(a_1,\ldots,a_m|\lambda) q(\lambda),
	\end{align}
	where $q_C(a_1,\ldots,a_m|\lambda)$ is the restriction of $q(\cdot|\lambda)$ to the given measurement context $C$. 
	If the distributions $p_C$ cannot be reproduced by a non-contextual hidden-variable model, the physical setting is called \emph{contextual}. 
\end{enumerate}

Bell's theorem and the Kochen-Specker theorem \cite{Kochen_1967} showed that quantum theory is contextual for Hilbert space dimensions $d \geq 3$. 
It has been found that contextuality is closely related to quantum computational speedups \cite{raussendorf2013contextuality,Howard_2014,Bermejo_Vega_2017}. 
This further suggests to formulate a resource theory, where the free objects correspond to non-contextual hidden-variable models, and free operations correspond to consistent wirings of models. See \cite{amaral2019resource} and \cite{Abramsky_2019} for two different approaches to formulating a resource theory of contextuality. 
Ref.~\cite{karvonen2021neither} showed that, in the resource-theoretic framework of \cite{Abramsky_2019}, strict catalysis is not useful,
mirroring the case of Bell non-locality, see Section~\ref{subsubsec:LOSR}. Since the set of free operations considered by \cite{Abramsky_2019} strictly larger than that of \cite{amaral2019resource} (see \cite{Budroni_2022}), it is not clear whether strict catalysis can be useful in the latter. 
Furthermore, as far as we know, it is currently unknown whether other types of catalysis, such as correlated catalysis, can be useful in contextuality.


\subsection{Thermodynamics}\label{subsub:appl_thermo}

Quantum thermodynamics is an interdisciplinary field that adopts the tools of multiple frameworks like stochastic thermodynamics \cite{esposito2009nonequilibrium,seifert2012stochastic}, open quantum systems \cite{breuer2002theory,kosloff2013quantum}, and quantum information \cite{goold_role_2016,Vinjanampathy_2016} to extend thermodynamic considerations to the quantum realm. 
Although quantum thermodynamics often uses different frameworks to make its considerations quantitative, the central questions it raises are universal across complementary frameworks. Most of these questions can be formulated in a unified way by analysing the closed evolution of the system and environment. 
For that, we consider the interaction of a system of interest $\ms{S}$ prepared in some state $\rho_{\ms{S}}$ with a thermal environment $\ms{E}$ in the state $\gamma_{\ms{E}} = e^{-\beta \hat H_{\ms{E}}}/Z_{\ms{E}}$ with $Z_{\ms{E}} = \Tr e^{-\beta \hat H_{\ms{E}}}$, by means of a global unitary $U$.  
The final state of the composite system $\ms{SE}$ after interaction is given by
\begin{align}
	\label{eq:general_thermo_map}
	\sigma_{\ms{SE}} = U(\rho_{\ms{S}} \ot \gamma_{\ms{E}})U^{\dagger}.
\end{align}
The unitary $U$ encodes all physically-plausible types of interactions with respect to their strength (weak or strong coupling), complexity (local
or collective), or duration (short or long relative to natural time scales). It also enables us to study time-dependent Hamiltonians and work protocols, and furthermore makes no assumptions about the structure of $\ms{E}$, which can be both macroscopic or have dimensions comparable to those of $\ms{S}$. The map Eq.~(\ref{eq:general_thermo_map}) is therefore remarkably general and captures most of the effects encountered in thermodynamics.

Despite its generality, Eq.~(\ref{eq:general_thermo_map}) does not explicitly include one of the most ubiquitous assumptions in classical thermodynamics: auxiliary systems used cyclically in the thermodynamic process. These can be seen as an implicit way to use catalysis in textbook thermodynamics: For example, the operation of a piston when compressing an ideal gas can be seen as a catalytic transformation where the catalyst (piston) assists in transforming the system (ideal gas). 

The cyclic auxiliary systems can be modelled by adding a catalytic environment $\ms{C}$ interacting with both the system $\ms{S}$ and the thermal environment $\ms{E}$. The two environments are fundamentally distinct: While $\ms{E}$ describes an environment whose properties can freely change, $\ms{C}$ captures all of the degrees of freedom which must be preserved. We can therefore write the final state of the composite system $\ms{SCE}$ after an arbitrary unitary interaction $U$ as
\begin{align}
	\label{eq:general_thermo_map_cat}
	\sigma_{\ms{SCE}} = U(\rho_{\ms{S}} \ot \omega_{\ms{C}} \ot \gamma_{\ms{E}})U^{\dagger},
\end{align}
where we further require that $\Tr_{\ms{E}}[\sigma_{\ms{SCE}}] \approx \sigma_{\ms{S}} \ot \omega_{\ms{C}}$. 
Furthermore, unless stated otherwise, we will denote the local Hamiltonians of the three subsystems by $\hat H_{\ms{A}}$ with $\ms{A} \in \{\ms{S}, \ms{E}, \ms{C}\}$.  Catalysts evolving according to Eq. (\ref{eq:general_thermo_map_cat}) can be used to model various components of thermodynamic processes. Some particular examples include:

\begin{itemize}
	\item \textit{Finite environments.} 
	When the environment is significantly larger than the system, the back-action it experiences is usually negligible. However, in certain situations (e.g. when the environment is finite) it is important to quantify how it reacts due to the interaction with the system. This is relevant, for example, when one is interested in minimizing heat dissipation (see Section \ref{subsub:dissipation}). In such cases one can think about the finite thermal environment as an approximate catalyst.
	\item \textit{Heat engines. }Heat engines are machines that make use of the temperature gradient to generate useful work. The machinery of an engine can be viewed as a catalyst that facilitates the conversion of heat into work, while remaining unchanged during the operation. In fact, Clausius' statement of the second law of thermodynamics is formulated for a system undergoing a cyclic process, therefore a catalyst. For details see Sec. \ref{subsub:workext}.

	\item \textit{Clocks.} Control operations w.r.t. time are required for turning on and off a time-dependent interaction term in the Hamiltonian, thus allowing to run thermodynamic protocols in an autonomous manner. Any irreversible change in the clock system can deteriorate the protocol, and thus minimizing its disturbance is of key importance. Because of this, clock systems are modelled explicitly as catalysts, see Sec. \ref{subsub:CTOs} and \ref{subsubsec:batteries}.
	
	\item \textit{Apparatus.} Laboratory equipment usually facilitates experiments by augmenting control or improving performance, without undergoing change themselves (clocks being a special example). Such apparatus can be modelled as a catalyst when one wants to avoid using it as a source of energy or entropy (see Sec. \ref{subsubsec:cooling} and \ref{subsubsec:CV}).
	
\end{itemize}

\subsubsection{Minimizing dissipation in thermodynamic protocols}\label{subsub:dissipation}
Due to the seminal work of Landauer \cite{landauer1961irreversibility} it has been recognized that a logically irreversible erasure of information from a memory system leads to an unavoidable increase in the entropy of its environment. 
More specifically, let $\ms{S}$ be the memory system prepared as $\rho_{\ms{S}}$ with Hamiltonian $\hat H_{\ms{S}}$, and $\ms E$ be its thermal environment with Hamiltonian $\hat H_{\ms{E}}$. The erasure process is modelled using Eq. (\ref{eq:general_thermo_map}) and leads to a joint state $\sigma_{\ms{SE}}$. The heat transferred to the environment $\ms{E}$ is defined as the change in average energy, i.e. $Q := \Tr[\hat H_{\ms E}(\sigma_{\ms E} - \rho_{\ms E})]$. Erasing information from $\ms S$ is equivalent to transforming it into a deterministic (pure) state, which causes an inevitable increase in the entropy of the environment equal to $\Delta H_{\ms{E}} := H(\sigma_{\ms{E}}) - H(\rho_{\ms{E}})$  \cite{bennett1982thermodynamics}. Due to the unitarity of the underlying dynamics, the entropy change on the system is given by $\Delta H_{\ms{S}} := H(\sigma_{\ms S}) - H(\rho_{\ms S}) = - \Delta H_{\ms{E}}$. The heat $Q$ and $\Delta H_{\ms S}$ are therefore related by
\begin{align}
	\label{eq:landauer_bound}
	\beta  Q \geq  - \Delta H_{\ms S},
\end{align}
which can be seen as a fundamental bound on the minimal amount of heat dissipated to the environment \cite{landauer1961irreversibility}. Ref.~\cite{reeb2014improved} derived a sharpened (equality) version of the above formula, i.e.
\begin{align}
	\label{eq:landauer_exact}
	\beta Q = -\Delta H_{\ms S} + I(\ms S:\ms E)_{\sigma_{\ms SE}} + D(\sigma_{\ms E}\| \rho_{\ms E}),
\end{align}
which, due to the non-negativity of mutual information and relative entropy, implies Eq. (\ref{eq:landauer_bound}). The quantity $Q_{\text{diss}} :=  Q + \Delta H_{\ms{S}}/\beta$ is known as the \emph{dissipated heat} and captures the irreversible character of the thermodynamic process \cite{jarzynski2011equalities}. In the case when $\Delta H_{\ms S} \leq 0$, Ref.~\cite{reeb2014improved} also proved the following lower bound
\begin{align}
	\label{eq:reebs_wolf_bound_opt}
	\beta Q_{\text{diss}} \geq \frac{2 (\Delta H_{\ms S})^2}{\log^2(d_{\ms E}-1) + 4}.
\end{align}
This is a strict improvement over the Landauer's bound $Q_{\text{diss}} > 0$ whenever the environment involved in the process has a finite Hilbert space dimension $d_{\ms E}$. The bound was further shown to be achievable for specific states $\rho_{\ms S}$ and Hamiltonians $\hat H_{\ms S}$. Interestingly, it is currently not known whether Eq. (\ref{eq:reebs_wolf_bound_opt}) is tight for general quantum processes. Furthermore, as Eq. (\ref{eq:reebs_wolf_bound_opt}) results from mathematical properties of the relative entropy, it is not clear whether there exists a physical process which achieves $Q_{\rm diss}\propto \left( \log d_{\ms E} \right)^{-2}$. 
More specifically, let us consider the minimal \emph{achievable} heat dissipation $Q^*_{\text{diss}}$
\begin{align}
	\label{eq:dissipation_achievable}
	Q_{\text{diss}}^* := \min_{\hat H_{\ms{E}}, U} Q_{\rm diss},
\end{align}
where $Q = Q(\rho_S, \beta, \hat H_{\ms{S}}, \hat H_{\ms{E}}, U)$, and similarly for $\Delta H_{\ms{S}}$. 
The best known protocols that aim to minimize $Q_{\text{diss}}^*$ were proposed in \cite{reeb2014improved}, and further analyzed in \cite{skrzypczyk2014work,Baumer_2019}. In all these cases heat dissipation decreases linearly with $\log d_{\ms E}$, i.e.
\begin{align}
	\label{eq:dissipation_scaling_logdE}
	\beta  Q_{\text{diss}}^* = \mathcal{O}\left(\frac{1}{\log d_{\ms E}}\right).
\end{align}
Let us now describe a protocol that achieves this scaling. Consider $\ms E$ to be a system composed of $n = \log d_{\ms E} / \log d_{\ms S}$ subsystems, $\ms E = \ms E_1 \ms E_2 \ldots \ms E_n$, where each $\ms{E}_i$ has the same dimension $d_{\ms S}$ as the system of interest $\ms S$ and its own Hamiltonian $\hat H_{\ms{E}_i}$. The global unitary, denoted $U_{\pi}$, is a sequential swap between $\ms S$ and each of $\ms{E}_i$, leading to the overall action
\begin{align}\label{eq:unitaryemb}
	U_{\pi} \ket{i_0}_{\ms S} \ket{i_1}_{\ms E_1}  \ket{i_2}_{\ms E_2} \ldots \ket{i_n}_{\ms E_n} &= \\ \ket{i_n}_{\ms S} \ket{i_0}_{\ms E_1} & \ket{i_1}_{\ms E_2} \ldots \ket{i_{n-1}}_{\ms E_n}\nonumber 
\end{align}
This is the same unitary used in describing the relationship between correlated-catalytic and multi-copy transformations (Sec. \ref{subsec:partial-order-regularization}), and in the construction of universal approximate catalysts (Sec. ~\ref{subsubsec:approx-cat-construction}). To further perform the minimization in Eq.~\eqref{eq:dissipation_achievable}, one optimizes over local Hamiltonians $\hat H_{\ms{E}_i}$, which yields the scaling in Eq. (\ref{eq:dissipation_scaling_logdE}). In fact, for any thermal environment composed of non-interacting subsystems, the scaling of $Q_{\text{diss}}$ is at most linear in $\log d_{\ms{E}}$ \cite{reeb2014improved}.

The above protocol for minimizing $Q_{\text{diss}}^*$ can be viewed, alternatively, as the problem of finding the least-disturbed approximate catalyst where the thermal environment itself is treated as the approximate catalyst or embezzler. 
This follows from the correspondance $\gamma_{\ms{E}} \propto \log \hat H_{\ms{E}}$, which makes optimizing Eq. (\ref{eq:dissipation_achievable}) over $\hat H_{\ms{E}}$ equivalent to optimizing it over all possible density operators. To demonstrate this, let us consider a thermal environment whose Hamiltonian $\hat H_{\ms{E}}$ is chosen such that its thermal state coincides with the universal approximate catalyst described in Sec.~\ref{subsubsec:approx-cat-construction}. Specifically, choose $\hat H_{\ms E} = -\frac{1}{\beta} \log \rho_{\ms E}$, where $ \rho_{\ms E}$ is given by

\begin{align}
	\rho_{\ms E} = \frac{1}{n-1} \sum_{i=1}^{n-1} \rho^{\ot i}_{\ms S} \ot \sigma_{\ms S}^{\ot (n-i)}.
\end{align}
This means that $\gamma_{\ms{E}} = \rho_{\ms E}$ and furthermore $d_{\ms E} = d_{\ms S}^{n}$. Applying the unitary $U_{\pi}$ from Eq.~\eqref{eq:unitaryemb} to $\rho_{\ms{S}}\otimes \rho_{\ms E}$ gives 
\begin{align}
	\sigma_{\ms{SE}} = U_{\pi} (\rho_{\ms S} \ot \rho_{\ms E})U^{\dagger}_{\pi} = \sigma_{\rm S} \ot \sigma_{\ms E},
\end{align}
with  $\Delta(\rho_{\ms E},\sigma_{\ms E}) \leq 1/(n-1)$, as in the analysis leading to Eq.~\eqref{eq:omegaprime_and_omega}.
For generic states $ \rho_{\ms E}, \sigma_{\ms E}$ of the above structure, numerical analysis shows that $ D(\sigma_{\ms E}\|\rho_{\ms E}) $ scales as $O(1/\log d_{\ms E})$ as in \eqref{eq:dissipation_scaling_logdE}. 
This brings us to the conclusion that processes in which heat dissipation obeys this scaling are actually performing thermal embezzlement on the thermal environment. This connection between heat dissipation and thermal embezzlement, to the best of our knowledge, has not yet been appreciated in the literature. Furthermore, it is an interesting open question whether one can use stricter notions of catalysis to engineer better, i.e. less dissipative, thermodynamic protocols.

An alternative approach for lowering heat dissipation $Q_{\text{diss}}$ was proposed in Ref.~\cite{Henao2022}. 
There, one includes an additional system $\ms{C}$ (the catalyst) so that the joint system $\ms{SCE}$ evolves according to Eq. \eqref{eq:general_thermo_map_cat}. 
The unitary $U$ is then composed of two steps, i.e. $U = V_{\ms{SCE}}(V_{\ms{SE}}\ot \mathbb{1}_{\ms{C}})$. 
That is, one first implements a unitary $V_{\ms{SE}}$ on the system and the thermal environment, obtaining $\sigma_{\ms{SE}} = V_{\ms{SE}}(\rho_{\ms{S}} \ot \gamma_{\ms{E}})V_{\ms{SE}}^{\dagger}$. 
Then one applies a second unitary $V_{\ms{SCE}}$ to the joint system $\ms{SCE}$, leading to the final global state $\sigma'_{\ms{SCE}} = V_{\ms{SCE}}(\sigma_{\ms{SE}} \ot \omega_{\ms C})V_{\ms{SCE}}^{\dagger}$. 
The second interaction (which involves the catalyst) is introduced to mitigate the heat $Q_{\text{diss}}$ dissipated to the thermal environment during the first interaction. 
Ref.~\cite{Henao2022} showed that, for any {correlated} state $\sigma_{\ms{SE}}$ obtained in the first step of the protocol,  i.e., if $I(\ms S:\ms E)_{\sigma}>0$, there always exists a unitary $V_{\ms{SCE}}$  that can reduce the local entropy of the environment, without changing the local states of the system and the catalyst, that is $\sigma'_{\ms S} = \sigma_{\ms{S}}$ and $\sigma'_{\ms C} = \sigma_{\ms{C}}$. 
As a consequence, the final dissipated heat $Q'_{\text{diss}}$ computed for the state $\sigma_{\ms{SCE}}'$ is lower than $Q_{\text{diss}}$. Based on this observation, they argued that the use of a catalyst allows to mitigate heat dissipation in the process of information erasure (i.e. when $\Delta H_{\ms{S}} < 0$). 
This reduction in heat dissipation can be understood as a consequence of broadcasting correlations from $\ms{SE}$ to the joint system $\ms{SCE}$ which results in lowering the entropy of the enviornment, i.e. $H(\sigma_{\ms{E}}') - H(\sigma_{\ms{E}}) < 0$.

Let us conclude this section by mentioning another situation in which the thermal environment is formally treated as a (strict) catalyst. 
In the theory of open quantum systems, the Born-Markov approximation gives rise to a Markovian master equation for the dynamics on the thermal environment $\ms{E}$ \cite{breuer2002theory}.  
Moreover, when combined with a rotating wave approximation, it implies that the steady-state of the dynamics on $\ms{E}$ is a thermal state. 
The Born-Markov approximation is often stated as assuming that the thermal environment remains thermal at all times, and does not build up correlations with the system $\ms{S}$, i.e. that the evolution from Eq. (\ref{eq:general_thermo_map}) can be approximated as
\begin{align}
	U(\rho_{\ms{S}} \ot \gamma_{\ms{E}})U^{\dagger} \approx \sigma_{\ms{S}} \ot \gamma_{\ms{E}}.
\end{align}	
From the point of view of this review it may hence be seen as a form of catalysis. 
We emphasize, however, that the validity of the approximation in fact only requires that \emph{(a)} the correlations between subsystems $\ms{S}$ and $\ms{E}$ are negligible as measured by the interaction, and \emph{(b)} that the two-time correlation functions of the operators of the interaction term on the thermal environment match the thermal ones with high accuracy. These requirements may already be fulfilled for a thermal environment that is only locally thermal, for example in the sense of the eigenstate thermalization hypothesis, see Refs.~\cite{deutsch1991quantum,Srednicki1994,Polkovnikov2011,Gogolin2016,DAlessio2016}. 
Strict catalysis is therefore not actually required for the approximation to holds. However, the fact that the thermal environment has to remain locally thermal may be seen as a form of approximate catalysis of the resource content (the local thermality) of the bath.

\subsubsection{Work extraction}\label{subsub:workext}
Consider the general map from Eq. (\ref{eq:general_thermo_map}) specified to the case when the thermal environment is not included in the global dynamics, i.e. $U = U_{\ms{S}} (t) \ot \mathbb{1}_{\ms{E}}$. 
For that, consider a quantum system $\ms S$ that evolves under a Hamiltonian $\hat H_{\ms S}(t) = \hat H_{\ms S} + \hat V(t)$, where $\hat V(t)$ is a time-dependent external potential. 
From control theory we know that any interaction $\hat V(t)$ can be formally written using a unitary $U_{\ms{S}} = \mathcal{T} \exp\left(-\frac{i}{\hbar} \int_{0}^t \hat H_{\ms{S}}(s) \,\text{d}s  \right) $, where $\mathcal{T}$ is the time-ordering operator. 
We are going to consider cyclic process during which the system $\ms{S}$ performs work on an external source.  
A thermodynamic process is called \emph{cyclic} when the system is coupled at time $t = 0$ to an external source of work, and fully decouples at time $t = T$, i.e. $V(0) = V(T) = 0$. Since the system generally does not return to its initial state at time $T$, this allows to perform thermodynamic work on the external source. 
The maximal amount of work that can be extracted under a cyclic process is known as the \emph{ergotropy}, and can be written as \cite{Allahverdyan2004}
\begin{align}
	\label{eq:ergotropy_def}
	W_{\mathcal{U}}(\rho_{\ms S}) := \max_{U \in \mathcal{U}_{\ms{S}}} \Tr \left[\hat H_{\ms S} (\rho_{\ms S} - \sigma_{\ms{S}} )\right],
\end{align}
where $\sigma_{\ms{S}} := U\rho_{\ms S} U^{\dagger}$ and $\mathcal{U}_{\ms{S}}$ stands for the set of all unitary operators on $\cDS$. We also identified $\rho_{\ms S} \equiv \rho_{\ms S}(0)$. Quantum states from which no work can be extracted using unitary operations are called \emph{passive states}. 
The concept of passivity is fundamental in thermodynamics: for instance, it singles out Gibbs thermal states as the only passive states that remain passive when taking arbitrary many uncorrelated copies ~\cite{Pusz1978,Lenard1978}.

The above scenario can be extended to the case when an additional system (a catalyst) is used to increase the amount of work extracted in the process. Notice that due to the Basic Lemma of catalysis (see Lemma \ref{lem:fundamental_nogo}), no strict catalyst exists that can be used to increase the ergotropy in Eq. \eqref{eq:ergotropy_def}. 
However, this is no longer the case for the other, less stringent, notions of catalysis. To see this, we will now consider an extension of ergotropy to the case of correlated-catalysis. For that, one can define the \emph{correlated-catalytic ergotropy} as 
\begin{align}
	W_{\mathcal{U}}^{\,\text{CC}}(\rho_{\ms S}) &:= \max_{U \in \mathcal{U}_{\ms{SC}}}\, \max_{\omega_{\ms C}}\,\, \Tr[\hat H_{\ms S} ( \rho_{\ms S}- \sigma_{\ms{S}})],
\end{align}
where $  \sigma_{\ms{SC}} = U(\rho_{\ms S}\otimes\omega_{\ms{C}})U^\dagger $ and $\sigma_{\ms{C}}= \omega_{\ms C}$. It was first observed in \cite{sparaciari_energetic_2017} that passive states undergoing unitary evolutions can be activated when a suitable catalyst is used. 
In particular, \cite{sparaciari_energetic_2017} showed for three-level systems and \cite{wilming_entropy_2020,lipka-bartosik_catalytic_2021} showed in general that
\begin{align}
	W_{\mathcal{{U}}}^{\,\text{CC}}(\rho_{\ms S}) = \Tr[\hat H_{\ms S}(\rho-\gamma_{\beta^*}(\hat H_{\ms S})) ],
\end{align}
where $\gamma_{\beta^*}(\hat H_{\ms S})$ is the unique Gibbs state at inverse temperature $\beta^*$ satisfying $H(\rho_{\ms S}) = H(\gamma_{\beta^*}(\hat H_{\ms S}))$.
What one observes here is that the catalyst allows to evolve the main system in a way that doesn't keep its spectrum invariant, which allows to reach a Gibbs state with the same entropy. This final state is optimal, in the sense that among all states of the same entropy $ \gamma_{\beta^*}(\hat H_{\ms S}) $ always achieves the minimum of average energy. 

The above results were further generalized in \cite{lipka-bartosik_catalytic_2021} to general resource theories $\mc R=(\mc S,\mc O)$ that satisfy certain reasonable assumptions \footnote{The result holds for any resource theory $\mc R=(\mc S,\mc O)$ in which the free operations $\mc{O}$ allow for permuting identical subsystems and condition operations on classical information (see also Sec. \ref{subsubsec:multi-copy}).} More concretely, one can consider two very different notions of generalized ergotropy, one being achievable under correlated-catalytic operations, that is
\begin{align}
	\hspace{-8pt} W^{\,\text{CC}}_\mathcal{O}(\rho_{\ms S}) := &\max_{\mathcal{F} \in \mathcal{O}}\, \max_{\omega_{\ms C}}\,\, \Tr[\hat H_{\ms S} ( \rho_{\ms S}- \sigma_{\ms{S}})]\\
	&\,\text{such that} \,\, \sigma_{\ms{SC}} \!=\!	\mathcal{F} (\rho_{\ms S}\otimes\omega_{\ms{C}}) \,\text{and}\, \sigma_{\ms{C}}= \omega_{\ms C},\!
\end{align}
and the other being the asymptotic rate of extractable work, $W_{ \mathcal{O}}^\infty(\rho_{\ms S}) := \lim_{n \rightarrow \infty} W_{\mathcal{O}}(\rho^{\ot n})/n$,
which is computed using the total Hamiltonian of $n$ independent copies of the quantum system $\ms{S}$, that is $\hat H_{\ms{S}}^n := \sum_{i=1}^n \hat H_{\ms{S}_i}$ with $ \hat H_{\ms{S}_i} \equiv \mathbb{1}_{/\ms{S}_i} \ot \hat H_{\ms{S}_{i}}$.   

It then follows from the partial-order regularization using correlated catalysis described in Sec.~\ref{subsec:partial-order-regularization} that $ W^{\,\text{CC}}_\mathcal{O}(\rho_{\ms S}) = W_{ \mathcal{O}}^\infty(\rho_{\ms S})$. In other words, the use of an appropriate catalyst allows to extract exactly the same amount of work from a passive state as on average (i.e. per copy) in the limit of asymptotically many copies.
In fact, even though $W_{\mc{O}}^{\text{CC}}$ and $W_{\mc{O}}^{\infty}$ have been defined according to the Hamiltonian of the system $\hat H_{\ms{S}}$, the techniques used to prove this fact are also applicable to other observables, like particle number or the overlap with a fixed quantum state. In this sense, correlated catalysis allows to activate passive states with respect to arbitrary observables, as long as the corresponding asymptotic rate is larger than the single-copy one \cite{lipka-bartosik_catalytic_2021}.

At this point, one should note that the ergotropy $W_{\mathcal{U}}$ defined in Eq. (\ref{eq:ergotropy_def}) is not the sole existing quantification of extractable work in thermodynamics. While it quantifies thermodynamic work as the difference in average energy induced on the main system, this energy difference is not stored explicitly in another physical system (a battery). In fact, due to the conservation of energy, the surplus energy must be implicitly released into some external environment which, in general, can be very difficult to access in a future use.

The natural question that arises is whether the energy extracted from a quantum system can be stored in a way that is fully usable in the future. Ref.~\cite{gallego_thermodynamic_2016} argued from operational principles, that one should always model this "battery" or "work-storage device" explicitly in a way compatible with the intended future use of the energy. 
In a similar spirit, explicit battery models are often used to define work on a single-shot level in the context of the resource theory of athermality, see e.g. \cite{brandao_second_2015,faist2015minimal}.
Here, storing an amount of energy $E$ corresponds to preparing an appropriate battery system in an energy eigenstate of energy $E$, modelling a deterministic energy storage. For more details see section \ref{subsubsec:batteries}.

\subsubsection{Thermal operations}
\label{subsub:CTOs} 
Let us now consider the system of interest $\ms{S}$ interacting with both the thermal ($\ms{E}$) and the catalytic ($\ms{C}$) environments according to Eq. (\ref{eq:general_thermo_map_cat}). Assume furthermore that there are no external sources of work, so that the total energy is conserved and the global unitary $U$ satisfies $[U, \hat H_{\ms{S}} + \hat H_{\ms{E}} + \hat H_{\ms{C}}] = 0$. 
This scenario is the main subject of the resource theory of athermality, which we briefly introduced in Sec. \ref{subsubsec:athermality} 
\footnote{The resource theory of athermality was originally introduced with a number of additional technical assumptions about the spectrum of the  thermal environment (i.e., its Hamiltonian $H_\ms{B}$), see \cite{Janzing2000} and \cite{Horodecki2013,brandao_second_2015}. 
	However, one can equivalently consider all possible Hamiltonians for the thermal environment and then show that no advantage is obtained  in terms of possible state transformations.}

The first steps of explicitly studying catalysis in the resource theory of athermality were done in Ref.~\cite{brandao_second_2015}. 
There the authors raised the question of whether it is possible to characterize all possible state transformations on the system of interest $\ms{S}$ when both the catalytic and thermal environments can be chosen arbitrarily (i.e. up to their respective constraints). 
The authors started out from the state transition conditions for thermal operations (thermo-majorization) introduced in Section \ref{subsubsec:athermality}  and showed that it can be relaxed by strict catalysis into a set of monotonic entropic conditions, known as the \emph{generalized second laws} of quantum thermodynamics.  

More concretely, for energy-incoherent states that fulfill $[\rho_{\ms S},\hat H_{\ms S}]=0$, they identified a family of monotones 
\begin{equation}\label{eq:genfreeE}
	F_\alpha^\beta \left(\rho_{\ms S}, \hat H_{\ms S}\right):= F^{\beta}_{\text{eq}}(\hat H_{\ms{S}}) +  \frac{1}{\beta} D_\alpha\left(\rho_{\ms S} \| \gamma_\beta(\hat H_{\ms S})\right),
\end{equation}
where $ D_\alpha(\cdot\|\cdot)$ are the (classical) R{\'e}nyi divergences, $F_{\text{eq}}^{\beta}(H) := -\frac{1}{\beta}\log Z_\beta(H)$  is the (equilibrium) Helmholz free energy and $Z_\beta(H) := \Tr e^{-\beta H}$ is the partition function. 
When the initial state $\rho_{\ms S}$ and target state $ \sigma_{\ms{S}} $ are energy-incoherent, the necessary and sufficient conditions for a strictly catalytic state transition $\rho_{\ms{S}} \AC{{\rm TO}}~ \sigma_{\ms{S}}$ are then given by 
\begin{equation}\label{eq:genfreeE_conds}
	F_\alpha^\beta \left(\rho_{\ms S}, \hat H_{\ms S}\right) > F_\alpha^\beta \left(\sigma_{\ms S}, \hat H_{\ms S}\right) \quad \text{for} \quad  \alpha\in \mathbb{R}\setminus\{0\}.
\end{equation}

By relaxing the notion of catalysis from strict to arbitrarily strict catalysis (recall Sec.~\ref{sub:arb_strict}), all conditions with $\alpha<0$ can be dropped. 
Refined alternative statements and proofs of this result have since been done in Refs. \cite{gour2021entropy,rethinasamy_relative_2020}. The fully-quantum case where $\rho_{\ms{S}}$ or $\sigma_{\ms S}$ are energy-coherent is still open, even in the case when no catalyst is used. 
In the non-catalytic case several necessary monotones are known \cite{brandao_second_2015,lostaglio2015description,lostaglio_quantum_2015} but the minimal sufficient set of complete monotones remains undetermined, with the exception of when the system $\ms{S}$ is a qubit \cite{Cwiklinski2015}. 

For energy-incoherent states in the i.i.d. limit, the generalized free energies from Eq. (\ref{eq:genfreeE}) for all $\alpha \in \mathbb{R}$ converge to a single quantity, namely the non-equilibrium free energy 
(see also \eqref{eq:free-energy}),
\begin{align}
	F^\beta (\rho_{\ms S},\hat H_{\ms S}) &= \lim_{\alpha\rightarrow 1} F_\alpha^\beta(\rho_{\ms S},\hat H_{\ms S}).
\end{align}
As a consequence, both thermo-majorization conditions and the generalized second laws converge to a single inequality, namely the monotonicity of the non-equilibrium free energy.
This convergence is consistent with our understanding from macroscopic thermodynamics, strengthening the role of thermal operations as a physical description of the microscopic regime.

Historically, the resource theory of thermodynamics served as a natural starting ground for considering different relaxations to the notion of catalysis. 
It was envisioned in \cite{brandao_second_2015} that catalysts describe the apparatus used to control the thermodynamic process, such as a clock (a reference frame) that keeps track of the duration in which an interaction Hamiltonian is switched on (see also Sec.~\ref{subsec:circumventing_conservation_laws}).

At that time the existing research on quantum reference frames highlighted the central issue of their inevitable degradation, i.e. that the state of the reference frame does not return exactly to its original state whenever information about the main system is inferred from the reference frame \cite{poulin2007dynamics}, as opposed to the ideal setting~\cite{Bartlett2007,PhysRevA.82.032320}. 
This was part of the original motivation to investigate the robustness of the state transition conditions (\ref{eq:genfreeE_conds}) under small errors in the catalyst. 
It was realized that earlier results by \cite{van_dam_universal_2003,leung_characteristics_2014} imply the emergence of embezzlement in catalytic thermal operations, resulting in the breaking down of the generalized second laws of thermodynamics -- even the macroscopic second law singled out in the i.i.d regime. 
Indeed, the construction of Sec.~\ref{subsubsec:approx-cat-construction} also applies to thermal operations. This undesirable effect is unphysical and indicates that more care needs to be taken when allowing for an error in the catalyst. 

Ref. \cite{brandao_second_2015} then identified three different regimes of catalysis, corresponding to 1) embezzling, 2) approximate catalysis where the remaining monotone is the non-equilibrium free energy, and 3) strict catalysis.
In particular, the second regime can be seen as the intersection of approximate \emph{and} correlated catalysis (which, interestingly, used the same catalyst construction as discussed in Sec. \ref{subsubsec:multi-copy} and Sec. \ref{subsubsec:approx-cat-construction}). 
In this regime \cite{brandao_second_2015} assumed that the disturbance of the catalyst, in trace distance, decreases with the size of the catalyst. 
This allows to show that the non-equilibrium free energy $F^{\beta}$ is the only relevant monotone that fully characterizes state transitions under this relaxation of catalysis. 
On the other hand, the problem of thermal embezzlement was later picked up in Ref.~\cite{ng_limits_2015}, which derived the minimal achievable error as a function of the catalyst dimension. Ref.~\cite{ng_limits_2015} also showed that additional physical constraints on the catalyst, such as an upper bound on the average energy, allowed for the recovery of various monotones.

A conceptually different functioning of the catalyst was probed in Ref.~\cite{woods_resource_2019}, which analysed the problem of implementing energy-preserving unitaries for thermal operations in an \emph{autonomous} manner, i.e. via a time-independent global Hamiltonian, {see also \cite{Malabarba2015,Silva2016,Erker2017,woods2019autonomous} and Sec.~\ref{subsec:asymmetry}}.
To achieve this, they explicitly modelled a quantum clock that provides the necessary timing information, and consequently degrades in the process. In the language of this review such systems can be viewed as approximate catalysts. The authors of Ref.~\cite{woods_resource_2019} start by showing that some amount of back-action on the clock system will always be inevitable, unless the clock has an unphysical  Hamiltonian (such as those without a ground state). 
The authors then proceed by 
bounding the maximal error (in trace distance) that can be induced on the clock, denoted by $\varepsilon_{\rm emb}$, and show that such a back-action would necessarily affect the accuracy of the implemented state transition on the system. 
With this, they derived an upper bound on the resolution error $ \varepsilon_\mathrm{res} $ of achieving the desired transformation on the system, i.e. the minimum achievable trace distance between a final system state with the original target state. 
The upper bound on $ \varepsilon_\mathrm{res}$ in particular depends on the dimension of the system, the catalyst used, and the clock considered, and also the allowed error $ \varepsilon_\mathrm{emb}$. The main result of Ref.~\cite{woods_resource_2019} is that both errors $\varepsilon_\mathrm{res}$ and $\varepsilon_\mathrm{emb} $ can together vanish even when the clock system has a physical Hamiltonian (what the authors coin as a quasi-ideal clock). This result implies that all CTOs (i.e. thermal operations with an exact catalyst) can be 1) implemented in a fully autonomous manner with an inexact catalyst, where 2) the state transition conditions in Eq.~\eqref{eq:genfreeE} remain robust.

While strict catalysis was first investigated in the context of entanglement transformations under LOCC, the concept of correlated catalysis is much better founded in the framework of thermal operations.  This is because majorization relation, via Nielsen's theorem, fully describes state transitions only in the case of pure bipartite systems. In contrast, its thermodynamic analog, relative majorization, characterizes state transitions under thermal operations also between mixed states. This implies that it can also be used to address the extensions of exact catalysis, e.g. allowing for residual correlations between $\ms S$ and $\ms C$ (correlating catalysis) or within different parts of $\ms C$ (marginal-correlated catalysis), even after tracing out the environment $\ms{E}$.
In \cite{lostaglio_stochastic_2015,muller2016generalization}, an explicit multipartite catalyst was constructed such that by allowing for final correlations to be created between the partitions, it enabled the bypassing of generalized second laws from Eq. (\ref{eq:genfreeE_conds}) for all $\alpha$ except {$\alpha=0$ and $\alpha =1$.} The guiding intuition behind this work is the fact that the relative entropy $D_{1}$ is the unique quantity out of the whole family of R{\'e}nyi divergences $D_{\alpha}$ which is super-additive (with the exception of $ D_0 $ which is unstable under perturbations of the state). 
Subsequently, Ref.~\cite{gallego_thermodynamic_2016} argued from operational principles that one should consider the notion of correlated catalysis.
Ref.~\cite{henrik_wilming_axiomatic_2017} showed that both correlated catalysis and marginal-correlated catalysis single out the non-equilibrium free energy $F^{\beta}$ as the unique continuous and additive monotone.
A natural question that emerged from these considerations is whether monotonicity of the non-equilibrium free energy alone could be sufficient for the convertability via correlated catalysis.
In other words, it was conjectured that (when allowing for an arbitrarily small error on the final state, which eliminates the constraint arising from $F^\beta_0$)
\begin{align}\label{eq:correlating-catalysis-free-energy}
	F^\beta(\rho_{\ms S},\hat H_{\ms S}) \geq F^\beta(\sigma_{\ms S}, \hat H_{\ms S}) \Rightarrow\rho_{\ms S} \AC{\rm TO}{\rm corr.} \sigma_{\ms S}. 
\end{align}
For energy-incoherent states, this was proven by \cite{muller_correlating_2018} and \cite{rethinasamy_relative_2020}, while it is wrong for general coherent states, since thermal operations cannot build up coherences. 
This can be circumvented by enlarging the set of free operations from thermal operations to Gibbs-preserving (GP) operations \cite{Faist2015}. \cite{shiraishi_quantum_2021}, using the construction discussed in Sec.~\ref{subsec:partial-order-regularization}, showed that Eq. \eqref{eq:correlating-catalysis-free-energy} is true for general quantum states if $\mathrm{TO}$ is replaced by $\mathrm{GP}$. 
These results show that correlating catalysis can lift the infinite family of second laws \eqref{eq:genfreeE_conds} to just the monotonicity of the standard non-equilibrium free energy. 
It is worth noting that in these results the catalyst can remain arbitrarily little correlated to the system of interest provided that its dimension is large enough.

As we conclude this section, it is important to note that identifying an appropriate state of a catalyst for thermal operations is generally a challenging task. To date, there is no comprehensive method that can determine the state of the catalyst required by a given state transformation. This issue was partially addressed in Ref. \cite{lipka-bartosik_all_2021}. Specifically, for a given pair of energy-incoherent states $\rho_{\ms{S}}$ and $\sigma_{\ms{S}}$ that satisfy Eq. \eqref{eq:genfreeE_conds} for $\alpha \geq 0$, they observed that a quantum state with randomly distributed occupations in the energy basis can, with high probability, serve as an approximate catalyst for the transition $\rho_{\ms{S}} \xrightarrow{\text{TO}} \sigma_{\ms{S}}$. This success probability furthermore increases for catalysts with a larger dimension. Moreover, Ref. \cite{lipka-bartosik_all_2021} claimed to formally prove that for any pair of energy-incoherent states satisfying the second laws of Eq. \eqref{eq:genfreeE_conds} for $\alpha \geq 0$, any energy-incoherent quantum state can act as an approximate catalyst for the transition, as long as enough copies of the catalyst are available. Unfortunately, it was later discovered that the proof of this statement has a gap (i.e. Lemma $4$ therein is not valid), which leaves this claim unresolved.

\subsubsection{Thermal operations with limited control}

Thermal operations capture generic energy-preserving interactions between a system and its surroundings. 
However, sometimes it can be appealing to study specific classes of thermal operations which can admit more straightforward experimental implementations, or better reflect actual control capabilities. 
For this reason, two main classes of thermal processes have been studied in the recent years. 
The first class is known as \emph{elementary thermal operations} (ETOs)~\cite{lostaglio2018elementary}, which are system-bath interactions that can be decomposed into sequences of two-level operations on the system. 
The second class is called Markovian thermal processes (MTPs)~\cite{lostaglio2022continuous}, which is a hybrid approach that combines Markovian master equations with resource-theoretic formulations to describe memoryless system dynamics. While these two classes of operations are formally different from each other, for energy-incoherent initial states, it was shown in Ref.~\cite{lostaglio2022continuous} that the set of states reachable via MTPs is a subset for that of ETOs. 

Strictly catalytic versions of the above classes have been studied~\cite{son2022catalysis}. 
The first point of observation is that in both ETOs and MTPs, while a lot of freedom is still given in preparing and using thermal Gibbs states, such states are technically no longer free states in the usual resource-theoretic sense. Take ETOs for example: For each $2$-level operation on the system of interest, one is allowed to append any Gibbs state and turn on the system-bath interaction; however before proceeding to the next step, this Gibbs state is traced out and replaced with a fresh copy of another Gibbs state. Because of this, in contrast to thermal operations, even a catalyst in the Gibbs state can be useful \cite{korzekwa2022optimizing,son2023hierarchy}.   
As for the second point, since the gap between ETO/MTPs and TOs exist due to a restricted form of Markovianity of the former, intuitively this gap should be closed when a proper memory system is allowed (e.g. in the form of a catalyst in a Gibbs state).

The above idea was formalized and rigorously examined in two studies using different approaches. Ref.~\cite{czartowski2023thermal} developed a protocol  using MTPs, combined with an explicit modelling of a memory acting as a (strict) catalyst. They showed that, in the infinite temperature limit, MTPs can approximately reach \emph{all} states that are reachable by TOs, with the precision increasing with the size of the catalyst. They also conjectured that the same holds at finite temperatures. 
Ref.~\cite{son2023hierarchy} took a different approach, by focusing on ETOs. The authors show that, for any TO, it is possible to decompose the global energy-conserving unitary into a sequence of 2-level unitaries acting on the system up to an arbitrary precision using a Suzuki-Trotter expansion.  They then prove that any TO
can be implemented with arbitrary precision using ETOs with a strict catalyst prepared in a Gibbs state. 
For energy-incoherent initial states, they also show that any final state achievable via strictly catalytic TOs is also achievable via strictly catalytic ETOs and MTPs.


\subsubsection{Work extraction with explicit batteries}
\label{subsubsec:batteries}
The concept of work is not easy to define in the formalism of resource theories. This is because the resource theoretic approach requires us to explicitly model the physical mechanism of storing and supplying thermodynamic work. Recall that in Sec. \ref{subsub:workext} work was modelled via an external potential that performs work on the system. This approach is not sufficient in the resource theoretic mindset where all external resources must be accounted for explicitly. The usual approach of extending the scenario from Eq. (\ref{eq:general_thermo_map_cat}) is to add an explicit battery system $\ms{B}$ that supplies and stores the thermodynamic work associated with the process.

With this in mind, Ref.~\cite{brandao_second_2015} modelled the battery system $\ms{B}$ using a two-level system with an energy gap $w$ initialized in a pure energy eigenstate $\ketbra{0}{0}_{\ms{B}}$. They defined (deterministic) work as the maximum value of $w$ such that the following state transition is possible
\begin{equation}\label{eq:workext_in_CTO}
	\rho_{\ms{S}} \otimes \ketbra{0}{0}_B\AC{{\rm TO}}{\text{arb.}}~ \sigma_{\ms{S}} \otimes \ketbra{1}{1}_B,
\end{equation}
where a positive value of $w$ corresponds to extracting work from the system, while a negative value  corresponds to supplying work for the transformation. Using the generalized second laws Eq.~\eqref{eq:genfreeE_conds} in the above state transition it can be found that the optimal value of $w$ that allows to bring the system from one energy-incoherent state $\rho_{\ms{S}}$ to another one $\sigma_{\ms{S}}$ is given by
\begin{equation}\label{eq:work_distance}
	W_{\text{dist}} =\inf_{\alpha\geq 0} \left[F^\beta_\alpha\left(\rho_{\ms S}, \hat H_{\ms S}\right)-F^\beta_\alpha\left(\sigma_{\ms S}, \hat H_{\ms S}\right)\right],
\end{equation}
which is also known as the \emph{work distance} \cite{brandao_second_2015}. 
The immediate question arises as to whether the exact form of $W_{\text{dist}}$ has to depend strongly on the battery model. 
\cite{brandao_second_2015} shows that $W_{\text{dist}}$ can be defined more generally as the ability of extracting and storing work in a pure (i.e. zero entropy) state. 
They illustrate that by considering another battery model, namely the purity battery \cite{bennett1982thermodynamics,faist2015minimal}, one arrives at Eq.~\eqref{eq:work_distance} when applying Landauer's erasure to relate purity of the battery with thermodynamic work. 

Building on these results, Ref.~\cite{woods_maximum_2019} studied the scenario in which catalysts are used as controls/machines undergoing cyclic process in the presence of two baths at different inverse temperatures $\beta_{\text{c}}$ and $\beta_{\text{h}}$. More specifically, they considered a setting where the hot bath acts as the background that defines the set of catalytic thermal operation to be $\mathrm{TO}(\beta_{\text{h}})$, and the other bath is of a finite size, initialized in the Gibbs state of inverse temperature $\beta_{\text{c}} $ with its respective Hamiltonian $H_{\ms S}$.

Casting the problem of work extraction into a question of battery state preparation allowed the authors of Ref.~\cite{woods_maximum_2019} to establish a generic framework that provides systematic tools for investigating the quality of extracted energy. 
They characterized numerous types of extracted work, according to how the entropy of the battery compares with the extracted average energy. 
The first type, which the authors dubbed perfect work, includes the situation in \eqref{eq:workext_in_CTO}, where there is strictly zero increase in entropy of the battery state. 
A slight relaxation of this condition, called near-perfect work, allows for an increase in entropy of the battery as long as it is arbitrarily small compared to the amount of extracted work. 
\cite{woods_maximum_2019} analyzed the impact of the generalized free energies (\ref{eq:genfreeE}) on such a heat engine setting, and concluded that they place fundamental limitations on the maximal efficiency $\eta$ for the case of a single qubit of the system, where
\begin{equation}
	\eta = \frac{W_{\rm ext}}{\Delta H},
\end{equation}
$W_{\rm ext}$ being extracted, near-perfect work, and $ \Delta H $ being the amount of energy change in the hot bath.
In particular, the existence of generalized free energies emerging as monotones apart from the well-known non-equilibrium free energy implied that for certain qubit Hamiltonians, the maximum achievable efficiency falls strictly below the Carnot limit $\eta_{\rm C} = 1-\frac{\beta_{\text{h}}}{\beta_{\text{c}}}$. 
It was furthermore shown that allowing for correlations to build up between the system and the catalyst (i.e relaxing the notion of catalysis from strict to correlated) cannot be used to bypass this limitation.
This seems to contradict the insight from Sec. (\ref{subsubsec:multi-copy}) which states that correlated catalysis completely regularizes the single-shot partial order into the asymptotic one. 
However, it should be noted that \cite{woods_maximum_2019} derived the optimal efficiency of work extraction in the limit of vanishing extractable work (Carnot efficiency limit). 
In this limit the correlations have to vanish sufficiently quickly relative to the extracted work in order to achieve the Carnot efficiency. 
In particular, they have to vanish quicker than the scaling we discussed in Sec. (\ref{subsubsec:multi-copy}) when analyzing the single-shot regularization. 
These additional requirements imply that the single-shot effects, as captured by the generalized second laws, may still persist even when correlations between the system and the catalyst are considered.

{As discussed in Sec.~\ref{subsub:CTOs}, correlating catalysis can lift the second laws to essentially just the monotonicity of the free energy.}
{Therefore using a correlating catalyst, the amount of extractable work (in the sense of the battery model from Eq. Eq.~\eqref{eq:workext_in_CTO}) can be made arbitrarily close to
	\begin{equation}\label{eq:work-correlating-catalysis}
		W = F^\beta\left(\rho_{\ms S}, \hat H_{\ms S}\right)-F^\beta \left(\sigma_{\ms S}, \hat H_{\ms S}\right), 
	\end{equation} 
	which can be significantly larger than $W_\textrm{dist}$. This is true either for incoherent states when considering thermal operations or for arbitrary states when considering the more general Gibbs-preserving maps as free operations.}

Another perspective on catalysis in thermodynamics was unveiled in the context of fluctuation theorems \cite{Jarzynski_1997,Crooks_1999,Tasaki2000}, where the notion of work is defined by the two-point measurement scheme, namely the difference between energy measurement outcomes before and after a thermodynamic process,
\begin{equation}
	W = E_f - E_i.
\end{equation} 
More specifically, fluctuation theorems rule out the possibility of extracting macroscopic (i.e. scaling extensively with the number of copies of the system) amount of work with non-negligible probability in the case of unitary evolution of a system prepared in a Gibbs state. 
Ref.~\cite{boes_by-passing_2020} proposed a protocol that uses correlated-catalysis in thermal operations to extract a macroscopic amount of work  with a non-negligible probability, in such as way that the averaged second law still holds. Crucially, in order to achieve macroscopic work, the logarithm of the dimension of the catalyst and its free energy must also grow linearly with the number of system's copies. The general results of Ref.~\cite{rubboli2021fundamental} applied to this particular scenario imply that this extensive scaling of the catalyst's dimension and free energy is necessary.

{Moving away from resource theoretic considerations, \cite{rodriguez_catalysis_2023} considered the charging of quantum batteries in a concrete physical model, where a quantum battery (harmonic oscillator) is charged by coupling it to a harmonic oscillator driven by a (classical) laser-field.
	They showed that adding an intermediate coupler can enhance the energy transfer from charger to battery while at the same time removing the requirement to fine-tune the laser-frequency to the involved coupling strengths.  
	At the same time the intermediate coupler itself stores almost no energy and hence effectively works as a catalyst.}

\subsubsection{Cooling}
\label{subsubsec:cooling}
The efficient cooling of quantum systems has been a central topic in quantum thermodynamics, given both its fundamental importance stemming from the third law of thermodynamics~\cite{nernst1907experimental}, as well as its significance for quantum technologies~\cite{bloch2008many,langen2015ultracold}.

Let us once again consider the situation described by Eq. (\ref{eq:general_thermo_map_cat}), in the special case of thermal operations, where the unitary $U$ conserves global energy. The problem of cooling down a quantum system $\ms{S}$ prepared in some quantum state $\rho_{\ms{S}}$ can be cast as a state transition problem, i.e.
\begin{align}
	\rho_{\ms{S}} \AC{\mathcal{O}_{{\rm TO}(\beta)}}{\rm arb.}~ \sigma_{\ms{S}},
\end{align}	
where if one sets $ \sigma_{\ms{S}} = \gamma_{\beta'}(\hat H_{\ms{S}})$, then the goal of cooling is defined as achieving a large inverse temperature $\beta'$. Moreover, the assumption that $ \sigma_{\ms{S}}$ must necessarily be thermal can be dropped -- generic measures for cooling can also be used, such as simply maximizing the overlap of $\sigma_{\ms{S}}$ with the system ground state, or minimizing its average energy. 

To understand the ultimate limits of cooling under the framework of thermal operations, one can apply the generalized second laws from Eq. (\ref{eq:genfreeE_conds}) to the above state transition and study the amount of (non-equilibrium) resources required to perform this task. This approach was taken in Ref.~\cite{wilming2017third}, which showed that while general state transitions are governed by a continuous family of monotones indexed by $\alpha \in \mathbb{R}$,
in the context of cooling down to $T' := (\beta')^{-1} \rightarrow 0$, only one monotone remains relevant, called vacancy:	
\begin{equation}
	\mathcal{V}_\beta(\rho_{\ms S}, \hat H_{\ms S}):=D\left(\gamma_\beta(\hat H_{\ms S}) \| \rho_{\ms S}\right).
\end{equation}
The vacancy is related to the (Petz) R\'enyi divergences by 
\begin{equation}
	\left.\frac{\partial}{\partial_\alpha}\right|_{\alpha=0} D_\alpha\left(\rho_{\ms S} \| \gamma_\beta(\hat H_{\ms S})\right)=V_\beta(\rho_{\ms S}, \hat H_{\ms S}).
\end{equation}
Since $\gamma_\beta(\hat H_{\ms S})$ is of full rank, the vacancy diverges for states $\rho_{\ms S}$ which do not have full rank (in particular, for states approaching zero temperature). 
Using a resourceful state $ \rho_{\ms{R}} $ with Hamiltonian $\hat H_{\ms R}$, one can cool an initially thermal system $\ms{S}$ to a generic target state $\sigma_{\ms{S}}$ only if 
\begin{equation}\label{eq:thirdlaw}
	\mathcal{V}_\beta\left(\rho_{\ms R}, \hat H_{\ms R}\right) \geq \mathcal{V}_\beta\left(\sigma_{\ms S}, \hat H_{\ms S}\right).
\end{equation}
Furthermore, Eq.~\eqref{eq:thirdlaw} is achievable up to a correction factor that vanishes when $\sigma_{\ms{S}}$ is a thermal state of some temperature $T_{\ms S}\rightarrow 0$,  meaning that this inequality fully characterizes the fundamental limits of cooling. 
Similar results were previously established in Ref.~\cite{Janzing2000}, for the special case where the target system is composed of qubits, and assuming that the resource $\rho_{\ms R}$ has an i.i.d. structure. 
In this respect, \cite{wilming2017third} establishes the third law of thermodynamics in the fully single-shot regime by allowing strict catalysis.  
\cite{wilming2017third} also discusses the robustness of Eq. \eqref{eq:thirdlaw} under approximate catalysis with respect to errors in the catalyst measured by a vacancy change. 
In particular, closeness of the catalyst to its original state in terms of vacancy change is a sufficiently strong measure to maintain the robustness of Eq.~\eqref{eq:thirdlaw}.

The advantages of correlated catalysis can also be seen in cooling. These were discussed in \cite{boes_von_2019}, specially for the special case where the system Hamiltonian is irrelevant. 
In such situations, cooling usually then refers to the process of preparing states of high purity. 
We have seen in Sec.~\ref{subsec:QM} that correlated catalysis within unitary quantum mechanics allows for a system $\ms S$ to go from $\rho_{\ms{S}}$ to $\sigma_{\ms{S}}$ as long as $H(\rho_{\ms{S}}) > H(\sigma_{\ms{S}})$. 
Ref.~\cite{boes_von_2019} observes that this result can be applied to cooling in an interesting way. 
To illustrate this, consider $\ms S$ to contain two uncorrelated systems (in general, states of full rank), $\rho_{\ms{S}} = \rho_{\ms{S1}}\otimes \rho_{\ms{S2}}$, such that both $ H(\rho_{\ms{S1}}), H(\rho_{\ms{S2}}) < (\log d_{\ms S})/4$, and $d_{\ms S1} = d_{\ms S2} = \sqrt{d_{\ms S}}$. Then one can use a unitary transformation involving a correlating catalyst to bring the system to any final state of the form
\begin{equation}
	\sigma_{\ms{S}} = \sigma_\ms{S1}\otimes \frac{\mathbb{1}}{d_{\ms S2}},
\end{equation}
where $\sigma_{\ms{S1}}$ has to be a state of full rank, but importantly can have arbitrarily small amount of entropy.
In particular, $\sigma_{\ms{S1}}$ can be arbitrarily close to a pure state.
Similar cooling rates were obtained using  protocols based on the idea of data compression ~\cite{schulman1999molecular,boykin2002algorithmic,dahlsten2011inadequacy}, for the special case of qubits where instead of a catalyst, \emph{many identical copies} of the system are used, e.g. taking $ \rho_{\ms S}^{\otimes n} $ qubits and distilling roughly $Rn$ many close-to-pure qubits in the asymptotic limit, where
\begin{equation}
	R \approx 1-H(\rho_{\ms S})
\end{equation}
and we use $\log(2)=1$. 
In contrast, with the help of a correlating catalyst (albeit probably with high dimension), one can perform the distillation and reach the optimal rate using as few as 2 copies of the system. 
The authors in \cite{boes_von_2019} cautioned that the end result, while similar to the asymptotic limit in terms of rates, has a subtle difference -- a repeated use of the catalyst creates correlations among all the final cooled states. 
Therefore, if the catalytic cooling protocol is repeated $n$ times with the same catalyst, the final state on the resulting $n$ copies of $\ms{S1}$ is not given by ${\sigma_{\ms{S1}}}^{\otimes n}$. 

Moving towards a more practical mindset, explicit protocols for catalytic cooling have been proposed \cite{henao_catalytic_2021}, for systems of low dimensions. This work relates the capacity for cooling to the passitivity of non-equilibrium states, when joined with a bath or a catalyst. Such a setting has also been studied in the context of ergotropy \cite{sparaciari_energetic_2017}.
The protocols designed in Ref. \cite{henao_catalytic_2021} are particularly interesting because of their use of catalysts of small dimension (a qubit) and relatively simple operations (up to three-body interactions). 

The starting point of the investigation in Ref. \cite{henao_catalytic_2021} is as follows: assuming that $\rho_{\ms S} = \sum_i p_{\ms S}{(i)} \ketbra{i_{\ms S}}{i_{\ms S}}$ is a passive state with respect to Hamiltonian $\hat H_{\ms S}$ (hence expressed diagonally in its energy eigenbasis $\lbrace\ket{i_{\ms S}}\rbrace_i$), under what conditions does the bipartite state $\rho_{\ms S}\otimes \rho_{\ms B}$ remain passive w.r.t.  $\hat H_{\ms S}$? Here, $ \rho_{\ms B} $ represents a thermal state/bath that is used in the process of cooling. The authors show that the passitivity of such a bipartite state is determined fully by

\begin{equation}\label{eq:passitivity_bipartite}
	\frac{p_{\ms S}(i)}{p_{\ms S}(i+1)} \geq \frac{p_{\ms B}{(1)}}{p_{\ms B}{(d_{\ms B})}}\qquad \text{ for all } i,
\end{equation}
where $ p_{\ms B}{(i)} $ are the descendingly ordered eigenvalues of the state $ \rho_{\ms B} $. From this analysis, one can observe that if $ \rho_{\ms B} $ is set to be a Gibbs state with temperature $T$, then the higher $T$ is, the smaller the RHS of Eq.~\eqref{eq:passitivity_bipartite}, in particular it converges to 1 in the infinite temperature limit. On the other hand, the left hand side is always greater or equal to 1 due to passitivity of $\rho_{\ms S}$. In summary, Eq.~\eqref{eq:passitivity_bipartite} will eventually be satisfied for some high enough temperature $ T $, making $\rho_{\ms S}\otimes\rho_{\ms B}$ passive w.r.t. $ \hat H_{\ms S} $ -- disallowing further cooling of ${\ms S}$. 

Eq.~\eqref{eq:passitivity_bipartite} spells out the limitations of achievable cooling by means of unitary operations, having access only to a thermal reservoir. The authors then proceed to show that an extension of Eq.~\eqref{eq:passitivity_bipartite} holds for an \emph{ancilla-assisted} cooling process: in other words, passitivity including a general ancilla $ \rho_{\ms C} $ is again fully determined by 
\begin{equation}\label{eq:passitivity_bipartite_catalytic}
	\frac{p_{\ms S}{(i)}}{p_{\ms S}{(i+1)}} \geq \frac{p_{\ms B}{(1)}}{p_{\ms B}{(d_{\ms B})}}\cdot \frac{p_{\ms C}{(1)}}{p_{\ms C}{(d_{\ms C})}}\qquad \text{ for all } i,
\end{equation}
where $ p_{\ms C}{(1)}$ and $p_{\ms C}{(d_{\ms C})} $ are the maximum and minimum eigenvalues of $ \rho_{\ms C} $ respectively. Now, if we had a cooling protocol that preserves the ancilla, we then also have a catalytic cooling process. Hence, \cite{henao_catalytic_2021} focus on a particular subset of cooling processes that can be decomposed in a two-step process: the first step corresponds to identifying a unitary $U_{\rm cool}$, that decreases the average energy of system $ {\ms C} $ whenever Eq.~\eqref{eq:passitivity_bipartite_catalytic} is satisfied. This process $ U_{\rm cool} $ in general alters $ \rho_{\ms C} $, and hence is subsequently followed up by a restoring unitary, i.e. a unitary $ V_{\rm res} $ that brings the local reduced state of the catalyst back to $ \rho_{\ms C} $. In particular, the authors show that $ U_{\rm cool} , V_{\rm res} $ can be chosen such that they act on orthogonal subspaces. More importantly, the authors showed that the question of whether an ancilla can truly be recovered is fully characterized: by constructing an explicit $ V_{\rm res} $ that consists of a series of partial 2-level swaps, and checking if each of these swaps satisfy a simple condition that can be checked algorithmically.


\subsection{Asymmetry and coherence}
\label{subsec:asymmetry}
In section~\ref{sec:rt} we introduced the resource theory of asymmetry $\mathcal{R}_{\text{asym}}$, which models the physical constraints arising from conservation laws and the corresponding symmetries present in the system. To recap, the resource theory of asymmetry deals with scenarios where every physical system $\ms S$ carries a unitary representation $W_S$ of some group $G$ (e.g. $G = \textsf{SO}(3)$ corresponds to rotations in three dimensional Euclidean space $\mathbb{R}^{3}$), and
the allowed dynamics has to be covariant with respect to this group.
If $G$ is a (connected) Lie group, then the representations of its generators can be physically interpreted as conserved quantities, as any covariant unitary channel has to leave them invariant.
This can be seen as an expression of Noether's theorem. 
Indeed, if $W_S(\exp(h)) = \exp(-\mathrm i \hat H_{\ms S})$ for a generator $h$ in the Lie algebra of $G$, and $\mc F[(\cdot)] = U(\cdot)U^\dagger$ is a covariant unitary quantum channel on $\ms S$, then 
\begin{align}\label{eq:energy-conservation}
	{	\mc F^*[\hat H_{\ms S}]	= U^\dagger \hat H_{\ms S} U = \hat H_{\ms S}.}
\end{align}
For example, if $G=\mathbb R$ is the group of time-translations, then $\hat H_{\ms S}$ can be identified with the Hamiltonian of the system and \eqref{eq:energy-conservation} expresses energy-conservation. 
A quantum system $\ms S$ is then said to be resourceful if its state $\rho_{\ms S}$ is asymmetric with respect to the group representation, meaning that there exists at least one $g\in G$ such that $W_{\ms S}(g) \rho_{\ms S} W_{\ms S}(g)^\dagger \neq \rho_{\ms S}$.

Closely related to the resource theory of asymmetry is the \emph{resource theory of coherence} $\mathcal{R}_{\rm coh}$, introduced  in Refs. \cite{Baumgratz2014,Levi2014}, and reviewed in Ref.~\cite{Streltsov2017}.
From the point of view of the resource theory of coherence, the degree of coherence of a quantum state measures the magnitude of the off-diagonal terms in some fixed basis $\mathbb H = \{\ket i\}$.
For example, a qubit in a state $(\ket 0 + \ket 1)/\sqrt{2}$ is highly coherent with respect to the basis $\mathbb{A}=\{\ket 0,\ket 1\}$, but completely incoherent with respect to the basis $\mathbb{B}=\{\ket+,\ket-\}$. 
From the point of view of the resource theory of asymmetry, we can consider a representation of the group $G = \mathbb R$ of time-translations induced by a non-degenerate Hamiltonian $\hat H_{\ms S}$ with eigenbasis $\mathbb H$ via $t\mapsto \exp(-\mathrm i \hat H_{\ms S} t)$. 
Then a quantum state is said to be coherent if it is asymmetric with respect to the group representation.  Consequently, in this case the set of free states in $ \mathcal{R}_{\rm coh} $ coincides with the set of free states in $\mathcal{R}_{\rm asym} $. 

Despite the above similiarity of  $\mathcal{R}_{\text{coh}}$ and $\mathcal{R}_{\text{asym}}$ for the case $G=\mathbb R$, a crucial distinction arises in the allowed free operations. 
While in $\mathcal{R}_{\text{asym}}$ only covariant operations are allowed, in $\mathcal{R}_{\text{coh}}$ the so-called \emph{incoherent operations} \cite{Baumgratz2014,Levi2014,Streltsov2017} are allowed,
which is a strictly larger set than the corresponding covariant operations \cite{Marvian2016a}.
More specifically, in the resource theory of asymmetry, the set of free operations depends on the concrete representation of time-translations (induced by the Hamiltonian $\hat H_{\ms S}$ above), and not just on the basis $\mathbb H$. 
This distinction has significant consequences for catalysis, as we will review in section~\ref{subsec:no-broadcasting}. 
For critical comparisons of different approaches to establishing a physically consistent resource theory of coherence, see Refs. \cite{Marvian2016} and \cite{Marvian2016a, Chitambar2016}.

Quantitative measures of coherence and asymmetry can be defined in multiple ways. 
The most common one is perhaps the \emph{relative entropy of asymmetry} defined in Sec.~\ref{subsubsec:rt-asymmetry}. 
In the special case of coherence with respect to a basis $\mathbb H=\{\ket i\}$ it reduces to the  \emph{relative entropy of coherence}, given by \cite{Aberg2006}
\begin{align}\label{eq:REC}
	\mc A_{\mathbb H}(\rho) = H(\mc G_{\mathbb H}[\rho]) - H(\rho), 
\end{align}
where the twirling channel is given by $\mc G_{\mathbb H}=\sum_i \proj{i}\rho\proj{i}$.
The relative entropy of coherence is monotonic under incoherent operations, satisfies several desireable properties of a proper coherence measure \cite{Baumgratz2014}, and has a plausible operational  interpretation as the maximal distillable coherence from the state \cite{Winter2016}. 

\subsubsection{Apparent violations of conservation laws} 

\label{subsec:circumventing_conservation_laws}
Coherence (or asymmetry) present in a quantum state can be viewed as a resource that enables implementing coherent operations on quantum systems. 
More specifically, in order to bring an atom $\ms S$ from the ground state $\ket{g}_{\ms{S}}$ to the excited state $\ket{e}_{\ms{S}}$ using time-translation covariant (energy-conserving) operations, we need an auxiliary system $\ms{E}$ from which the missing energy can be taken. 
However, if we wanted to turn $\ms S$ into a coherent superposition of the ground and excited states, i.e. $(\ket{g}_{\ms S} + \ket{e}_{\ms S})/\sqrt{2}$, then providing the necessary energy is not enough to perform this transformation. 
Under the constraint of energy conservation, the unitarity of quantum mechanics forbids this type of operation if the atom and the energy reservoir $\ms E$ have initially definite energies. 
More specifically, there does not exist a unitary operator $U$ for which
\begin{align}
	U\ket{g}_{\ms S}\otimes\ket{E}_{\ms E} = \frac{1}{\sqrt 2}(\ket{e}_{\ms S} + \ket{g}_{\ms S})\otimes\ket \psi_{\ms E}
\end{align}
and at the same time $[U,\hat H_{\ms S}+\hat H_{\ms E}]=0$, where $\ket{E}_{\ms E}$ is an eigenstate of $\hat H_{\ms E}$ and $\ket{\psi}_{\ms E}$ is arbitrary. This problem can be circumvented by introducing a reservoir of coherence, i.e. a large quantum system prepared in coherent superposition of many energy levels. 
In the case of equally spaced energy levels, such a reservoir of coherence constitutes a \emph{phase reference} and is usually achieved with the help of electromagnetic fields, e.g. laser beams \cite{mandel_wolf_1995} or radio-frequency fields \cite{Vandersypen2005}. 
The coherence between orthogonal states of the field is a resource that enables mixing different energies in a coherent, rather than probabilistic, way. 
Similarly, instead of energy preservation, we can consider the more general notion of covariance with respect to some group $G$ as constraint (such as spatial rotations). 
In this case the coherence reservoir needs to be replaced by a more general (quantum) reference frame for the respective group (see the review \cite{Bartlett2007} and references therein). 

It has been observed that coherence reservoirs are useful in the context of thermodynamics \cite{Janzing2003a,Vaccaro2008,Janzing2006,lostaglio2015description,Malabarba2015,mitchison2015coherence,lostaglio_quantum_2015,korzekwa_extraction_2016,woods_resource_2019}, 
where a coherence reservoir is often understood as a clock that provides timing-information.
For example, the non-equilibrium free energy of a quantum system in a state $\rho_{\ms S}$ and Hamiltonian $\hat H_{\ms S}$ (see \eqref{eq:free-energy}) naturally splits into an incoherent and a coherent part:
\begin{align}
	F^\beta (\rho_{\ms S},\hat H_{\ms S}) = F_\beta(\mc G_{\mathbb H}[\rho_{\ms S}],\hat H_{\ms S}) + \frac{1}{\beta} \mc A_{\mathbb H}(\rho_{\ms S}),
\end{align}
{with $\mc A_{\mathbb H}(\rho_{\ms S})$ as in \eqref{eq:REC} for the eigenbasis $\mathbb H$ of $\hat H_{\ms S}$}. Without the coherence reservoir only the incoherent part, i.e. $F^\beta(\mc G_{\mathbb H}[\rho_{\ms S}],\hat H_{\ms S})$, can be extracted as work \cite{Janzing2006}. A coherence reservoir can thus be seen as a thermodynamic resource.
Similarly, general quantum reference frames are thermodynamic resources for other conserved quantities such as angular momentum or spin.

A "good" quantum reference frame $\ms C$ can be used in such a way that its state changes only minimally while at the same time lifting all conservation laws. Indeed, it was proven first in Refs.~\cite{Tajima2018,Tajima2020} (and later independently in Refs.~\cite{Chiribella_2021,Yang_2022}) for the case of coherence, that whenever a quantum reference frame $\ms C$ in a pure state can be used to approximately implement a unitary dynamics on $\ms S$ to high precision using covariant dynamics on $\ms S\ms C$, then it can be done in such a way that the state on $\ms C$ only changes little. Ref.~\cite{Luijk2023} treats the same problem for the resource theory of asymmetry with respect to arbitrary groups. See also \cite{Tajima2018,Tajima2020,takagi_universal_2020,Tajima2021,Tajima2022} for trade-off relations between the precision of the implemented unitary and the required resources on the quantum reference frame.

The underlying reason for these results is that the (approximately) coherent dynamic on $\ms S$ requires that (almost) no information about the state on $\ms S$ flows to the environment, since otherwise $\ms S$ would necessarily become entangled to the environment. 
This is the \emph{information-disturbance tradeoff} in quantum mechanics \cite{Fuchs1996a,Fuchs1996,Kretschmann2008}. 
Using the terminology we introduced in Sec. \ref{subsec:app_cat} the coherence reservoir essentially acts as an approximate catalyst. Moreover, if it can be used to implement \emph{arbitrary} state-transformations with high precision, it is even an embezzler: 
By Stinespring's theorem, any quantum channel on $\ms S$ can be formally implemented to high precision using unitaries acting on an extended Hilbert space of systems $\ms E\ms S$, where $\ms E$ is an auxiliary environment with dimension $d_{\ms S}^2$.
Thus, if $\ms C$ can be used to implement arbitrary unitaries on $\ms S \ms E$ to high precision, it is an embezzler. 
We will now present a simple example illustrating these considerations.

Suppose we wish to implement some quantum channel $\mc E$ on a two-level system $\ms S$ with Hamiltonian $\hat H_{\ms S}=\omega \dyad{1}_{\ms S}$, where $\ket{0}_{\ms S}$ and $\ket{1}_{\ms S}$ denote its ground and excited states, respectively. 
Moreover, suppose we are only able to apply energy conserving unitaries, but have access to a harmonic oscillator $\ms C$ with matching frequency $\omega$ (in resonance with $\ms{S}$) and energy eigenstates $\ket{n}_{\ms C}$. 
We now provide a construction that manages to apply $\mc E$ to arbitrary accuracy on $\ms S$ if the initial state on the oscillator is sufficiently coherent. 
The discussion follows \cite{aberg_catalytic_2014}, see also \cite{messinger_coherence_2020}.
Consider a family of subspaces $\mc H_n$ 
spanned by the "logical" states $\ket{\bar 0}_{n},\ket{\bar 1}_{n}$ defined for all $n \geq 1$ by
\begin{align}
	\ket{\bar 0}_n := \ket{0}_{\ms S}\ket{n}_{\ms C},\qquad \ket{\bar 1}_n := \ket{1}_{\ms S}\ket{n-1}_{\ms C}.
\end{align}
Any unitary acting separately on subspaces $\mc H_n$ 
is energy-preserving (note that $\ket{0}_{\ms S}\ket{0}_{\ms C}$ must be an eigenstate of the unitary). 
However, within each subspace $\mc H_n$, the unitary is unconstrained. 
Let us now further introduce an auxiliary system $\ms E$ with basis $\{\ket{\alpha}_{\ms E}\}$ that serves as the dilating system in the Stinespring dilation of $\mc E$ with unitary $V$ on $\ms E\ms S$.
For simplicity we assume that it has a trivial Hamiltonian (or we only access to one energy subspace of some larger system). 
{To implement $V$, and thereby $\mc E$,} we define an energy-preserving unitary $U$ on $\ms E\ms S\ms C$ by its matrix elements via
\begin{align}
	\bra\alpha_{\ms E}\bra{\bar k}_n U \ket\beta_{\ms E}\ket{\bar l}_n := \bra\alpha_{\ms E} \bra k_{\ms S} V \ket \beta_{\ms E}\ket l_{\ms S}
\end{align}
for all $n$ and $\bra\alpha_{\ms E}\bra 0_{\ms S} \bra 0_{\ms C} U \ket\beta_{\ms E} \ket 0_{\ms S}\ket 0_{\ms C} = \delta_{\alpha\beta}$. 
In other words, $U$ acts as $V$ on subspaces $\mc H_{\ms E}\otimes \mc H_n$ and trivially on $\mc H_{\ms E} \otimes \mathrm{span}\{\ket{0}_{\ms S}\ket{0}_{\ms C}\}$.

Let us now write a general pure product state on $\ms E \ms S\ms C$ as 
\begin{align}
	\ket\chi_{\ms E}	\ket{\psi}_{\ms S}\ket{\phi}_{\ms C} &= \ket\chi_{\ms E}\left(\alpha \ket{0}_{\ms S}+\beta \ket{1}_{\ms S}\right)\left(\sum_{n=0}^\infty c_n \ket{n}_{\ms C}\right)\nonumber \\
	&= \sum_{n=1}^\infty c_n\left( \alpha\ket\chi_{\ms E}\ket{\bar 0}_n + \beta \frac{c_{n-1}}{c_n}\ket\chi_{\ms E}\ket{\bar 1}_n\right)\nonumber\\ 
	&\quad+ c_0 \alpha\ket\chi_{\ms E}\ket{0}_{\ms S}\ket{0}_{\ms C}.
\end{align}
Then if $c_{n-1} \approx c_n$ for many values of $n$ (which requires each $c_n$ to be small) we get 
\begin{align}
	\label{eq:asym_coh_actual}
	U\ket\chi_{\ms E}\ket{\psi}_{\ms S}\ket{\phi}_{\ms C} \approx \left(V\ket\chi_{\ms E}\ket{\psi}_{\ms S}\right)\ket{\phi}_{\ms C}. 
\end{align}
In other words: if the state on $\ms C$ is spread smoothly over many energy levels, then we can implement the desired unitary on $\ms S\ms E$ to an arbitrary accuracy.
Concretely, we may choose a coherent superposition of $M$ energy levels occupied starting from the energy level $\ket{n_0}_{\ms{C}}$, that is $c_n = 1/\sqrt{M}$, where $n_0+M \geq n\geq n_0>0$, and $c_n = 0$ otherwise. 
The fidelity between the {two sides of} Eq. (\ref{eq:asym_coh_actual}) is then given by
\begin{align}\label{eq:assym_coh_flat_2}
	\left|\bra\chi_{\ms E}\bra\psi_{\ms S}\bra\phi_{\ms C}U^\dagger (V\ket\chi_{\ms E}\ket\psi_{\ms S})\ket\phi_{\ms C}\right|^2 \geq 1 - \frac{2}{M}. 
\end{align}
Similarly, we can choose the initial state on $\ms C$ to be a coherent state
\begin{align}
	\label{eq:coh_state}
	\ket\alpha_{\ms C} := \e^{-\frac{|\alpha|^2}{2}}\sum_{n=0}^\infty \frac{\alpha^n}{\sqrt{n!}}\ket{n}_{\ms C},
\end{align}
and obtain an error that decreases with the mean photon number $|\alpha|^2$, which is also the variance of the photon number. 
We emphasize that this effect is made possible because the state of a large quantum system may change arbitrarily little, as measured by the norm on vectors in the Hilbert space, despite the fact that the mean of some observable changes by a finite amount (here the energy). 
In this way, a finite amount of energy may be transferred coherently to $\ms S$ while perturbing the state on $\ms C$ arbitrarily little. 
In this sense, conservation laws can be \emph{apparently} violated. 
This is closely related to the continuity properties of extensive quantities described in more detail in Sec. (\ref{subsubsec:non-continuity}).

The seminal paper \cite{aberg_catalytic_2014} further showed that if the flat state of \eqref{eq:assym_coh_flat_2} is used, then despite the fact that the state on $\ms C$ is changed, it can be reused a finite number of times $n_0$ to implement \emph{exactly the same} quantum channel on $n_0$ uncorrelated copies of $\ms S$ (while also making use of $n_0$ copies of $\ms E$; see also Sections \ref{sec:concept} and \ref{subsec:corr_cat}).
In this special sense the coherence stored in $\ms C$ is catalytic (or \emph{repeatable}, as dubbed in Ref.~\cite{korzekwa_extraction_2016}), while its quantum state is not.
\cite{aberg_catalytic_2014} referred to this phenomenon as "catalytic coherence".
However, as emphasized by \cite{vaccaro_is_2018,messinger_coherence_2020}, the resulting state on the $n_0$ copies of $\ms S$ is not uncorrelated. Thus this procedure does not implement the quantum channel $\mc E^{\otimes n_0}$. 
The reason for this is that the system $\ms C$ mediates correlations that are being build up between all the copies of $\ms E \ms S$, so that their final joint-state is \emph{not} given by $(V\rho_{\ms E\ms S}V^\dagger)^{\otimes n_0}$, see also Section~\ref{subsubsec:cooling}.  

\subsubsection{Correlations and broadcasting of quantum information}
\label{subsec:no-broadcasting}

The no-cloning theorem \cite{Park1970,Wootters1982,Dieks1982} states that it is impossible to prepare exact, uncorrelated copies of an unknown quantum state: There is no quantum channel $\mc E$ such that $\mc E[\rho_{\ms S}] = \rho_{\ms S}\otimes\rho_{\ms S}$ for all states $\rho_{\ms S}$ on some fixed system $\ms S$.
More generally, the no-broadcasting theorem \cite{Barnum1996} implies that there is no quantum channel $\mc E$ from $\ms S$ to two copies of $\ms S$, such that $\tr_{\ms S_i}[\mc E[\rho_{\ms S}]]=\rho_{\ms S}$ for $i=1,2$.  
That is, it is impossible to convert an unknown quantum state $\rho_{\ms S}$ to one where both marginals coincide with $\rho_{\ms S}$ ("broadcasting $\rho_{\ms S}$"), but the two marginals are possibly correlated. (Indeed, if this was possible, one could clone pure quantum states!) 
In fact, this feat is not even possible when $\mc E$ is required to implement broadcasting for a pair of non-commuting quantum states (instead of all quantum states).

One may ask whether, instead of broadcasting the full quantum state, one could broadcast only some aspect of it, such as some amount of its coherence or, more generally, its asymmetry with respect to some group $G$. 
For that, in the language of the resource theory of asymmetry, we ask whether it is possible to implement the state transition 
\begin{align}\label{eq:asym-broadcasting}
	\rho_{\ms C} \xrightarrow[\mathcal{O}_{\rm asymm.}]{} \sigma_{\ms C_1\ms C_2},
\end{align}
such that $\sigma_{\ms C_1} = \rho_{\ms C}$ and $\sigma_{\ms C_2} \notin \mathcal{S}_{\rm asymm.}$, i.e. the final state on subsystem $\ms C_2$ is not symmetric with respect to $G$. 
Observe that if this were possible, then there would have to exist a covariant quantum channel $\mc F$ on $\ms C_1 \ms C_2$, and a \emph{symmetric} state $\rho_{\ms C_2}$ such that
\begin{align}
	\mc F[\rho_{\ms C} \otimes \rho_{\ms C_2}] = \sigma_{\ms C_1 \ms C_2}.
\end{align}
Indeed, if $\mc F_1$ was the covariant quantum channel implementing \eqref{eq:asym-broadcasting}, then the covariant quantum channel $\mc F = \mc F_1 \circ \tr_{\ms C_2}$ would work for any state $\rho_{\ms C_2}$. 
But, in this case, we can see that $\ms C=\ms C_1$ acts as a catalyst that becomes correlated to $\ms S=\ms{C}_2$.
Therefore we can equivalently ask whether there is a correlated-catalytic transformation of the form
\begin{align}\label{eq:no-broadcasting-corr}
	\rho_{\ms S} \AC{\mathcal{O}_{\rm asymm.}}{\rm corr.} \sigma_{\ms S},
\end{align} 
where $\rho_{\ms S}\in\mc S_{\rm asymm.}$ and $\sigma_{\ms S}\notin \mc S_{\rm asymm.}$.
We refer to the above task as \emph{broadcasting of asymmetry}: the possibility to use an asymmetric quantum state (a quantum reference frame) in a catalytic manner to transform a symmetric state into an asymmetric state. 
Similarly, we call the special case of strict catalysis \emph{cloning of asymmetry}:
\begin{align}\label{eq:no-broadcasting-strict}
	\rho_{\ms S} \AC{\mathcal{O}_{\rm asymm.}}{} \sigma_{\ms S}.
\end{align} 
It has been established that broadcasting and cloning of asymmetry are both impossible in the case where $G$ is a connected Lie group \cite{lostaglio_coherence_2019,marvian_no-broadcasting_2019}, see also \cite{Janzing2003}. In particular, broadcasting of coherence is impossible. 
With this, it is tempting to think that the power of correlating catalysis is the same as that of strict catalysis in the resource theory of asymmetry of a connected Lie group. 
Interestingly, this is not true: there are states $\rho_{\ms S}, \sigma_{\ms S} \notin  \mathcal{S}_{\rm asymm.} $ such that Eq.~\eqref{eq:no-broadcasting-corr} is possible while \eqref{eq:no-broadcasting-strict} is not \cite{ding_amplifying_2021}. 
That is, when acting on asymmetric states, it is generally useful to build up correlations between the system and the catalyst.
Moreover, when allowing for marginal-correlated catalysis (see section~\ref{subsec:corr_cat}), catalysts can essentially lift all restrictions on state transitions in the case of coherence \cite{Takagi2022}. Refs.~\cite{Shiraishi2023,Kondra2023} show that essentially the same statement is true for correlated catalysis while still allowing for arbitrary small correlations between system and catalyst. 
	Thus, as long as an arbitrary small amount of coherence is present in the initial state, correlated catalysis can lift \emph{all restrictions} imposed by demanding covariance under time-translation. 
	In other words, an arbitrary small amount of initial coherence can completely circumvent the no-broadcasting theorem for quantum coherence. 
	It is an open problem to generalize the result to non-commutative Lie groups.

In fact, establishing correlations between the catalyst and \emph{some} other degrees of freedom is not only useful, but necessary. 
Ref.~\cite{Luijk2023} shows that, if a system $\ms C$ acts as a useful catalyst for a state transition on $\ms S$ in the resource theory of asymmetry for a connected Lie group (e.g., in the case of coherence), then it must necessarily become correlated either to $\ms S$ or to the environmental degrees of freedom $\ms E$ that dilate the covariant quantum channel on $\ms S\ms C$ to covariant unitary dynamics on $\ms E\ms S\ms C$, just as in Lemma~\ref{lem:fundamental_nogo}.
(In fact, the same conclusion can be derived for the case of thermal operations, see Section~\ref{subsubsec:athermality}.)
In particular, this implies that a useful catalyst can never be in a pure state, a result previsouly also shown in Ref.~\cite{ding_amplifying_2021} for the special case of coherence, and in Ref.~\cite{marvian_theory_2013} for the case where the states of interest on $\ms S$ are pure, but the group in question is an arbitrary compact connected Lie group.

This situation is distinctly different from the resource theory of coherence where the full class of incoherent operations are allowed, which renders pure catalysts useful and make it possible to characterize catalytic state transitions on pure states using R\'enyi entropies, just as in the case of LOCC \cite{bu_catalytic_2016}{, because the state-transition conditions are characterized by majorization \cite{du_conditions_2015}. {In the context of incoherent operations Refs. \cite{Liu2020} and \cite{xing2020catalytic} further study strictly catalytic transformations between mixed states}.
	Ref.~\cite{chen_one-shot_2019} study the one-shot distillation of coherence using catalysts in this framework, relate the distillable coherence to the dimension of the catalyst and show that coherence can be embezzled. Somewhat surprisingly, \cite{lami_catalysis_2023} show that neither the (asymptotic) coherence cost nor the distillable coherence under incoherent operations change when strict, correlated or marginal-correlated catalysts are allowed for.}

Ref.~\cite{rubboli2021fundamental} further provides tools to lower bound the dimension of correlated catalysts in certain limiting cases in this setting. 
The distinction in terms of catalysis with pure states between the resource theory of asymmetry and the resource theory of coherence hence mirrors the one between LOSR and LOCC in the context of entanglement (see section~\ref{subsubsec:LOSR}). 

We emphasize that the restriction to connected groups (relevant for the constraints imposed by conservation laws) is necessary for the above results to hold. 
As shown in Ref.~\cite{marvian_theory_2013}, if $G$ is a finite group, and $\rho_{\ms S},\sigma_{\ms S}$ are arbitrary states on the system $\ms S$ with a unitary representation $W_{\ms S}(g)$ of $G$, then there exists a system $\ms C$ with a unitary representation $W_{\ms C}(g)$ of $G$ and a pure state $\proj{\phi}_{\ms C}$ together with a covariant quantum channel $\mc F$ such that
\begin{align}
	\mc F[\rho_{\ms S}\otimes \proj\phi_{\ms C}]=\sigma_{\ms S}\otimes\proj\phi_{\ms C}.
\end{align}
In other words, strict catalysis with pure catalysts can lift all constraints on possible state transitions. 
This is possible because the state $\ket\phi_{\ms C}$ can be chosen to be a perfect quantum reference frame for the group $G$, i.e. a state satisfying $\bra\phi W_{\ms C}(g)\ket\phi = 0$ for all $g\neq 1$. The channel $\mc F$ can then be defined as
\begin{align}
	\mc F[X] &= \sum_{g\in G} \tr[(\id \otimes W_{\ms C}(g)\proj{\phi}_{\ms C}W_{\ms C}(g)^\dagger)X] \times \cdots \nonumber \\
	&\quad \cdots\times W_{\ms S}(g)\, \rho_S'\, W_{\ms S}(g)^\dagger \otimes W_{\ms C}(g)\proj{\phi}_{\ms C}W_{\ms C}(g)^\dagger.
\end{align}
A similar construction can be used to implement any unitary transformation $V$ on $\ms S$ (or $\ms E\ms S$) perfectly and not just approximately as in Section~\ref{subsec:circumventing_conservation_laws}.

\subsubsection{Locality of interactions and conservation laws}

Locality of interactions is a fundamental property of physical systems. In the case of short-rangeed interactions, it implies a finite speed of propagation of information, as highlighted by the Lieb-Robinson bound \cite{lieb1972finite}. In this sense locality puts restrictions on the short-time dynamics and implies that certain unitary evolutions. 
Nevertheless, any unitary time evolution of a composite physical system can be implemented with a time-dependent Hamiltonian using only local interactions, as long as it is allowed to evolve for a sufficiently long time. 
This is the essence of a fundamental result in quantum computing which states that unitary transformation on a composite system can be represented (to arbitrary accuracy) by a quantum circuit consisting only of $2$-local unitary transformations, i.e. unitary transformations only acting non-trivially on at most two subsystems \cite{lloyd1995almost,divincenzo1995two,deutsch1995universality}.

As we have discussed throughout this section, the presence of symmetries in physical systems also puts constraints on their time-evolution as any realizable unitary has to respect the associated symmetry. In view of this, a natural question to ask is whether all symmetric unitaries on a composite system can be generated using local symmetric unitaries, in analogy with the above universality result that holds for general evolutions. 
This question was posed and subsequently answered in the negative in Ref.~\cite{marvian2022restrictions}:
The universality of interactions is no longer valid in the presence of continous symmetries, such as $\mathsf{U}(1)$ or time-translation covariance (energy conservation). 
Generic symmetric unitaries cannot be implemented, even approximately, using local symmetric unitaries.
In other words, a global unitary that obeys a certain symmetry in general cannot be decomposed as a combination of $2$-local unitaries, where each local unitary obeys the corresponding symmetry constraint. 
This implies that, in the presence of locality, symmetries of the Hamiltonian impose extra constraints on the time evolution of the system, which are not captured by the Noether’s theorem.

Interestingly, in certain cases the above no-go theorem can be circumvented using auxiliary systems. 
Such systems have to be prepared in a fixed state, and return to their initial states at the end of the process. 
In the terminology of this review this is an instance of strict catalysis. 
More specifically, Ref.~\cite{marvian2022restrictions} discussed the case of energy-conservation and demonstrated that the constraints resulting from the interplay of locality and energy-conservation can be circumvented provided that the composite system is allowed to interact with a qubit catalyst. 
In this case the use of a catalytic ancilla allows to lift the locality constraint even in the presence of symmetries. 
In this sense, the catalytic system manifestly opens up new dynamical pathways which could not be reached previously. 
A similar result was shown in \cite{Marvian2022} for the case of rotationally symmetric dynamcis and a catalyst of two qubits, see also \cite{Marvian2021qudit}.


\subsection{Continuous-variable systems and quantum optics}
\label{subsec:optics}
In this section, we temporarily move our attention to continuous-variable (CV) systems \cite{lloyd1999quantum,Braunstein_2005}. More specifically, we will consider an $n$-mode bosonic quantum system $\ms{S}_n$ with continuous degrees of freedom. 
A finite set of $n$ (continuous) degrees of freedom can be represented by $n$ pairs of Hermitian operators $\hat{x}_i$ and $\hat{p}_i$ fufilling the canonical commutation relations $[\hat x_i,\hat p_j]=\mathrm i \delta_{ij} \mathbbm{1}$. We also define the (bosonic) annihilation and creation operators $a_i$ and $a_i^{\dagger}$ via $a_i := (\hat{x}_i + i \hat{p}_i)/\sqrt{2}$. 
Following the usual nomenclature of quantum optics, we will refer to the labels $i$ as ``modes'', since they typically correspond to modes of the electromagnetic field.

Of special importance in quantum optics are so-called quadratic Hamiltonians, i.e. Hamiltonians that can be expressed as polynomials of order two in the canonical operators. A standard example of such a Hamiltonian is the $n$-mode quantum harmonic oscillator, i.e. $\hat{H}_{\text{HO}} = \frac{1}{2} \sum_{i=1}^n (\hat{p}_i^2 + \omega_i^2 \hat{x}_i^2) = \sum_{i=1}^n \omega_i (a_i^{\dagger} a_i + \frac{1}{2})$. 
Quadratic Hamiltonians are very common and provide a consistent approximation of quantum dynamics in many experimentally-relevant situations, e.g.  ion traps \cite{paul1990electromagnetic,bruzewicz2019trapped}, opto-mechanical systems \cite{hansch1975cooling,stenholm1986semiclassical}, nano-mechanical oscillators \cite{aspelmeyer2014cavity}, as well as many other systems \cite{bogoliubov1947theory,itzykson2012quantum}. 
Unitaries that can be implemented using quadratic Hamiltonians are known as Gaussian unitaries.

A common tool used to represent quantum states of continuous variable systems is the Wigner function. 
This is a quasiprobability distribution that assigns a value to each point in phase space, allowing to visualize quantum states and observables in a manner that is similar to classical probability distributions in the classical phase space. The Wigner function $W(x, p)$ of a quantum state $\rho$ of a single continuous variable is defined as
\begin{align}\label{eq:w_fun}
	W(x, p) = \frac{2}{\pi} \int_{\mathbb{R}} e^{i 2 p x'} \langle x-x'|\rho|x+x' \rangle \, \text{d}x'.
\end{align}
It is normalized, $\int_{\mathbb{R}^2} W(x, p) \,\text{d}x\, \text{d} y = 1$, and its integral over one canonical coordinate gives the probability of measuring the conjugate coordinate, e.g. $\int_{\mathbb{R}} W(x,p) \,\text{d} p = \langle x|\rho|x \rangle$.
The Wigner function can be defined in a similar manner also for multi-mode systems.

An important family of states encountered in quantum optics are Gaussian states, whose Wigner function is a Gaussian. A typical example are the \emph{coherent states} $\ket{\alpha}$, which are eigenstates of the annihilation operator $a$ and whose wave-function in phase-space is a Gaussian centered at the complex number $\alpha$.
Any Gaussian state has a non-negative Wigner function. Therefore a Wigner function which is negative somewhere neccesarily corresponds to a non-Gaussian state. An important example are Fock states. For a single mode they are the energy eigenstates $\ket{k}$ of the quantum harmonic oscillator $\hat{H}_{\text{HO}}$ and more generally correspond to states with a fixed number of photons. 
In fact, all pure quantum states with a non-negative Wigner function correspond to Gaussian states \cite{Hudson1974,soto1983wigner}.
In other words, all non-Gaussian pure states exhibit Wigner negativity. 
This is not true for mixed states, as indicated by numerous examples \cite{walschaers2021non}. While Gaussian states on multiple modes can exhibit genuine quantum phenomena, such as Bell non-locality \cite{nha2004proposed,garcia2004proposal}, they are in general easier to prepare than non-Gaussian quantum states \cite{Braunstein2005,Weedbrook2012}. 
Certain quantum phenomena, such as entanglement and coherence, can result in negative values of the Wigner function in some regions of phase space, indicating genuine non-classical features of the quantum state \cite{kenfack2004negativity,spekkens2008negativity}. 

Of special importance for quantum optics are Gaussian operations, which map Gaussian states to Gaussian states.
They are usually defined as quantum channels that can be implemented using linear quantum optics, which also makes them relatively easy to manipulate experimentally \cite{Weedbrook2012}. 
A particular example are the Gaussian unitaries mentioned above. Ref.~\cite{gagatsos_majorization_2013} considers the utility of a beam-splitter in creating bipartite pure entangled states, and examines the majorization structure of the corresponding Schmidt coefficients generated. With this, they reveal sets of incomparable states via LOCC which can be activated by catalysis.

The experimental feasibility of Gaussian operations comes naturally with the complementary weakness that Gaussian operations may only simulate a small subset of all quantum channels. 
Since any Gaussian operation acting on Gaussian states can be simulated efficiently on a classical computer \cite{Mari2012}, Wigner negativity is also closely related to computational speed-ups, see e.g. Refs. ~\cite{Veitch2012,PhysRevA.71.042302,PhysRevLett.88.097904}.
 
This mirrors the situation for Clifford operations on discrete systems, see Section~\ref{subsec:computation}. In order to implement more complex quantum operations one therefore has to use non-Gaussian states, motivating to investigate potential  tools to prepare non-Gaussian states.

\subsubsection{Multi-photon catalysis}
When only Gaussian operations are available, e.g. in setups involving linear optical systems, non-Gaussian quantum states become valuable resources \cite{Kok_2007,slussarenko2019photonic}. Because of that, one might wonder whether such expensive states can be prepared only once, and then reused multiple times in some specific task. In other words, one might ask whether non-Gaussian states can serve as useful catalysts. This approach was investigated in Ref.
\cite{Lvovsky2002} who considered a phenomenon they termed "quantum optical catalysis", which has been further generalized to "multi-photon catalysis" in later works \cite{Scheel_2003,Bartley2012,xu_enhancing_2015,Hu2017}, see also Refs. \cite{,zhouentanglement2018,Birrittella_2018,Zhang2021}. 
The basic idea behind (multi-) photon catalysis is that one can build a linear optical process acting on two modes, i.e. $\ms{S}$ and $\ms{C}$, which takes a coherent (Gaussian) state $\ket{\alpha}_{\ms{S}}$ in one mode, and a Fock state $\ket{k}_{\ms{C}}$ of $k$ photons in the second mode, and outputs a non-classical state of light contingent upon measuring $k$ photons in the output mode $\ms{C}$ of the experiment. While very interesting from an experimental point of view, such a process cannot be seen as catalytic in the sense described in this review for two main reasons: 
$(i)$ The density operator describing the measured output port $\ms{C}$ is not given by $\ketbra{k}_{\ms{C}}$, so that there is a non-zero probability that $l \neq k$ photons are measured instead, and $(ii)$ even if $k$ photons are measured on $\ms{C}$, the measurement is typically destructive, so that the photons cannot be reused for further processes. 
Interestingly, these two issues can be resolved when the $k$-photon state $\ket{k}_{\ms{C}}$ is replaced by a two-level system in resonance with the coherent state of light $\ket{\alpha}_{\ms{S}}$. In the next section we describe such an alternative method of preparing non-classical states of light based on generic light-matter interactions. 

\subsubsection{Activation of non-classicality}
\label{subsec:wigner_neg}
Let us now describe another approach for preparing non-classical states of light. The analysis presented here is based on Ref.~\cite{deOliveira2023}. Consider a two-level system interacting with a single mode of light trapped in an optical cavity. Specifically, let $\ms S$ denote the electromagnetic field in the cavity with a bosonic annihilation operator $a$, and let $\ms C$ be a two-level system with energy levels $\ket{g}$ and $\ket{e}$, and the raising/lowering operators $\sigma_+ = \dyad{e}{g} = \sigma_{-}^{\dagger}$. These two systems interact via the Jaynes-Cummings Hamiltonian \cite{Jaynes1963,Larson2021}, which reads 
\begin{align}
	\hat{H}_{\ms{SC}} &= \gamma_{\ms S} a^{\dagger} a + \gamma_{\ms C}  \dyad{e} + \hat H_{\text{int}}\\
	&:=\hat H_{\ms S}+\hat H_{\ms C} + \hat H_{\text{int}},
\end{align}
where $\hat{H}_{\text{int}} := g \left(\sigma_+ a + \sigma_- a^{\dagger} \right)$. The quantity $\gamma_{\ms S}$ is the angular frequency of the mode and $\gamma_{\ms C}$ is the transition frequency of the two-level system.  
We assume that the two-level system is driven on resonance, meaning that $\gamma_{\ms S} = \gamma_{\ms C}=\gamma$. 
This ensures that the unitary evolution $U(t) := e^{-i \hat H_{\ms{SC}} t}$ generated by $\hat{H}_{\ms{SC}}$ conserves the total energy of both systems when considered non-interacting: $[U(t), \hat H_{\ms S}+\hat H_{\ms C}] = 0$.

Let us assume that the electromagnetic field $\ms S$ is initialized in a coherent state $\rho_{\ms S} = \dyad{\alpha}_{\ms S}$, where $\ket{\alpha}$ is defined in Eq. (\ref{eq:coh_state}). 
Suppose that the atom $\ms{C}$ starts in a state $\omega_{\ms C}(\tau)$ being the solution of the operator equation
\begin{align}
	\label{eq:xt}
	\omega_{\ms C}(\tau) = \Tr_{\ms{S}}[U(\tau) \rho_{\ms{S}} \ot \omega_{\ms C}(\tau) U(\tau)^{\dagger}],
\end{align} 
where $\tau \in \mathbb{R}$. In other words, the atom is prepared in a state chosen such that at time $t = \tau$ it returns to its initial state. 

\begin{figure}[t]
	\centering
	\includegraphics[width=\textwidth]{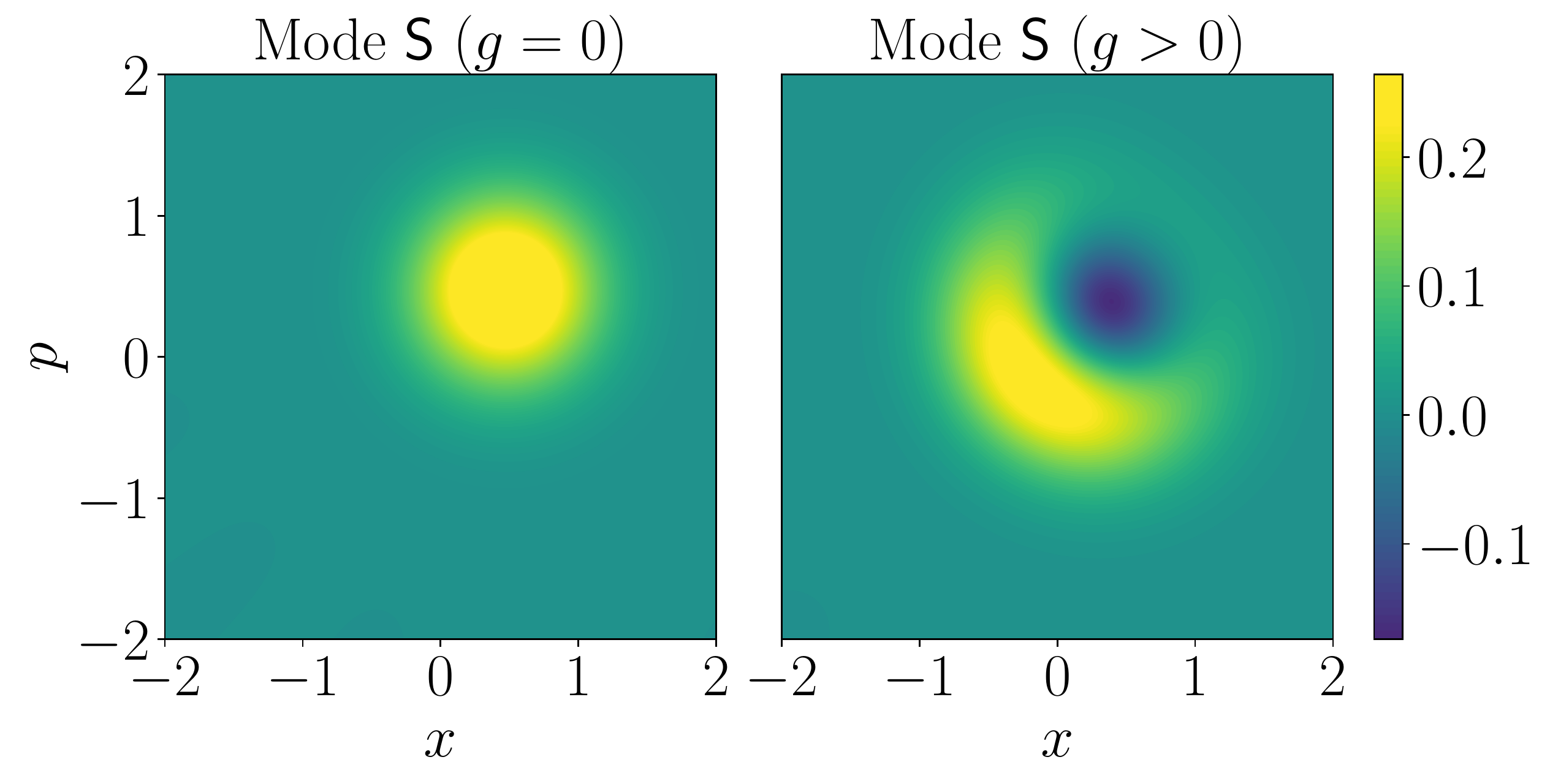}
	\caption{\textbf{The Wigner function of the mode $\ms S$ at time} $t = \tau$. The left panel shows the evolved state of the cavity in the case of no-coupling $(g = 0)$. The right panel shows the evolved state in the case of non-zero coupling $(g > 0)$. Non-classicality is clearly generated, while the process is (correlated-) catalytic. The evolution was computed with parameters $\alpha = (1+i)/2$, $\gamma = 2 \pi$, $\tau = 6$ and $g = 0.05 \gamma$.
	}\label{fig:2}
\end{figure}

The strength of interaction between the cavity and the atom is specified by the coupling constant $g$ of the interaction Hamiltonian $\hat{H}_{\text{int}}$. From now on we will focus on the state of the cavity and analyse two different scenarios: \emph{(a)} when there is no interaction between the cavity and the atom $(g = 0)$, and \emph{(b)} when the systems interact, i.e. $g \neq 0$. We will be interested in investigating the non-classical properties of the state of the electromagnetic cavity at time $t = \tau$. For that we will analyse the Wigner function $W(x, p)$ of the mode $\ms{S}$.

Consider the case when $g = 0$. In this scenario $U(t) = U_{\ms{S}}(t) \ot U_{\ms{C}}(t)$ and both systems evolve independently. The Wigner function of the cavity remains positive at all times, because $U_{\ms S}(t)$ is a Gaussian unitary, which cannot produce Wigner negativity \cite{Weedbrook2012}.

When  $g \neq 0$, the field $\ms{S}$ and the atom ${\ms C}$ exchange energy over time. 
By construction, after time $t = \tau$, the atom returns to its initial state, $\omega_{\ms C}(\tau) = \omega_{\ms C}(0)$. However, the same is not true for the state of the electromagnetic field, which in general can evolve into a state with a negative Wigner function, as shown in Fig. \ref{fig:2}.  
Importantly, this is achieved by starting with a coherent state for which the Wigner function is everywhere positive, while the atom returns exactly to its initial state. This means that it cannot be responsible for delivering any resources (such as energy) to the cavity. 
If we could decouple the atom from the field and have it interact with another matching mode in the same initial coherent state for time $\tau$, then the second mode would end up in the same state with the same (negative) Wigner function. Importantly, in this procedure the second mode will generally end up correlated with the first one.

\subsubsection{Continuous variable quantum computation}
A related effect to the one in the preceding subsection has been previously discussed in Refs.~\cite{lau_universal_2016,lau_universal_2017} in the context of quantum computation. Suppose one has access to a qubit $\ms C$ (e.g., an atom) and is able to 
\begin{enumerate}[label = \alph*)]
	\item implement coherent rotations of the form $\exp(\mathrm i \theta X_{\ms C})$, 
	\item implement a controlled interaction of the qubit with each of a set optical modes $a_i$ in the form $U_i=\exp(\mathrm i (\pi/2) (\mathbb{1}_{\ms C}-Z_{\ms C}) a_i^\dagger a_i)$, and 
	\item implement passive linear mode transformations, i.e. beam-splitters and phase-shifters. 
\end{enumerate}
Here, $X$ and $Z$ refer to the corresponding Pauli matrices.
The authors of Refs.~\cite{lau_universal_2016,lau_universal_2017} show that, under these conditions, it is possible to implement universal quantum computation using the qubit $\ms{C}$ in a strictly catalytic way. More specifically, up to initialization and final measurements, the qubit remains in the state $\ket{+}_{\ms C}$ throughout the computation. Remarkably, this scheme is possible for a large variety of encodings of the logical states of computation into the quantum states of the optical modes without changing the implemented operations.  See Section \ref{subsec:computation} for a further discussion on catalysis in the context of computation.

\subsubsection{Gaussian thermal operations}
\label{subsubsec:CV}
Gaussian systems, due to their experimental feasibility, provide an interesting platform for investigating thermodynamic properties of quantum systems. Applications of the Gaussian toolkit to thermodynamic problems involve, among many others, work extraction  \cite{brown2016passivity,singh2019quantum,Francica2020}, battery charging \cite{friis2018precision}, optimizing thermodynamic protocols \cite{mehboudi2022thermodynamic} or quantifying the information gain from a continuously monitored system \cite{belenchia2020entropy}.

In this section we discuss the thermodynamic analogue of catalysis within the Gaussian setting. Sometimes it's reasonable to impose additional limitations on the allowed Gaussian operations., e.g., when one is interested in the thermodynamic cost of implementing specific operations.  In this line, Ref.~\cite{serafini2020gaussian,narasimhachar2021thermodynamic} introduced a class of quantum channels they termed Gaussian thermal operations (GTOs). These channels form a subset of thermal operations (see Sec. \ref{subsubsec:athermality}), restricted to the case when the Hamiltonians $\hat{H}_{\ms S}$ and $\hat{H}_{\ms B}$ are quadratic and the systems $\ms S$ and $\ms B$ interact via a Gaussian energy-preserving unitary.    
The authors of Ref.~\cite{serafini2020gaussian,narasimhachar2021thermodynamic} formulated GTOs in terms of maps acting on covariance matrices, which is a standard representation in the Gaussian regime, and determined the necessary and sufficient conditions for the existance of a GTO state transformation between Gaussian states. 
Subsequently, Ref.~\cite{Yadin2021} investigated two types of catalysis within the class of GTOs: strict and correlated catalysis, which Ref.~\cite{Yadin2021} refers to as strong and weak catalysis respectively. They found that strict catalysis provides no advantage over the non-catalytic GTOs, mirroring the no-go results for LOSR entanglement (see Sec.~\ref{subsubsec:LOSR}), contextuality (see Sec. ~\ref{subsubsec:contextuality}) and the resource theory of asymmetry (see Sec.~\ref{subsec:no-broadcasting}). More specifically, they showed that all possible state transitions under strictly catalytic GTOs on the mode $\ms{S}$ can be described as thermalisations towards the bath mode $\ms{B}$, which can be easily achieved without the help of any catalyst mode $\ms{C}$. On the other hand, when the system and the catalyst form a single-mode Gaussian state, the authors of Ref.~\cite{Yadin2021} found that a correlated catalyst can significantly enlarge the set of achievable states. 
More specifically, they discuss the power of correlating catalysis in terms of their ability to concentrate thermodynamic resources in a subset of modes. Without a catalyst, or with a strict catalyst, the Gaussian resources in each mode need to be monotonic under GTOs, leading to severe restrictions. However, a correlating catalyst allows for moving of thermodynamic resources from one mode to another, under the condition that the \emph{total} thermodynamic resource across different modes is monotonic.
Finally, for the case of multiple modes they found explicit necessary conditions on state transformations, which they subsequently expressed using the majorization relation.


\subsection{Reversible computation and quantum computation}
\label{subsec:computation}
We now turn to the topic of computation and discuss catalysis in the context of (classical) reversible computation and fault-tolerant quantum computation. {In a seminal work \cite{landauer1961irreversibility}, Landauer showed that reversible (one-to-one) logical operations, such as \texttt{NOT}, can be performed without dissipating heat. On the other hand, irreversible (many-to-one) operations, such as erasure, always dissipate heat proportional to the number of bits of information lost} (see also Section~\ref{subsub:dissipation}). This naturally leads to the question of whether thermodynamically irreversible processes are needed to perform arbitrary computations. 

Remarkably, it turns out that any classical computation (i.e. one that can be realized by a Turing machine) can be performed in a logically reversible manner. {This is possible by either saving the entire history of the process in ancillary systems \cite{Bennett1973}, or by embedding the irreversible mapping into a more complex (and reversible) mapping \cite{fredkin1982conservative}}. Therefore, in principle, any classical computation can be performed with arbitrarily little heat dissipation. 

In order to implement a logically irreversible function \mbox{$f: \{0,1\}^n\rightarrow \{0,1\}$} taking $n$ bits as input, one makes use of $m$ 
auxiliary bits initialized to $0$. Then, one constructs a reversible (invertible) function $R_f:\{0,1\}^{n+m} \rightarrow \{0,1\}^{n+m}$ such that
for any $x\in\{0,1\}^n$
\begin{align}\label{eq:rev_map}
	R_f(x,0,\ldots,0) = f(x),g(x)
\end{align}
where $,$ denotes concatenation of bit-strings and $g(x)$ corresponds to some "garbage" bit string on $n+m-1$ bits. 
It turns out that it is always possible to construct such a reversible function $R_f$ if sufficiently many auxiliary bits are available.  
To see this, note that a classical bit can always be reversibly copied to an empty register via the \texttt{CNOT} gate. Therefore, by appending an additional bit to the string in Eq. (\ref{eq:rev_map}), one can implement $f$ reversibly by $(i)$ copying the result of the computation $f(x)$ into the appended bit, and $(ii)$ undoing $R_f$ on the first $n+m$ bits. This leads to the process
\begin{align}
	x,0\cdots 0,0 &\xlongrightarrow{R_f} f(x),g(x),0 \xlongrightarrow{\text{\texttt{CNOT}} } f(x),g(x),f(x)\\
	&\xlongrightarrow{R_f^{-1}} x,0\cdots 0,f(x)
\end{align}
The above approach is known as \emph{uncomputation} \cite{Bennett1997}.
Instead of a string of $0$'s, the auxiliary bits could be initialized in any other fixed state $y$ (but the function $R_f$ depends on $y$). 
From the resource theoretic perspective, we can view reversible computations as a restricted class of (free) operations $\mc O_{\mathrm{Rev}}$ and the $m$ auxiliary bits in a fixed state $y$ as a catalyst. 
In this context, the catalyst implements the reversible catalytic transformation
\begin{align}
	x,0 \AC{\mc O_{\mathrm{Rev}}}{} x,f(x).
\end{align}
If the input $x$ is not deterministic 
then the auxiliary bits remain perfectly catalytic and uncorrelated to the remaining bits. However, in this case the data-bit {(corresponding to the output of $f$)} and the original input register become correlated.
{The amount of correlations between a random input and the output of a function $f$ can be seen as a measure for how irreversible $f$ is: For a reversible function the input can be reconstructed from the output, hence the input and output are perfectly correlated.  Conversely, for the constant function the output is independent of the input.}

We now turn our attention to quantum computation {in the circuit model} \cite{nielsen2002quantum}. Since {this} is a reversible model for computation, 
{most of our previous remarks on reversible computation are applicable to quantum computation as well.} However, due to the no-cloning theorem, quantum data cannot be simply copied, which leads to subtleties regarding uncomputation \cite{Bennett1997}.

Even more interesting effects related to catalysis appear due to the need for quantum error correction: Since quantum computation is inherently fragile to errors,
it requires fault-tolerance techniques, in particular quantum error correcting codes.
The most common approach to quantum error correction uses {so-called} stabilizer codes \cite{Gottesman1996,Calderbank1997,Gottesman1997}: 
{Let $\mc P_n$ be the group of operators that can be written as a "word" of single-qubit Pauli operators on a set of $n$ qubits, possibly with an additional phase $\pm 1$ or $\pm \mathrm i$. We refer to any $P\in \mc P_n$ as a Pauli operator. 
	Consider a set of $m$ independent\footnote{{Independence means that none of the operators $S_j$ can be expressed as a product of the remaining operators $\mc S\setminus\{S_j\}$.}}, {mutually} commuting and hermitian Pauli operators $\mc S = \{S_j\}_{j=1}^m$ with $S_j\neq -\mathbb{1}$.}
Their common eigenspace {corresponding to eigenvalue} $+1$ has dimension $2^{n-m}$ and is called a \emph{stabilizer subspace}. 
The vectors in a stabilizer subspace represent the logical states of a stabilizer code which encodes $k=n-m$ qubits.  
Unitary operators that map Pauli operators to Pauli operators are called "Clifford operators", and they are generated by the set of unitary operator $\lbrace H,S, \mathrm{CNOT} \rbrace$. Here, $H$ is the Hadamard gate, $S$ is the phase gate and CNOT is the controlled-not gate. They typically correspond to unitary operations that can be implemented directly on the logical states represented by stabilizer states.  
Unfortunately, Clifford unitaries are insufficient to implement universal quantum computation. 
In fact, any {circuit} consisting of Clifford unitaries acting on stabilizer states and followed by measurements {of Pauli operators on some qubits} can be efficiently simulated on a classical computer \cite{Gottesman1998a,Aaronson2004,Veitch2012}. The  above formulation closely resembles the situation for Gaussian states discussed in Section~\ref{subsubsec:CV}, see also Refs. \cite{Gross2006,Mari2012}.
This situation naturally leads to a resource theoretic approach, where stabilizer states are the free states and Clifford unitaries combined with measurements {of Pauli operators} and classical feed-forward correspond to free operations. As can be expected, the operations such defined map stabilizer states to stabilizer states \cite{Veitch2014}.
We refer to this set of free operations simply as "Clifford operations" or "Stabilizer operations" $\mc O_{\mathrm{Stab}}$. 
Resource states, corresponding to non-stabilizer states, can be used to implement  non-Clifford unitaries such as the single-qubit $T$-gate, defined as
\begin{align}
	T = \begin{pmatrix}
		1&0\\
		0 & \mathrm{e}^{\mathrm i \pi/4}
	\end{pmatrix}
\end{align}
in the $\{\ket 0,\ket 1\}$ basis.
One way to implement a $T$-gate is via the technique known as gate teleportation (also known as state injection) \cite{Gottesman1999}. Specifically, by first preparing a suitable non-stabilizer state $\ket T$, one can then use Clifford operations to simulate a $T$-gate. The pure states required to implement a $T$-gate have been coined \emph{magic states} and can be distilled using Clifford operations from noisy magic states \cite{Knill2004,Bravyi2005}. Examples of single-qubit magic states are the eigenstates of the Clifford operator $H$, that is
	\begin{align*}
		H \ket{H_0} &= \ket{H_0}, \qquad \qquad H \ket{H_1} = - \ket{H_1},
	\end{align*}
	as well as the state $\ket{T} := T\ket+$.
	A binary vector $\bm{v} = (v_1, \ldots, v_n)$ further specifies the magic state $H_{\bm{v}} = \bigotimes_{i=1}^n \ket{H_{v_i}}$ on $n$ qubits. Several resource measures have been introduced to quantify the amount of \emph{magic} contained in a quantum state. These measures can be used to bound the classical simulation cost of quantum circuits \cite{Veitch2012,Veitch2014,Pashayan2015,Bravyi2016,Howard2017,Ahmadi2018,Bravyi2019,Wang2019,Seddon2019,Wang2020,Raussendorf2020}.
In particular, the classical simulation cost of a quantum circuit acting on stabilizer states scales exponentially with the number of $T$ gates \cite{Aaronson2004}. 
Hence, if we implement a quantum circuit using non-Clifford unitaries via gate teleportation, its classical simulation cost will scale with the number of required magic states.


It is therefore natural to ask how one can reduce the number of required resource states to implement certain non-Clifford unitaries.
Since different non-Clifford gates require different resource states, it is useful to know whether catalysis can be used to convert certain
resource states into others more efficiently. This would reduce the number of magic states that need to be distilled in order to implement a given quantum circuit. We summarize two notable results in this direction. An early result was Ref.~\cite{campbell2011catalysis} which first showed that catalysis indeed happens in the resource theory of magic states. More specifically, for a state $\ket{\psi} = (\ket{H_{000}} + \ket{H_{111}})/\sqrt{2}$ they presented a Clifford circuit implementing the strictly catalytic transformation
	\begin{align}
		\ket{\psi} \ket{H_0} \rightarrow \ket{H_0} \ket{H_0}.
	\end{align}
	Notice that $\ket{\psi}$ cannot be by itself converted into $\ket{H_0}$ via Clifford operations, see Ref. ~\cite{campbell2011catalysis} for details. They furthermore showed that strictly catalytic transformations are also possible when $\ket{\psi}$ is replaced by a mixed state. A more recent paper \cite{Gidney2018} focused on the T-count of the adder circuit $A$, which is characterized by its action on the $n$-qubit computational basis:
	\begin{equation}
		A \ket{i}\ket{j} \rightarrow \ket{i}\ket{i+j ~{\rm mod}~ 2^n},
	\end{equation}
	which is a critical subroutine in quantum Fourier transform. By using catalysis, \cite{Gidney2018} was able to introduce a circuit decomposition that reduced the number of T-counts from $8n+O(1)$ in previous constructions to $4n+O(1)$.

Building on the above investigations, Ref.~\cite{Beverland2020} provided a more systematic study of catalysis for magic states, and generalized previous constructions from Refs. \cite{Selinger2013,Gidney2018} and \cite{Gidney2019}. In particular, they observed that catalysis can be used to provide lower-bounds on the implementation cost of {general unitaries}. Furthermore, they show that for a large class of state conversions, some degree of catalysis is necessary. \cite{Beverland2020} then proceeds to characterize a broad class of magic state conversions that can be achieved via strict catalysis, and in particular demonstrate that catalysts can increase multi-copy conversion rates to and from certain magic states. 

	\begin{figure}
		\includegraphics{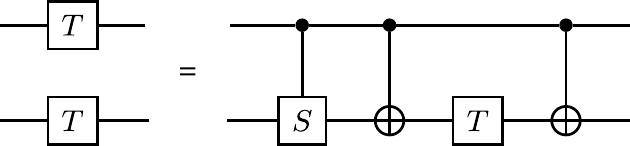}
		\caption{{\bfseries Example of magic state catalysis:} The circuit identity shows that $\ket{T}$ catalyzes the transition $\ket{CS}\rightarrow\ket{T}$.}
		\label{fig:circuit-identity}
	\end{figure}

We close this section with a further example of catalysis from \cite{Beverland2020}. 
	While the $S$-gate acting as $\ket0 \mapsto \ket 0,\ket 1\mapsto \mathrm i \ket 1$ is a Clifford gate, the controlled $S$-gate, denoted with $CS$, is not. It can be nevertheless implemented using the resource state $\ket{CS} = CS\ket+\ket+$. A no-go theorem in \cite{Beverland2020} implies that $\ket{CS}$ cannot be converted to $\ket T=T\ket+$ using Clifford operations. 
		Consider, however, the circuit equality shown in Fig.~\ref{fig:circuit-identity}:
	The circuit on the right corresponds to $\ket{CS}\ket{T}$ and Clifford operations.
	Together they yield $T\otimes T$ corresponding to two $\ket T$ states when applied to $\ket+\ket+$. 
	Therefore $\ket T$ catalyzes the otherwise impossible transition $\ket{CS}\rightarrow \ket T$: 
	A single $\ket{T}$-state together with $n$ copies of $\ket{CS}$ can be converted to $n+1$ copies of $\ket T$ using Stabilizer operations.

\section{Further topics and related areas}
\label{sec:further-topics}
In this section we discuss several topics that are related to catalyis, but either do not fit into the previous sections or do not precisely fit into the definition of catalysis that we use throughout the rest of the review.

\subsection{Catalytic transformations of  dichotomies}
\label{subsubsec:dichotomies}
A central task in statistics is to infer a parameter $\theta\in\Theta$ from experimental samples. 
The experiment can described by a $\theta$-dependent family of probability distributions $\{p_\theta\}$ that is typically assumed to be known.
In the simplest case of binary experiments, $\Theta=\{0,1\}$, the task amounts to distinguishing between two hypotheses. 
The two hypotheses are represented by a pair $(p_0,p_1)$ of probability distributions, which is often called a \emph{dichotomy}.
A dichotomy $(p_0,p_1)$ is said to be more informative than a second dichotomy $(q_0,q_1)$ if the latter arises from the former by stochastic processing \cite{blackwell1953equivalent,hardy1952inequalities,le_cam_comparison_1996,cohen1998comparisons}.   
We also say that the first dichotomy \emph{relatively majorises} the second when there exists a stochastic map $T$ such that $T p_\theta = q_\theta$~\cite{hardy1952inequalities}.

In quantum mechanics, the objects to be compared are \emph{quantum dichotomies}, denoted by $(\rho, \sigma)$ for two density operators $\rho$ and $\sigma$. 
{We say that the dichotomy $(\rho_0, \rho_1)$ is more informative than $(\sigma_0, \sigma_1)$ if there exists a quantum channel $\mathcal{E}$ such that $\sigma_\theta = \mathcal{E}(\rho_\theta)$ and say they are equivalent if there also exists a (quantum) channel $\mc R$ such that $\rho_\theta = \mc R(\sigma_\theta)$ \cite{petz_sufficient_1986,petz_sufficiency_1988,Ohya1993,shmaya2005comparison,jencova_sufficiency_2006,chefles2009quantum}.}
Importantly, when the two density operators forming a quantum dichotomy commute, they can be simultaneously diagonalised, and thus can be treated classically. 
When the two dichotomies do not commute, the inference task becomes genuinely quantum, and is notoriously harder to characterise. 

Some authors considered a generalization of the problem of comparing dichotomies by using a dichotomy that acts as a catalyst. More specifically, Ref.~\cite{rethinasamy_relative_2020} introduced the notion of catalytic transformations between dichotomies of probability vectors (which they termed catalytic relative majorization). 
In this task, the catalyst consists of a pair $(\vec{c}_0, \vec{c}_1)$ of probability vectors $\vec{c}_\theta$. 
The catalyst is then used in conjunction with the original input pair $(\vec{p}_0, \vec{p}_1)$ to generate the output $(\vec{q}_0,\vec{q}_1)$ by means of a stochastic matrix $T$ acting as
\begin{align}
	\label{eq:cat_stoch_proc_dich}
	T (\vec{p}_\theta \ot \vec{c}_\theta) = \vec{q}_\theta \ot \vec{c}_\theta. 
\end{align}
In the language of this review, this corresponds to strict catalysis. However, in complete analogy to Eq. (\ref{eq:cat_stoch_proc_dich}), one can consider here other types of catalysis, e.g. correlated catalysis. 
While the physical interpretation of such a catalytic transformation is less clear, the problem is well-defined and interesting from the mathematical perspective.

The main result of Ref.~\cite{brandao_second_2015}, discussed in Sec.~\ref{subsub:CTOs}, can be formally translated into the language of dichotomies by interpreting thermal state as one of the probability distributions in a dichotomy, thereby providing necessary and sufficient conditions for (arbitrarily) strict catalytic transformations of dichotomies in terms of R\'enyi divergences. 
\cite{Mu_2021} independently obtained these results for generic dichotomies.  
Ref.~\cite{farooq2023asymptotic} generalized the results from dichotomies to the case where $\Theta$ is a general finite set and allowing for an arbitrarily small error on the system (but not the catalyst).  
Ref.~\cite{rethinasamy_relative_2020} obtained the necessary and sufficient conditions for transforming dichotomies via correlated catalysis, showing that the Kullback-Leibler divergence is the essentially unique relevant monotone.  

The problem of (catalytically) transforming generic quantum dichotomies is more difficult. 
Ref.~\cite{shiraishi_quantum_2021} showed that whenever $D(\rho_0\|\rho_1) \geq D(\sigma_0\|\sigma_1)$ there is a correlated catalytic state transition from $(\rho_0,\rho_1)$ to $(\sigma_0^\epsilon,\sigma_1)$ and when allowing for an arbitrarily small error on system (but not catalyst).

One can also ask whether catalysis is relevant for the equivalence of dichotomies. For example, could it be true that $(\rho_0,\rho_1)$ and $(\sigma_0,\sigma_1)$ are not equivalent, while $(\rho_0\otimes \omega_0,\rho_1\otimes \omega_1)$ is equivalent to $(\sigma_0\otimes \omega_0,\sigma_1\otimes \omega_1)$? Ref. \cite{galke_sufficiency_2023} showed that this form of non-trivial strict catalysis does not occur for classical dichotomies (or commuting quantum dichotomies) and conjectured that the same holds true in the quantum case. 
In particular, this shows that there is no catalysis for inter-convertibility of incoherent quantum states via thermal operations (cf. Sec.~\ref{subsub:CTOs}), mirroring the situation for pure state LOCC (see Sec.~\ref{subsubsec:locc-proper}).

\subsection{Continuity of extensive quantities and embezzlement}
\label{subsubsec:non-continuity}
We have emphasized throughout that embezzlement may be seen as arising from the fact that trace-distance can change arbitrarily little while a resource measure changes by some finite amount. 
In other words, it is a consequence of the lack of continuity of resource measures. 
More generally, we now discuss how approximate catalysis can be used in a simple way to give bounds on how continuous \emph{extensive} resource measures can be.
Let $\mc D(\star)$ denote the set of density matrices of any finite dimension, i.e. $\mc D(\star) = \bigcup_{d} \mc D(\mathcal{H}_d)$. Given an arbitrary function $f$ from density matrices to real numbers, we define

\begin{definition}[Asymptotic continuity \cite{Donald_2002}] Let $f: \mc D(\star) \rightarrow \mathbb{R}_+$. We say that $f$ is asymptotically continous if there exist a (Lipshitz) constant $K$ such that, for any $d$ and $\rho, \sigma \in \mc D(\mathcal{H}_d)$, 
	\begin{align}
		\label{eq:asym_cont_def}
		|f(\rho) - f(\sigma)| \leq K \norm{\rho-\sigma}_1 \log d + \eta(\norm{\rho-\sigma}_1),
	\end{align}
	where $\eta$ does not depend on $d$ and satisfies $\eta(0) = 0$.
\end{definition}
A standard example of an asymptotically-continuous quantity is the von Neumann entropy. In this case Eq. (\ref{eq:asym_cont_def}) reduces to the well-known Fannes-Adenauert  inequality \cite{Audenaert_2007,Fannes1973}. In the case of von Neumann entropy, the property of asymptotic continuity often leads to asymptotic continuity of many other entanglement measures. This, on the other hand, is often a key step in deriving expressions or bounds for the asymptotic rates in quantum information processing tasks, see e.g. \cite{Terhal_2002,Horodecki_1998} for examples.

Following \cite{Leung2019}, the above definition of continuity can be generalized to capture the maximal change in the function $f$ with respect to the dimension of the density operator that takes the role of its argument. Let us consider the following definition.

\begin{definition}[``More than asymptotic'' continuity] 
	\label{def:masymp_cont}
	Let $f: \mc D(\star) \rightarrow \mathbb{R}_+$. We say that $f$ is more than asymptotically continous if there exist $\alpha < 1$ and a (Lipshitz) constant $K$ such that, for any $d$ and $\rho, \sigma \in \mc D(\mathcal{H}_d)$, 
	\begin{align}
		\label{eq:more_than_asym_cont_def}
		|f(\rho) - f(\sigma)| \leq K \norm{\rho-\sigma}_1 (\log d)^{\alpha} + \eta(\norm{\rho-\sigma}_1),
	\end{align}
	where $\eta$ does not depend on $d$ and satisfies $\eta(0) = 0$.
\end{definition}

Let us now characterize three properties of reasonable resource measures. We say that $f$ is non-constant if there exist two density matrices $\rho$ and $\sigma$ such that $f(\rho) \neq f(\sigma)$. For tensor product Hilbert spaces corresponding to composite quantum systems, a function $f$ is said to be permutationally-invariant if it does not depend on the particular labelling of the subsystems. 
More precisely, for any $n \in \mathbb{N}$ and any density operator $\rho$ on $\ms R = \ms R_1 \ot \ms R_2 \cdots \ms R_n$ and any permutation $\pi \in \mathcal{S}_n$, it holds that $f(U_{\pi} \rho U_{\pi}') = f(\rho)$, where $U_{\pi}$ is the unitary that permutes registers $\ms R_i$ according to permutation $\pi$. 
Third, a function $f$ is additive over tensor products if $f(\rho\otimes \sigma) = f(\rho) + f(\sigma)$ holds for any density matrices $\rho$ and $\sigma$. 
These three are very natural properties of resource measures. 
For example, the total magnetization number of a collection of identical spins and many entanglement measures (e.g. entanglement entropy) satisfy them naturally. 
In this context, Ref.~\cite{Leung2019} proved that any function $f$ that is non-constant, permutation-invariant and additive can never be more than asymptotically continuous according to Def. \ref{def:masymp_cont}. 
This is a direct consequence of the general construction in Section~\ref{subsubsec:approx-cat-construction}, substantiating the relationship between embezzlement and asymptotic continuity. 

Let us briefly summarize the proof of \cite{Leung2019}.
For two density matrices $\rho_{\ms S},\sigma_{\ms S}$ such that $c = |f(\rho_{\ms S})-f(\sigma_{\ms S})|>0$, 
consider the approximate catalyst states $\omega_{\ms C}$ and $\omega'_{\ms C}$ on $n-1$ copies of $\ms S$ from section~\ref{subsubsec:approx-cat-construction}. They fulfill $\norm{\omega_{\ms C} -\omega'_{\ms C}}_1 \leq \frac{2}{n-1}$ and
\begin{align}
	U_{\pi} (\rho_{\ms S} \ot \omega_{\ms C}) U_{\pi}^{\dagger} = \sigma_{\ms S} \ot \omega'_{\ms C},
\end{align}
where $U_{\pi}$ is the unitary implementing a cylic permutation of all sub-systems. 
The permutation invariance property of $f$ implies that $f(\rho_{\ms S} \ot \omega_{\ms C}) = f(\sigma_{\ms S} \ot \omega'_{\ms C})$. Since $f$ is furthermore additive, we have
\begin{align}
	c = |f(\rho_{\ms S}) - f(\sigma_{\ms S})| = |f(\omega'_{\ms C}) - f(\omega_{\ms C})|.
\end{align}
Let us now assume, by contradiction, that $f$ is more than asymptotically continuous. 
Then, there exist a constant $K$ and $\alpha < 1$ such that
\begin{align}
	|f(\omega'_{\ms C}) - f(\omega_{\ms C})| &\leq K \norm{\omega_{\ms C}-\omega'_{\ms C}}_1 (n \log d)^{\alpha} + \eta(\norm{\omega_{\ms C}-\omega'_{\ms C}}_1) \nonumber\\
	&\leq 2K \frac{n^{\alpha}}{n-1} (\log d)^{\alpha}.
\end{align} 
Since $\alpha < 1$ by assumption, the RHS tends to zero as $n$ goes to infinity, which contradicts our initial assumption $c>0$.

\subsection{Quantum channel catalysis}
Throughout we have focussed on catalysis on the level of quantum states. 
But one can also consider catalysis on the level of quantum channels instead of quantum states:
Suppose we can implement some channel $\mc E$ on some system $\ms C$ that is not a free operation: $\mc E\notin \mc O$. 
Then there may be free operations $\mc F^{(\text{enc})}_{\ms S\ms C\rightarrow \ms S'\ms C}, \mc F_{\ms S'},$ and $\mc F^{(\text{dec})}_{\ms S'\ms C\rightarrow \ms S\ms C}$ such that
\begin{align}\label{eq:channel-catalysis}
	\mc F^{(\text{dec})}_{\ms S'\ms C\rightarrow \ms S\ms C}\circ(\mc F_{\ms S'}\otimes \mc E_{\ms C})\circ \mc F^{(\text{enc})}_{\ms S\ms C\rightarrow \ms S'\ms C} = \mc E_{\ms S}\otimes \mc E_{\ms C},
\end{align}
and where $\mc E_{\ms S}$ is not a free operation. Here, we allowed the \emph{encoding} map $\mc F^{(\text{enc})}_{\ms S\ms C\rightarrow \ms S'\ms C}$ and the \emph{decoding} map $\mc F^{(\text{dec})}_{\ms S'\ms C\rightarrow \ms S\ms C}$ to map between different systems $\ms S'\ms C$ and $\ms S\ms C$.
Then the quantum channel $\mc E_{\ms C}$ catalyzes the non-free quantum channel $\mc E_{\ms S}$.
A simple example is given when $\ms S'=\ms S$, $\mc E_{\ms C}$ simply prepares a fixed state $\omega_{\ms C}$, $\mc F_{\ms S}$ prepares a fixed free state $\rho_{\ms S}$ and the encoding map is trivial.
Then \eqref{eq:channel-catalysis} reduces to $\mc F[\rho_{\ms S}\otimes\omega_{\ms C}] = \sigma_{\ms S}\otimes\omega_{\ms C}$ with $\sigma_{\ms S} = \mc E_{\ms S}[\rho_{\ms S}]$, i.e. we recover strict catalysis for quantum state transitions.

An example of similar behaviour is discussed in \cite{Brun2006,Brun2014} in the context of quantum error correction: An ideal quantum channel $\mc E_{\ms C}=\mathrm{Id}_{\ms C}$ can be used to catalyze a noisy quantum channel into an ideal quantum channel using an entanglement-assisted quantum error correcting code.
One should note however that the (unassisted) quantum capacity of a noisy channel together with a noiseless channel is simply the sum of the two \cite{Bennet1996_mixedstate}. 
Therefore this form of channel catalysis cannot be used to effectively increase the rate of the noisy channel beyond its usual capacity.  
Systematic studies of quantum channel catalysis in various resource theories provide an interesting opportunity for further research.  

\subsection{Catalytic decoupling and resource erasure}
It is well known that it is impossible to construct a quantum channel that achieves perfect cloning of unknown quantum states on a given system $\ms S$ (see also Sec.~\ref{subsec:no-broadcasting}).
Similarly, it is impossible to construct a quantum channel $\mc E$ on a bipartite sytems $\ms S_1\ms S_2$ that removes the correlations for arbitrary quantum states $\rho_{\ms S_1\ms S_2}$ in the sense that
\begin{align}
	\mc E[\rho_{\ms S_1\ms S_2}] = \rho_{\ms S_1}\otimes\rho_{\ms S_2}\quad \forall \rho_{\ms S_1\ms S_2}.\quad \text{(impossible)}
\end{align}
\emph{Decoupling} refers to a process where one tries to remove the correlations between a system $\ms A$ from a system $\ms E$.
Decoupling protocols play an important role in quantum information theory \cite{groisman_quantum_2005,Horodecki2005partial,Hayden2008,Abeyesinghe2009,Dupuis2010,Dupuis2014,majenz_catalytic_2017,berta_conditional_2018,Li2021reliability} and are closely related to quantum state merging protocols (see Sec.~\ref{subsubsec:state-merging}). 
The amount of resources one has to invest to achieve decoupling can be seen as a measure for the correlations between $\ms A$ and $\ms E$ \cite{groisman_quantum_2005}.
One approach to decoupling is to try to unitarily concentrate all correlations between $\ms A$ and $\ms E$ into a sub-system $\ms A_2$ of $\ms A$, by way of a unitary $U_{\ms A}$ acting only on $\ms A$, and then 
to trace out $\ms A_2$:
\begin{align}
	\tr_{\ms A_2}[U_{\ms A}\rho_{\ms A\ms E} U_{\ms A}^\dagger]\approx \omega_{\ms A_1}\otimes\omega_{\ms E}.
\end{align}
The Hilbert space dimension of the subsystem $\ms A_2$ required to achieve the above with a given precision can then be interpreted as a measure for the amount of correlations between $\ms A$ and $\ms E$ \cite{Dupuis2010}.
In \cite{majenz_catalytic_2017}, the authors defined \emph{catalytic} decoupling by introducing an additional system $\ms C$ in a state $\sigma_{\ms C}$ that is initially uncorrelated to $\ms A\ms E$. 
In other words, we have $\ms A\ms C = \ms A_1\ms A_2$ and 
\begin{align}
	\tr_{\ms A_2}[U_{\ms A\ms C} \rho_{\ms A\ms E}\otimes\sigma_{\ms C} U_{\ms A\ms C}^\dagger]\approx \omega_{\ms A_1}\otimes\omega_{\ms E}.
\end{align}
Naturally, the Hilbert space dimension $d_{\ms A_2}$ of the subsystem $\ms A_2$ that one has to trace out is at most as large as in the case of standard decoupling.
In fact the authors show that $\log(d_{\ms A_2})$ precisely gives an operational meaning to the \emph{smooth max-mutual information}. 
Moreover, they show that one can choose $A_1 = A C_1$, i.e., $A_1$ corresponds to the input system $A$ and part of the catalyst $\ms C$, and $\omega_{\ms A} = \tr_{\ms C_1}[\omega_{\ms A_1}]=\rho_{\ms A}$. 
That is, the marginal on $\ms A$ does not change. 
Note however that the full catalyst itself is not returned in the same state. Additionally, the part $\ms C_2$ that is traced out must in general build up strong correlations to $\ms E$.
In particular, the catalyst may not immediately be reused to help decouple a further system from $\ms E$.

A task closely related to decoupling is the task of erasing resources from a quantum state using random unitary operations from the set of free operations in a resource theory. 
Specifically, one may then ask for the minimal amount of classical randomness that is required to return a resourceful state close to the set of free states. 
If one further allows for an (approximate, correlated) catalyst in a free state during the process, then the \emph{smooth max-relative entropy} with respect to the free states quantifies the amount of randomness required to achieve resource erasure \cite{berta_disentanglement_2018,anshu_quantifying_2018}.  

\section{Acknowledgements}
We thank Nicolas Brunner, Lauritz van Luijk, Martin Plenio, Hoi-Kwan Lau and Hiroyasu Tajima for useful comments on an earlier version of this paper. We also thank countless colleagues for discussions on the topics of this review over the last years -- mentioning all of them individually is unfortunately beyond the scope of this paragraph. 

H.W. acknowledges support by the Deutsche Forschungsgemeinschaft (DFG, German Research Foundation) through SFB 1227 (DQ-mat), Quantum Valley Lower Saxony, and under Germany's Excellence Strategy EXC-2123 QuantumFrontiers 390837967. N.N. is supported by the start-up grant for Nanyang Assistant Professorship of Nanyang Technological University, Singapore. PLB acknowledges the Swiss National Science Foundation for financial support through the NCCR SwissMAP.



\bibliographystyle{apsrmp4-2}

\bibliography{refs}
\appendix

\end{document}